\documentclass[12pt,a4paper]{article}
\usepackage[utf8]{inputenc}
\usepackage{amsmath, amssymb, amsthm, mathtools, tikz}
\usetikzlibrary{calc, decorations.markings, patterns.meta}
\usepackage{upgreek}
\usepackage{hyperref}
\numberwithin{equation}{section}
\usepackage{enumitem}

\makeatletter

\newcommand{\authorinfo}[3]{%
  \g@addto@macro\@authorinfo{%
    \par\medskip
    {\scshape #1}\par
    \begingroup
      \leftskip=2em
      \parindent=0pt
      #2\par
      \texttt{#3}\par
    \endgroup
  }%
}

\newcommand{\@authorinfo}{}

\newcommand{\printaddresses}{%
  \par\vspace{2\baselineskip}
  {\footnotesize\@authorinfo}
}

\makeatother

\usepackage[maxnames=10, maxalphanames=5, minalphanames=4, style=alphabetic]{biblatex}
\DeclareFieldFormat[article]{title}{\mkbibemph{#1\isdot}}
\DeclareFieldFormat{journaltitle}{#1\isdot}
\renewbibmacro{in:}{%
  \ifentrytype{article}
    {}
    {\printtext{\bibstring{in}\intitlepunct}}}

\newcommand{\ri}{{\mathrm{i}}}
\newcommand{\rd}{{\mathrm{d}}}
\newcommand{\SL}{\mathrm{SL}_2(\mathbb{C})}
\newcommand{\RSmp}{\mathbb{CP}^1\backslash\{t_i\}_{i=1}^n}

\newtheorem{theorem}{Theorem}[section]
\newtheorem{prop}[theorem]{Proposition}
\newtheorem{lemma}[theorem]{Lemma}
\newtheorem{corr}[theorem]{Corollary}
\newtheorem{definition}[theorem]{Definition}
\newtheorem{remark}[theorem]{Remark}
\newtheorem{example}[theorem]{Example}

\title{$(1,k)$ CFT and RH problem with the $c=-2$ case}

\date{}

\author{Mikhail Bershtein \and Andrei Grigorev \and Anton Shchechkin}

\authorinfo
    {Mikhail Bershtein}
    {Scuola Internazionale Superiore di Studi Avanzati (SISSA), Trieste, Italy
    
    Istituto Nazionale di Fisica Nucleare (INFN), Section of Trieste, Italy

    Institute for Geometry and Physics (IGAP), Trieste, Italy}
    {mbersht@sissa.it} 

\authorinfo
  {Andrei Grigorev}
  {Skolkovo Institute of Science and Technology, 121205, Moscow, Russia

   National Research University Higher School of Economics, Faculty of Mathematics, Moscow, Russia}
  {\nolinkurl{andrei.al.grigorev@gmail.com}}

\authorinfo
    {Anton Shchechkin}
    {Scuola Internazionale Superiore di Studi Avanzati (SISSA), Trieste, Italy
    
    Istituto Nazionale di Fisica Nucleare (INFN), Section of Trieste, Italy
    
    Institute for Geometry and Physics (IGAP), Trieste, Italy

     National Research University Higher School of Economics, Faculty of Mathematics, Moscow, Russia}
    {shch145@gmail.com}

\begin{document}

\maketitle

\begin{abstract}
    Following approach of Iorgov--Lisovyy--Teschner,  we construct solutions of the (modified) Riemann--Hilbert problem using conformal blocks of $(1,k)$ Virasoro models. 
    For $k>1$ case, the solution of this Riemann--Hilbert problem is not unique due to more singular behavior at punctures. 
    On the CFT side the dimension of the space of conformal blocks also increases.

    We specifically study the $k=2$ case, which corresponds to the central charge $c=-2$ and symplectic fermions. 
    We explicitly construct a corresponding solution of the modified Riemann--Hilbert problem in the case of 3 punctures and prove its uniqueness under suitable initial data conditions. 
    We also obtain new bilinear relations for $c=-2$ tau functions.
\end{abstract}

\begin{quote}
\tableofcontents
\end{quote}

\section{Introduction}
\label{sec:intro}
\paragraph{Context.} The subject of the paper is the Isomonodromy/CFT correspondence -- relation between local systems and conformal field theory (CFT). 
The seminal statement in this correspondence is the Gamayun--Iorgov--Lisovyy \cite{GIL12} conjecture (proven in \cite{ILT14}, \cite{Bershtein:2015bilinear}, \cite{Gavrylenko:2018fredholm}) of the CFT formula for the tau function of the Painlev\'e VI equation. 
This tau function is an isomonodromic tau function for the \(\mathrm{SL}_2\) local system on the sphere with 4 punctures. 
The corresponding CFT formula is a series of Virasoro (that is \(W\) algebra for \(\mathfrak{sl}_2\)) 4-point conformal blocks on the sphere. 
This relation has many extensions and connections in mathematics and physics, see e.g. references above and also \cite{Gavrylenko:2015isomonodromic}, \cite{Jimbo:2017qPainleve},  \cite{Bonelli:2017fermi}, \cite{Bonelli:2017painleveGauge}, \cite{Bershtein:2018cluster}, \cite{Bonelli:2020torus}, \cite{Nekrasov:2024blowups}, \cite{Jeong:2020riemann}. 

A remarkable approach to the Isomonodromy/CFT correspondence was proposed in \cite{ILT14}. 
The idea is to construct a solution of the Riemann-Hilbert (RH) problem with a given monodromy using insertions of degenerate vertex operators (\(\phi_{1,2}\) in CFT notation), cf. the Verlinde loop operators \cite{AGGTV09}, \cite{DGOT09}. 
Winding of the degenerate operators around generic Virasoro primary operators is expressed through fusion and braiding matrices (in a Moore-Seiberg formalism \cite{MS88}) and shifts of momenta proportional to \(\mathsf{b}\). 
Here we use the (Wick-rotated) central charge parametrization \(c=1-6(\mathsf{b}-\mathsf{b}^{-1})^2\). 
It appears that for \(\mathsf{b}^2=1\) (i.e. \(c=1\)) the shift operators commute with fusion and braiding of the degenerate field. 
This allows one to construct a solution of the RH problem. 

The commutativity mentioned above also holds for \(\mathsf{b}^2\in \mathbb{Z}\). 
Hence one can construct a solution of the (modified) RH problem using Virasoro conformal blocks with degenerate fields for such central charges. 
From now on we set
\begin{equation}\label{eq:intro c(k)}
    \mathsf{b}^2=k,\;\;   k\in \mathbb{Z}_{> 0},\qquad  c=1-6\frac{(k-1)^2}k.
\end{equation}
This \(c\) is the central charge of the \((1,k)\) Virasoro minimal model.

The series of Virasoro conformal blocks appearing in the tau functions and solutions of the RH problem can be viewed as conformal blocks of a CFT with an extended symmetry. 
The relevant extensions for central charges~\eqref{eq:intro c(k)} are triplet algebras \cite{FGST06}, the \(\mathfrak{sl}_2\) case of Feigin--Tipunin algebras \(\mathcal{FT}_k(\mathfrak{sl}_2)\)~\cite{FT10}. 
These algebras have a Lie group \(\mathrm{SL}_2\) of symmetries. 
This allows one to define conformal blocks of \(\mathcal{FT}_k(\mathfrak{sl}_2)\) depending on a local system. 
One can view solutions of the (modified) RH problem and corresponding tau functions (for any \(k\in \mathbb{Z}_{>0}\)) as examples of such conformal blocks. 
This remarkable proposal is due to J.~Teschner, see more in the lecture \cite{Teschner22PIRSA}\footnote{When the paper was almost ready, we learned about related project~\cite{FLT26}, which has some overlap with our results.}.

The \(\mathcal{FT}_k(\mathfrak{g})\) CFTs have important connections with quantum groups \(U_q(\mathfrak{g})\) and quantum cluster algebras at roots of unity and certain 3d TQFTs. 
Such relations were probably the main motivation for our work. 
We refer to \cite{CDGG21} and \cite{Teschner22PIRSA} for the original detailed exposition.

\paragraph{Results of the paper.}
In this paper we restrict ourselves to the \(\mathrm{SL}_2\) local systems on the Isomonodromy side and Virasoro symmetry on the CFT side. 
In this case the primary vertex operator \(V^{\theta_2}_{\theta_3,\theta_1}\colon \mathbb{M}_{\theta_1, c} \rightarrow \mathbb{M}_{\theta_3, c}\) is unique up to normalization. 
We define the vertex operators between sums of the Virasoro modules over the lattice. 
\begin{equation}
    \bar{V}^{\theta_2}_{[\theta_3]_k, [\theta_1]_k}(t)\colon \bigoplus_{n_1\in \mathbb{Z}}\mathbb{L}_{\theta_1{+}n_1k} \to \bigoplus_{n_3\in \mathbb{Z}}\mathbb{L}_{\theta_3{+}n_3k},
    \qquad \bar{V}^{\theta_2}_{[\theta_3]_k, [\theta_1]_k}(t)\Big|_{\mathbb{L}_{\theta_1{+}n_1k}\to \mathbb{L}_{\theta_3{+}n_3k}}\coloneqq V^{\theta_2}_{\theta_3{+}n_3k,\theta_1{+}n_1k}.
\end{equation}
These sums should be viewed as twisted modules over the triplet algebra \(\mathcal{FT}_k(\mathfrak{sl}_2)\).
We show that in the appropriate normalization the operators \(\bar{V}\) are closed under fusion with the degenerate \(\phi_{1,2}\) fields. 
This is essentially a manifestation of the commutativity between fusion and the lattice shifts for \(\mathsf{b}^2\in \mathbb{Z}\) mentioned above. 
We call the operators \(\bar{V}\) \emph{periodic vertex operators}.
One can view them as vertex operators for \(\mathcal{FT}_k(\mathfrak{sl}_2)\) twisted by a local system.
For \(k=1\), the periodic vertex operators appeared in \cite{GM16} in fermionic language.

The matrix elements of a solution of the modified RH problem are conformal blocks of periodic vertex operators at punctures \(t_i,\, {\scriptstyle i=1,\ldots, n}\) with two degenerate fields \(\psi_{\pm}\) at points \(z\) and \(z_0\) (normalization point):
\begin{gather}\label{eq:intro:Phi}
    \Big(\Phi\big((\theta_i)_{i=1}^n;(\sigma_i,s_i)_{i=2}^{n-2}\big|(t_i)_{i=1}^n\big)\Big)_{\varsigma_0\varsigma}(z,z_0)\sim 
    \left\langle \psi_{-\varsigma_0}(z_0) \psi_{\varsigma}(z)\left(\prod\nolimits_{j=n}^1 s_j^{-{a_0}/{\sqrt{k}}}\,\bar{V}_{[\sigma_j]_k,[\sigma_{j-1}]_k}^{\theta_j}(t_j)\right)\right\rangle.
\end{gather}
Here and below in the Introduction \(\sim\) represents proportionality up to factors that are not important for the main idea. 
These conformal blocks depend on intermediate momenta that we denote by \(\sigma_j,\, {\scriptstyle j=2,\ldots, n-2}\) and we also introduce dual coordinates \(s_j\).
The fusion of periodic vertex operators with degenerate fields mentioned above implies that \(\Phi\) has constant monodromy with respect to loops of \(z\). 
Furthermore, the \(\sigma_j,s_j\) are the  Fenchel--Nielsen coordinates of the corresponding local system. 
See Theorem~\ref{thm:cm2_linear_solve} for the precise statement; note that this theorem is conditional on the existence of conformal blocks \eqref{eq:intro:Phi} as analytic functions with prescribed properties. 

We emphasize two differences between \(k>1\) and \(k=1\) cases. 
First, the monodromy is invariant under shifts of \(\sigma\)'s by the lattice \(\mathbb{Z}\), while in the definition of the periodic vertex operators only the lattice \(k\mathbb{Z}\) was used. 
In CFT terms this is the difference between shifts by \(\mathsf{b}^{-1}\) and by \(\mathsf{b}\).
Hence, for a given monodromy we construct a \(k^{n-3}\)-parameter family of solutions of the RH problem. 
This corresponds to the growth of the dimension of the space of conformal blocks for \(\mathcal{FT}_k(\mathfrak{sl}_2)\) (for odd \(k\)).
See Corollary~\ref{cor:generic solution}.

Second, the corresponding RH problem has more solutions for \(k>1\).
This is due to a different, ``more singular'' behavior of \(\Phi(z)\) at punctures \(t_i,\, {\scriptstyle i=1,\ldots, n}\).
The new parameters of the solution can be seen in terms of the connection matrix \(A(z)=\Phi'(z)\Phi^{-1}(z)\).  
This matrix has poles at \(t_i,\, {\scriptstyle i=1,\ldots, n}\) but also additional poles \(w_j, \, {\scriptstyle j=1,\ldots, n(k-1)}\) without monodromy of \(\Phi(z)\) around them. 
These new poles and residues at them can be viewed as additional parameters of the \emph{modified RH problem}. 
These apparent (i.e. monodromy-free) singularities resemble the ones in the description of Bethe ansatz for the Gaudin model, see e.g. \cite{Frenkel:2003opers}.

We define \(k\)-tau functions as conformal blocks of the periodic vertex operators 
\begin{equation}
    \uptau\big((\theta_i)_{i=1}^n;(\sigma_i,s_i)_{i=2}^{n-2}\big|(t_i)_{i=1}^n\big)
    \sim 
    \left\langle \left(\prod\nolimits_{j=n}^1 s_j^{-{a_0}/{\sqrt{k}}}\,\bar{V}_{[\sigma_j]_k,[\sigma_{j-1}]_k}^{\theta_j}(t_j)\right)\right\rangle
\end{equation}
This is a natural generalization of the \(k=1\) formula \cite{GIL12}, \cite{ILT14} and the \(k=2\) definition~\cite{Bershtein:2019painleve}. 
Moreover, the space of tau functions for a given local system is also of dimension \(k^{n-3}\); a generic representative has the form 
\begin{equation}
    \uptau\big(\Lambda;(\theta_i)_{i=1}^n;(\sigma_i,s_i)_{i=2}^{n-2}\big|(t_i)_{i=1}^n\big)= 
    \sum_{0\le l_2,\dots, l_{n-2}\le k-1} \Lambda_{l_2,\dots,l_{n-2}}
    \uptau\big((\theta_i)_{i=1}^n;(\sigma_i+l_i,s_i)_{i=2}^{n-2}\big|(t_i)_{i=1}^n\big)
\end{equation}

The equations on tau functions usually have a bilinear form. 
In studying them we restrict ourselves to the case \(k=2\) (which corresponds to central charge \(c=-2\)), while the methods in principle work for any \(k>1\). 

The idea is to consider \(\det \Phi\). 
Since the monodromy of \(\Phi\) belongs to \(\mathrm{SL}_2\), the determinant is a rational function. 
The poles of \(\det \Phi\) are at the punctures \(t_i,\, {\scriptstyle i=1,\ldots, n}\) and zeroes are at the \(w_j, \, {\scriptstyle j=1,\ldots, n(k-1)}\) (apparent singularities of the connection matrix). 
It appears that this rational function has a CFT construction in the tensor square \(\mathcal{FT}_k(\mathfrak{sl}_2)^{\otimes 2}\):
\begin{gather}
    \label{eq:intro:det Phi}
    \det \Phi(z,z_0) \sim 
    \left\langle I(z_0) I(z) \left( \prod\nolimits_{j=n}^1 \left( s_j^{-{a_0}/{\sqrt{k}}}
    \bar{V}_{[\sigma_j]_k,[\sigma_{j-1}]_k}^{\theta_j}(t_j) \right)^{\otimes 2} \right)\right\rangle 
    \\
    \text{where}\qquad I(z)=\psi_+(z)\otimes \psi_-(z) - \psi_-(z)\otimes \psi_+(z).
\end{gather}
The OPE with the field \(I(z)\) determines singular behavior of \(\det \Phi\) near \(t_j\). 
This determines the rational function \eqref{eq:intro:det Phi} (assuming the existence of conformal blocks as analytic functions, as above).
Asymptotic behavior of this function in the limit \(z,z_0\to \infty\) gives bilinear equations on tau functions. 
We present them (for \(k=2\)) in Theorem~\ref{thm:tau_identities}. 
We also present (for \(k=2\)) a formula for the logarithmic derivative of the tau function in terms of the connection matrix, see Prop.~\ref{prop:tau_isom_def}. 
This can be compared with the definition of isomonodromic tau function through its logarithmic derivative.

\paragraph{Discussion.} 
This paper could be viewed as a step towards understanding the \((1,k)\) version of the Isomonodromy/CFT correspondence. 
A lot of basic questions remain unanswered.

The modified RH problem appears here from the CFT side.
It would also be interesting to motivate it from the isomonodromy side.
A related question is an isomonodromic definition of the space of \(k\)-tau functions for general \(k\).
Note that the term \(k\)-tau function in our paper is not based on integrable-system properties; rather, it is used by analogy with the previously studied \(k=1\) case \cite{GIL12} and \(k=2\) case \cite{Bershtein:2019painleve}, as well as with the quantum tau functions \cite{Bonelli:2025refined} discussed below (see also the terminology and motivation in \cite{Teschner22PIRSA}).
Such integrable-system properties of \(k\)-tau functions could take the form of generalizations of the \(k=2\) logarithmic derivative formula or the bilinear equations found in this paper.

Another important question is the relation to the BPS/CFT interpretations of the Isomonodromy/CFT correspondence \cite{Nekrasov:2024blowups}, \cite{Jeong:2020riemann} (based on previous works \cite{Nekrasov:2016qqCharacters}, \cite{Nekrasov:2017bpsCftV}, \cite{Nekrasov:2019bpsCftIV}). 
Perhaps, the apparent singularities \(w_j\) could be interpreted as Hecke modifications, which, in the BPS/CFT framework, are realized through the fusion of surface defects in \cite{Jeong:2026parallel}. 
Notice also that the method of \cite{Nekrasov:2024blowups}, \cite{Jeong:2020riemann} was based on the blowup relation between classical (Nekrasov--Shatashvili limit) functions and quantum (e.g. \(c=1\)) functions (see also \cite{Bershtein:2025highestWeight} for the CFT formulation). 
To understand \((1,k)\)-tau functions in this way it is natural to consider instanton counting on the Hirzebruch surface \cite{Bruzzo:2011poincare} in the presence of the surface defect.

There is different approach to the Isomonodromy/CFT correspondence through blowup relations, initiated in \cite{Bershtein:2015bilinear}. 
It appears that it can also be applied to \(k=2\) (i.e. \(c=-2\)); this was studied in \cite{Bershtein:2019painleve}. 
It appears that relations from \cite{Bershtein:2019painleve} are different from those in Theorem~\ref{thm:tau_identities}. 
It would be interesting to understand the relationship between these two sets of relations.

There are three types of bilinear equations in Theorem~\ref{thm:tau_identities}: algebraic, first-order and second-order. 
Perhaps the first-order equation can be viewed as a version of the KZ equation. 
The algebraic and second-order ones can be compared with algebraic and Toda equations, respectively, in Painlev\'e theory.
We believe that there are similar identities for tau functions when $k>2$. 
But we do not have a conjecture on their form and their possible applications. 

Finally, note that recently the Isomonodromy/CFT correspondence was deformed to arbitrary central charges on the CFT side and to \emph{quantum} Painlev\'e equations on the Isomonodromy side, see \cite{Bershtein:2018cluster}, \cite{Gavrylenko:2020irregular}, \cite{Bonelli:2025refined}. 
It would be interesting to understand the connection between this and the results of the paper. 

\paragraph{Plan of the paper.} 
In Section~\ref{sec:local_RH problem} we recall basic notions on \(\SL\) local systems on \(\mathbb{CP}^1\) with \(n\) punctures,  Fenchel--Nielsen coordinates, and convenient RH problems. Section~\ref{sec:CFT} serves the same goal for Virasoro representations, vertex operators and conformal blocks.

After these preparations we construct solutions of the modified RH problem from CFT in Section~\ref{sec:sol}. 
An important technical point here is the normalization of vertex operators. 
We use the Wick-rotated Ponsot--Teschner normalization \cite{PT99} with an additional sign factor that makes fusion matrices periodic (not just quasi-periodic). 
We also redefine the normalization of degenerate fields basically using a free field realization. 
Thereby, the matrix-valued function \(\Phi\) defined by formula~\eqref{eq:intro:Phi} has a constant \(\mathrm{SL}_2\) monodromy. 

The proof of Theorem~\ref{thm:cm2_linear_solve} is based on this fact. 
The \(k\)-modified RH problem (see Definition~\ref{def:RH_modif}) differs from the standard one by the asymptotic behavior at punctures \(t_i\). 
Due to this \(\det\Phi\) has poles at punctures \(t_i\) of order \(k-1\) and has additional zeroes (apparent singularities) at \(w_j\) mentioned above.

This is illustrated for \(k=2\) and \(n=3\). 
We present an explicit solution of the modified RH problem in terms of hypergeometric functions, see Prop.~\ref{prop:explicite_cb}. 
We also show that imposing additional normalization conditions, one can prove uniqueness of the solution of the \(2\)-modified RH problem for \(n=3\), see Prop.~\ref{prop:gRH problem_unique}. 
Together with Wick rules for symplectic fermions proven in Section~\ref{sec:fermions}, this in principle determines the periodic vertex operators in terms of local system data, similarly to the \(k=1\) case in \cite{GM16}.

In Section~\ref{sec:fermions} we continue the study of the \(k=2\) case. 
In this case the triplet algebra can be embedded into the doublet algebra, which is the symplectic fermion algebra. 
Also OPEs with degenerate fields in this case have relatively mild singularities (usually the order of poles is linear in \(k\)). 
This allows us to prove the main results of the section: bilinear relations (Theorem~\ref{thm:tau_identities}), the logarithmic derivative of the tau function (Prop.~\ref{prop:tau_isom_def}), and Wick rules for symplectic fermions (Prop.~\ref{prop:Wick SF}). 
Some computations are moved to the appendix.

\paragraph{Acknowledgements.} 
We thank I.~Chaban, T.~Creutzig, T.~Dimofte, S.~Jeong, P.~Gavrylenko, O.~Lisovyy, A.~Litvinov, A.~Marshakov, V.~Poberezhny, I.~Sechin, J.~Teschner, M.~Yamazaki for stimulating discussions.

M.B. is very grateful to Kavli IPMU and the University of Edinburgh for hospitality and excellent working conditions. 
A.G. is grateful to SISSA and BIMSA for hospitality.
The work of A.G. has been supported by the Russian Science Foundation under grant 26-11-00342.

\section{Local systems and Riemann-Hilbert problem}
\label{sec:local_RH problem}

In this section we review the notion of the $\SL$ local systems on the punctured sphere, the corresponding Riemann-Hilbert problem, and the isomonodromic deformations. For more details see \cite{JMU81}.

\subsection{$\SL$ local systems}
\label{ssec:local}

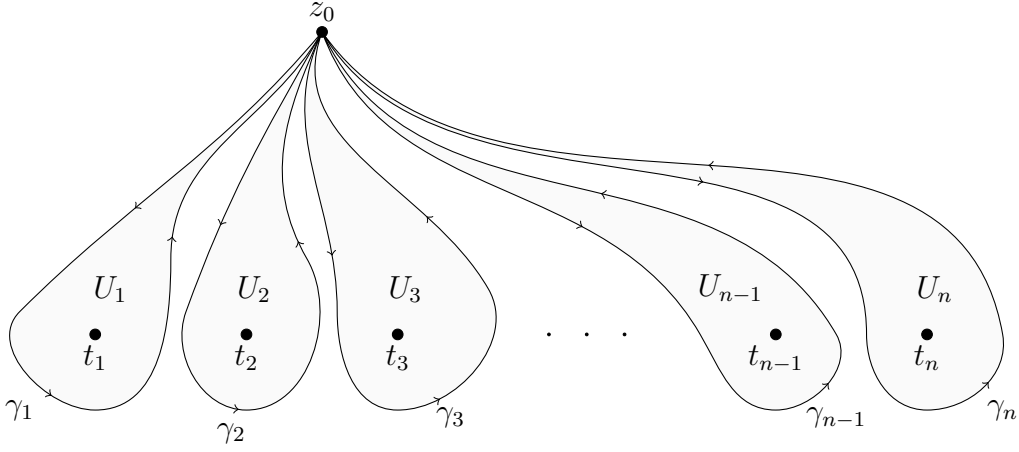
\begin{figure}[h]
    \centering
    \begin{tikzpicture}[decoration={
    markings,
    mark=at position 0.25 with {\arrow{>}},
    mark=at position 0.5 with {\arrow{>}},
    mark=at position 0.75 with {\arrow{>}}}]
     \coordinate (t1)  at (0,0);
     \coordinate (t2)  at (2,0);
     \coordinate (t3)  at (4,0);
     \coordinate (tm1) at (9,0);
     \coordinate (tm)  at (11,0);
     \coordinate (z0)  at (3,4); 

    \filldraw (6,0) circle (0.5pt) (6.5,0) circle (0.5pt) (7,0) circle (0.5pt);
    
    \draw[postaction=decorate,fill=gray!4] (z0) to[out=-130,in=45] ($(t1)-(1,-0.3)$)
              to[out=-135,in=180] ($(t1)+(0,-1)$)
              to[out=0,in=-90] ($(t1)+(1,1)$)
              to[out=90,in=-120] (z0);
    
    \draw[postaction=decorate,fill=gray!4] (z0) to[out=-120,in=70] ($(t2)-(0.8,-0.3)$)
              to[out=-110,in=180] ($(t2)+(0,-1)$)
              to[out=0,in=-60] ($(t2)+(0.8,1)$)
              to[out=120,in=-115] (z0);
    
    \draw[postaction=decorate,fill=gray!4] (z0) to[out=-115,in=90] ($(t3)-(0.8,-0.3)$)
              to[out=-90,in=180] ($(t3)+(0,-1)$)
              to[out=0,in=-60] ($(t3)+(1.2,0.6)$)
              to[out=120,in=-110] (z0);

    \draw[postaction=decorate,fill=gray!4] (z0) to[out=-70,in=120] ($(tm1)-(1,0)$)
              to[out=-60,in=180] ($(tm1)+(0,-1)$)
              to[out=0,in=-60] ($(tm1)+(0.8,0)$)
              to[out=120,in=-65] (z0);

    \draw[postaction=decorate,fill=gray!4] (z0) to[out=-60,in=90] ($(tm)-(0.8,0)$)
              to[out=-90,in=180] ($(tm)+(0,-1)$)
              to[out=0,in=-80] ($(tm)+(1,0)$)
              to[out=100,in=-55] (z0);

        \filldraw (t1) circle (2pt) node[anchor=north] {$t_1$};
    \filldraw (t2) circle (2pt) node[anchor=north] {$t_2$};
    \filldraw (t3) circle (2pt) node[anchor=north] {$t_3$};
    \filldraw (tm1) circle (2pt) node[anchor=north] {$t_{n{-}1}$};
    \filldraw (tm) circle (2pt) node[anchor=north] {$t_n$};  
    \filldraw (z0) circle (2pt) node[anchor=south] {$z_0$};

       \node at ($(t1)+(0.2,0.6)$) {$U_1$};
       \node at ($(t2)+(0.1,0.6)$) {$U_2$};
       \node at ($(t3)+(0.1,0.6)$) {$U_3$};
       \node at ($(tm1)+(-0.6,0.6)$) {$U_{n-1}$};
       \node at ($(tm)+(0.1,0.6)$) {$U_n$};

    \node at ($(t1)-(1,1)$) {$\gamma_1$};
    \node at ($(t2)-(0.2,1.3)$) {$\gamma_2$};
    \node at ($(t3)-(-0.7,1.1)$) {$\gamma_3$};
    \node at ($(tm1)+(0.8,-1.1)$) {$\gamma_{n-1}$};
    \node at ($(tm)+(1,-1)$) {$\gamma_n$};
\end{tikzpicture}
    
    \caption{Fundamental group generators $\gamma_i$}
    \label{fig:fundamental_group}
\end{figure}

\paragraph{Definition.} Consider the Riemann sphere with $n\geq3$ punctures $\RSmp$. Its fundamental group with respect to a base point $z_0$ has the standard description implying a fixed ordering of punctures $t_i$:
\begin{equation}\label{fundamental_group}
    \pi_1\big(\RSmp,z_0\big) = \Big\langle\big([\gamma_i]\big)_{i= 1}^n\Big|[\gamma_1][\gamma_2] \ldots [\gamma_n] = 1\Big\rangle,
\end{equation}
where $\gamma_i$'s are loops with the winding number $\delta_{ij}$ around puncture $t_j$, see Fig.~\ref{fig:fundamental_group}. Linear representations of the fundamental group of a path-connected topological space up to their equivalence are called \textit{local systems}. In particular, in this paper we consider \textit{anti-}representations with values in $\SL$.
\begin{definition}
        Consider anti-representations $\rho:\pi_1\big(\RSmp,z_0\big)\to\SL$
        with the generator images $M_i\coloneqq\rho([\gamma_i])$.  
        The space of local systems $\mathcal{L}_n$ on $\RSmp$ with values in $\SL$ is given by
    \begin{equation}\label{def_local}
        \mathcal{L}_n=\Big\{\big(M_i\big)_{i=1}^n\Big|M_i\in \SL:M_n M_{n-1}\ldots M_1 = \mathbf1_{2\times 2} \Big\}\Big/\sim\,,
    \end{equation}
    where $\sim$ stands for the overall $\SL$-conjugation.
\end{definition}

\paragraph{Linear system.}  A natural source of the fundamental group representations are the monodromy groups of the Fuchsian flat connections.
Namely, consider holomorphic $\mathfrak{sl}_2$-connections on the trivial bundle over $\RSmp$ of the form
\begin{equation}\label{A}
    A(z) = \sum_{i=1}^n\frac{A_i}{z-t_i}\,\,\mathrm{d} z, \qquad  A_i\in  \mathfrak{sl}_2(\mathbb{C})
\end{equation}
and the corresponding linear system together with the initial condition
\begin{equation}\label{linear_system}
    \Phi^{\prime}(z) = A(z)\Phi(z),\qquad \Phi(z_0) = \mathbf1_{2\times2}.
\end{equation}
Then, for any $[\gamma] \in \pi_1(\RSmp,z_0)$ the monodromy matrix is defined as the analytic continuation $\Phi_{\gamma}(z_0)$ along any representative of $[\gamma]$. Thus, each choice of residues $(A_i)_{i=1}^n$ gives an anti-representation $\pi_1(\RSmp,z_0) \to \SL$, called the monodromy group. Moving to the local systems, one considers $\SL$-conjugacy classes of the monodromy matrices; this removes the dependence on any chosen linear problem initial data, in particular, on the base point $z_0$.

\paragraph{Local behavior.}The monodromy matrix along each elementary loop $\gamma_i$ can be obtained from the local solution of the linear system \eqref{linear_system} in a small neighborhood of the corresponding puncture $t_i$. Such solution is described by the standard statement of the analytic theory of differential equations:
\begin{prop}\label{prop:Phi_puncture}
    For the residue $A_i$ with eigenvalues $\pm\theta_i, 2\theta_i\notin\mathbb{Z}$ there exists a punctured neighborhood $U_{t_i}^\times$ of $t_i$, where the solution of the initial value problem \eqref{linear_system} can be expressed in the form
    \begin{equation}\label{Phi_puncture}
        \Phi(z) = H_i(z)   \begin{pmatrix}(z-t_i)^{\theta_i} & 0\\ 0 & (z-t_i)^{-\theta_i}\end{pmatrix}C_i.
    \end{equation}
    Here $C_i\in\SL$ is a constant matrix, while $H_i(z)$ is a $\mathrm{GL}_2(\mathbb{C})$-valued holomorphic function in the non-punctured $U_{t_i}$ such that $A_i=H_i(0)\mathrm{diag}(\theta_i,-\theta_i)H_i^{-1}(0)$.
\end{prop}
\noindent This immediately yields the monodromy matrices $M_i$ around the generating loops $\gamma_i$:
\begin{equation}\label{Monodromy_puncture}
M_i=\Phi_{\gamma_i}(z_0)=C_i^{-1}\begin{pmatrix}e^{2\pi\ri\theta_i} & 0\\ 0 & e^{-2\pi\ri\theta_i}\end{pmatrix}C_i.
\end{equation}
Conversely, generic $M_i$'s from \eqref{def_local} can be represented in this form. Note that $C_i$'s are determined up to the left multiplication by matrices from the diagonal matrix subgroup $T=\big\{\mathrm{diag}(\lambda,\lambda^{-1})|\lambda\in \mathbb{C}^\times\big\}$.

\paragraph{$\boldsymbol\theta$-leaf of $\mathcal{L}_n$.}The space of local systems $\mathcal{L}_n$ can be foliated into the level sets of the matrix invariant $\mathrm{Tr}\, M_i=2\cos2\pi\theta_i$. Namely, denoting by $\boldsymbol{M}^{\mathrm{SL}_2}=(M_1,M_2,\ldots, M_n)^{\SL}$ a local system represented by tuple $(M_1,M_2,\ldots, M_n)$, we use
\begin{definition}
        The $\boldsymbol{\theta}$-leaf of the space of the local systems $\mathcal{L}_n$ for tuple $\boldsymbol{\theta}=(\theta_i)_{i=1}^n$ is the level set 
\begin{equation}
\mathcal{L}_n^{\boldsymbol{\theta}}=\Big\{\boldsymbol{M}^{\mathrm{SL}_2}\in\mathcal{L}_n\Big|\mathrm{Tr}\,M_i=2\cos2\pi\theta_i\Big\}.
\end{equation}
\end{definition}
Note that integer shifts and sign changes of $\theta_i$'s do not change the level set.
However, we formally distinguish the same level sets $\mathcal{L}_n^{\boldsymbol{\theta}}$ with different linear system parameters $\boldsymbol{\theta}$. 

For generic tuple $\boldsymbol{\theta}$ it is easy to see that $\dim \mathcal{L}_n^{\boldsymbol{\theta}}=2(n-3)$. Furthermore, $\mathcal{L}_3^{\boldsymbol{\theta}}$ consists of one point. 

\subsection{Fenchel--Nielsen parametrization of $\SL$ local systems}
\label{ssec:FN}

\begin{figure}[h]
    \centering
    \begin{tikzpicture}[decoration={
    markings,
    mark=at position 0.5 with {\arrow{>}}}]
    \draw[dashed] (0,0) arc (-90:90:0.2cm and 0.5cm); 
    \draw[postaction=decorate] (0,0) arc (-90:-270:0.2cm and 0.5cm);
    
    \draw[postaction=decorate] (-3,0) arc (-90:90:0.2cm and 0.5cm);
    \draw (-3,1) arc (90:270:0.2cm and 0.5cm);
    
    \draw(-1,2.2) arc (0:180:0.5cm and 0.2cm);
    \draw[postaction=decorate](-2,2.2) arc (180:360:0.5cm and 0.2cm);

    \draw (-3,0) .. controls (-2.5,0.2) and (-0.5,0.2) .. (0,0);
    \draw (-3,1) .. controls (-2,1) .. (-2,2.2);
    \draw (0,1) .. controls (-1,1) .. (-1,2.2);

    \filldraw (-3,0.5) circle (2pt) node[left=4pt] {$t_1$};
    \filldraw (-1.5,2.2) circle (2pt) node[above=4pt] {$t_2$};
    \node at (-0.5,0.5) {$c_2$};
    \node at (-2.5,0.5) {$c_{t_1}$};
    \node at (-1.5,1.7) {$c_{t_2}$};

    \begin{scope}[xshift=3cm]
    \draw[dashed] (0,0) arc (-90:90:0.2cm and 0.5cm); 
    \draw[postaction=decorate] (0,0) arc (-90:-270:0.2cm and 0.5cm);
    
    \draw(-1,2.2) arc (0:180:0.5cm and 0.2cm);
    \draw[postaction=decorate](-2,2.2) arc (180:360:0.5cm and 0.2cm);
    
    \draw (-3,0) .. controls (-2.5,0.2) and (-0.5,0.2) .. (0,0);
    \draw (-3,1) .. controls (-2,1) .. (-2,2.2);
    \draw (0,1) .. controls (-1,1) .. (-1,2.2);

    \filldraw (-1.5,2.2) circle (2pt) node[above=4pt] {$t_3$};
    \node at (-0.5,0.5) {$c_3$};
    \node at (-1.5,1.7) {$c_{t_3}$};
    
    \draw (0,1) .. controls (0.8,1) .. (1,1.1);
    \draw (0,0) .. controls (0.2,0.1) .. (1,0.1);
    \end{scope}

     \filldraw (4.5,0.5) circle (0.5pt) (5,0.5) circle (0.5pt) (5.5,0.5) circle (0.5pt);
    
    \begin{scope}[xshift=10cm]
    \draw (-3,1) .. controls (-3.8,1) .. (-4,1.1);
    \draw (-3,0) .. controls (-3.2,0.1) .. (-4,0.1);
    
    \draw[dashed] (-3,0) arc (-90:90:0.2cm and 0.5cm); 
    \draw[postaction=decorate] (-3,0) arc (-90:-270:0.2cm and 0.5cm);
    
    \draw(0,0) arc (-90:90:0.2cm and 0.5cm);
    \draw[postaction=decorate](0,1) arc (90:270:0.2cm and 0.5cm);
    
    \draw(-1,2.2) arc (0:180:0.5cm and 0.2cm);
    \draw[postaction=decorate](-2,2.2) arc (180:360:0.5cm and 0.2cm);
    
    \draw (-3,0) .. controls (-2.5,0.2) and (-0.5,0.2) .. (0,0);
    \draw (-3,1) .. controls (-2,1) .. (-2,2.2);
    \draw (0,1) .. controls (-1,1) .. (-1,2.2);

    \node at (-2.3,0.5) {$c_{n-2}$}; 
    \filldraw (-1.5,2.2) circle (2pt) node[above=4pt] {$t_{n-1}$};
    \filldraw (0,0.5) circle (2pt) node[right=4pt] {$t_n$};
    \node at (-1.5,1.7) {$c_{t_{n-1}}$};
    \node at (-0.5,0.5) {$c_{t_n}$};
    \end{scope}
    
    \end{tikzpicture}

    \caption{Pants decomposition of the Riemann sphere with $n$ holes}
    \label{fig:pants}
\end{figure}
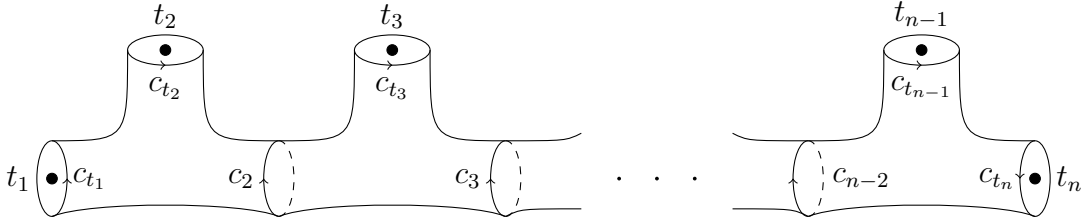

\paragraph{Pants decomposition.}
Let us remove from $\mathbb{CP}^1$ not just punctures $\{t_i\}_{i=1}^n$ but small open disks (holes) containing them, bounded by (oriented) cycles $c_{t_i}$. Then we can consider a pants decomposition of the sphere with $n$ ordered holes: we cut it by (oriented) cycles $c_i,\,{\scriptstyle i=2,\ldots,n-2}$ into the spheres with three holes as shown in Fig.~\ref{fig:pants}. Each cycle $c_i$, if considered as a loop with a base point $z_0\in c_i$, represents product $[\gamma_1][\gamma_2]\ldots [\gamma_i]$; the corresponding monodromy matrix is $M_iM_{i-1}\ldots M_1\eqqcolon M_{[i:1]}$. Therefore, the pants decomposition induces a map $\mathcal{L}_n\to \mathcal{L}_3^{\times(n-2)}$ acting by
\begin{multline}
(M_1,\ldots,M_n)^{\mathrm{SL}_2}\mapsto\Big(\big(M_1,M_2,M_{[2:1]}^{-1}\big)^{\mathrm{SL}_2}, \big(M_{[2:1]},M_3,M_{[3:1]}^{-1}\big)^{\mathrm{SL}_2}, \big(M_{[3:1]},M_4,M_{[4:1]}^{-1}\big)^{\mathrm{SL}_2},\ldots,\\ (M_{[n-3:1]},M_{n-2},M^{-1}_{[n-2:1]})^{\mathrm{SL}_2},(M_{[n-2:1]},M_{n-1},M_n)^{\mathrm{SL}_2}\Big).
\end{multline}
Additionally foliating $\mathcal{L}_n^{\boldsymbol{\theta}}$ into the level sets of the partial product traces
\begin{equation}
\mathcal{L}_n^{\boldsymbol{\theta};\boldsymbol{\sigma}}=\Big\{\boldsymbol{M}^{\mathrm{SL}_2}\in\mathcal{L}_n^{\boldsymbol{\theta}}\Big|\mathrm{Tr}M_{[i:1]} = 2\cos 2\pi\sigma_i,\quad i=2,\ldots,n{-}2\Big\},
\end{equation}
we can restrict the above map to
\begin{equation}
\phi\,:\,\mathcal{L}_n^{\boldsymbol{\theta};\boldsymbol{\sigma}} \to \mathcal{L}_3^{(\theta_1,\theta_2,\sigma_2)}\times \mathcal{L}_3^{(\sigma_2,\theta_3,\sigma_3)}\times \dots \times \mathcal{L}_3^{(\sigma_{n-2},\theta_{n-1},\theta_n)}.
\end{equation}

\paragraph{Decomposition of $\mathcal{L}_n^{\boldsymbol{\theta};\boldsymbol{\sigma}}$.} For generic tuples $\boldsymbol{\theta}$ and $\boldsymbol{\sigma}$ the codomain of map $\phi$ is just a point, while $\dim\mathcal{L}_n^{\boldsymbol\theta;\boldsymbol\sigma} = n-3$. Let us fix a generic local system $\check{\boldsymbol{M}}^{\mathrm{SL}_2}\in\mathcal{L}_n^{\boldsymbol\theta;\boldsymbol\sigma}$ together with its representative~$\check{\boldsymbol{M}}=(\check{M}_1,\ldots,\check{M}_n)$. Then a  representative of an arbitrary element of $\mathcal{L}_n^{\boldsymbol\theta;\boldsymbol\sigma}$ has the form $(S_i^{-1}\check{M}_iS_i)_{i=1}^n\eqqcolon \boldsymbol{M}_{\boldsymbol{S}}$ for some matrices $S_i\in\SL, {\scriptstyle i=1,\ldots, n}$. Any tuple of such form is a representative of an element in $\mathcal{L}_n^{\boldsymbol\theta;\boldsymbol\sigma}$ iff $\phi(\boldsymbol{M}_{\boldsymbol{S}}^{\mathrm{SL}_2})= \phi(\check{\boldsymbol{M}}^{\mathrm{SL}_2})$.
This condition restricts possible tuples $(S_i)_{i=1}^n$, namely we have isomorphism
\begin{equation}\label{S_param}
\mathcal{L}_n^{\boldsymbol{\theta};\boldsymbol{\sigma}}\cong\Big\{\big(S_i^{-1}\check{M}_iS_i\big)_{i=1}^n\Big|S_1=S_2,\,S_n= S_{n-1};\; S_{i+1}S_i^{-1}\in \mathbf{C}_{\mathrm{SL}_2}\big(\check{M}_{[i:1]}\big),\, i=2,\ldots, n{-}2\Big\}\Big/\sim,
\end{equation}
where by $\mathbf{C}_{\mathrm{SL}_2}(\cdot)$ we denote the centralizer of a group element. These centralizers act on $\mathcal{L}_n^{\boldsymbol{\theta};\boldsymbol{\sigma}}$ effectively up to the scalar sign (center $\mathbf{Z}_{\mathrm{SL}_2}$), i.e. we have 
\begin{equation}\label{centr_param}
\mathcal{L}_n^{\boldsymbol{\theta};\boldsymbol{\sigma}}\cong \prod_{i=2}^{n-2} \left(\mathbf{C}_{\mathrm{SL}_2}\big(\check{M}_{[i:1]}\big)/\mathbf{Z}_{\mathrm{SL}_2}\right), \qquad  \mathbf{Z}_{\mathrm{SL}_2}=\{\pm\mathbf{1}_{2\times2}\},
\end{equation}
where each factor is one-dimensional.

\paragraph{Parametrization of centralizers.}
To parametrize such centralizers explicitly, let us successively diagonalize partial products of given monodromy matrices by $\SL$ matrices as follows
\begin{equation}\label{partial_product_param}
M_{[i:1]}=C_{1,0}^{-1}C_{2,1}^{-1}\ldots C_{i,i-1}^{-1}\begin{pmatrix}e^{2\pi\ri\sigma_i} & 0\\ 0 & e^{-2\pi\ri\sigma_i}\end{pmatrix}C_{i,i-1}\ldots C_{2,1} C_{1,0},\quad i=1,\ldots, n{-}1,
\end{equation}
Here for convenience of the future exposition we added cases $i=1,n{-}1$ and define
\begin{equation}
    \sigma_1\coloneqq\theta_1, \quad \sigma_{n-1}\coloneqq-\theta_n, \qquad C_{1,0}\coloneqq C_1, \quad C_{n-1,n-2}\coloneqq C_n\cdot C_{1,0}^{-1}C_{2,1}^{-1}\ldots C_{n-2,n-3}^{-1}.
\end{equation}
For $2\sigma_i\notin \mathbb{Z}$, each centralizer $\mathbf{C}_{\mathrm{SL}_2}(M_{[i:1]})$ is given by the diagonal subgroup $T\subset \SL$ up to the same conjugation:
\begin{equation}\label{centralizer_diag}
\mathbf{C}_{\mathrm{SL}_2}(M_{[i:1]})=\bigg\{ C_{1,0}^{-1}C_{2,1}^{-1}\ldots C_{i,i-1}^{-1}\begin{pmatrix}
    s_{i}^{1/2}&0\\ 0 & s_i^{-1/2}
\end{pmatrix}  C_{i,i-1}\ldots C_{2,1} C_{1,0}\bigg|s_i^{1/2}\in\mathbb{C}^\times\bigg\}, \quad i=1,\ldots, n{-}1.
\end{equation}
Note that the branching ambiguity
of parameter $s_i^{1/2}$ can produce the scalar sign, thus parameters $s_i$ parametrize factor $\mathbf{C}_{\mathrm{SL}_2}(M_{[i:1]})/\mathbf{Z}_{\mathrm{SL}_2}$.

\paragraph{Fenchel-Nielsen coordinates.} Considering centralizers \eqref{centralizer_diag} for the fixed representative $\check{\boldsymbol{M}}$,
one can parametrize $\mathcal{L}_n^{\boldsymbol{\theta};\boldsymbol{\sigma}}$ given by \eqref{centr_param} with coordinates $(s_i)_{i=2}^{n-2}\in (\mathbb{C}^\times)^{n-3}$. Then the representatives $\boldsymbol{M}$ of a generic local system $\boldsymbol{M}^{\mathrm{SL}_2}\in\mathcal{L}_n^{\boldsymbol{\theta};\boldsymbol{\sigma}}$
can be written explicitly from \eqref{S_param}. Namely, for $\mathbf{C}_{\mathrm{SL}_2}\big(\check{M}_{[i:1]}\big)$ given by \eqref{centralizer_diag} with diagonalization matrices $\check{C}_{i,i-1}$, the partial products of tuple $\boldsymbol{M}$ are given by \eqref{partial_product_param} with diagonalization matrices 
\begin{equation}\label{FN_from_gluing}
C_{1,0}=S_1 \check{C}_{1,0}, \qquad C_{i,i-1}=\begin{pmatrix}
    s_i^{1/2}&0\\ 0 & s_i^{-1/2}
\end{pmatrix}\check{C}_{i,i-1}, \quad i=2,\ldots, n-2, \qquad C_{n-1,n-2}=\check{C}_{n-1,n-2}.
\end{equation}
 
The obtained coordinates $\sigma_i$ and $s_i$ for $i = 2,\ldots,n-2$ on $\mathcal{L}_n^{\boldsymbol{\theta}}$ are referred to Fenchel-Nielsen lengths and angles, respectively. Angles $s_i$ are not canonical due to the ambiguity in the choice of $\check{\boldsymbol{M}}$ and, furthermore, due to the freedom of the left multiplication of the corresponding $\check{C}_{i,i-1}$'s by matrices of the diagonal matrix subgroup $T\subset\SL$. Lengths $\sigma_i$ are defined up to integer shifts and also there are symmetries $(\sigma_i,s_i)\mapsto (-\sigma_i,s_i^{-1})$, coming from the transposition matrix action $C_{i,i-1}\mapsto \begin{pmatrix}
0 & \mspace{-10mu} 1\\[-0.1cm]
1 & \mspace{-10mu} 0
\end{pmatrix}C_{i,i-1}.$

\subsection{Riemann-Hilbert problem}
\label{ssec:RH problem}
\paragraph{Multi-valued function formulation.}
One can ask to reconstruct the linear problem~\eqref{linear_system} and its solution $\Phi(z)$ that yields a given local system. The problem of constructing a matrix function $\Phi(z)$ with a prescribed monodromy and local behavior is called \textit{Riemann-Hilbert problem}. For the considered $\SL$ case it is formally given by

\begin{definition}\label{def:RH_local}
    For a given generic local system $\boldsymbol{M}^{\mathrm{SL}_2}\in\mathcal{L}_n^{\boldsymbol{\theta}}$, the Riemann-Hilbert (RH) problem asks for a $2\times 2$ matrix holomorphic multi-valued function $\Phi(z)$ on $\RSmp$, such  
    that in the neighborhood of each $t_i$ it obeys the local behavior given by Prop.~\ref{prop:Phi_puncture} with $C_i$'s that yield \eqref{Monodromy_puncture}.
\end{definition}
The local behavior of Prop.~\ref{prop:Phi_puncture} near each puncture $t_i$ implies $\det\Phi(z)$ to be holomorphic in the neighborhood $U^\times(t_i)$. Thus, due to the Liouville's theorem we have that 
\begin{prop}\label{prop:detPhi}
The determinant $\det \Phi(z)$ of the RH solution is a nonzero constant function, in particular, $\Phi(z)$ is invertible on the whole $\RSmp$.
\end{prop}

As a consequence, we have that $A(z)=\Phi'(z)\Phi^{-1}(z)$ is given by \eqref{A}, i.e. the linear system of \eqref{linear_system} is reconstructed. Also, the solution of RH problem is essentially unique:
\begin{prop}\label{prop:RH_unique}
    For $\boldsymbol{M}^{\mathrm{SL}_2}\in\mathcal{L}_n^{\boldsymbol{\theta}}$ with generality condition
    $2\theta_i \notin \mathbb{Z},\,{\scriptstyle
    i = 1,\ldots, n}$, the RH problem solution $\Phi(z)$ is unique up to the left multiplication by a constant matrix.
\end{prop}
\noindent We can fix this constant matrix by the normalization condition $\Phi(z_0) = \mathbf1_{2\times 2}$ for a base point $z_0$.

\begin{figure}[h]
    \centering
    \begin{tikzpicture}
    \def\rn{0.3};
    \def\rp{1.3};
    \def\rs{2.3};
    \def\ru{5.3};
    
    \draw (\rn,0) -- (\rs+0.75,0) (\ru-1.25,0) -- (\ru+1,0);
    \draw (0,0) circle (\rp) circle (\rs) circle (\ru);

    \fill[even odd rule,pattern={Lines[angle=45,distance=6pt,line width=2pt]},
    pattern color=gray!10] (0,0) circle (\rs+0.75) circle (\ru-1.25);

    \draw[dash pattern=on 2pt off 8pt] (0,0) circle (\rs+0.75) circle (\ru-1.25);
    \filldraw ({(\rs+\ru)/2-0.5},0) circle (0.5pt) ({(\rs+\ru)/2-0.25},0) circle (0.5pt) ({(\rs+\ru)/2},0) circle (0.5pt);

    \filldraw (\rn,0) circle (2pt) node[anchor=south east] {$t_1$};
    \filldraw (\rp,0) circle (2pt) node[anchor=south west] {$t_2$};
    \filldraw (\rs,0) circle (2pt) node[anchor=south west] {$t_3$};
    \filldraw (\ru,0) circle (2pt) node[anchor=south west] {$t_{n-1}$};
    \filldraw (\ru+1,0) circle (2pt) node[anchor=south west] {$t_n$};

    \node at (0,-0.5) {$K_1$};
    \node at ({(\rp+\rs)/2*cos(-45)},{(\rp+\rs)/2*sin(-45)}) {$K_2$};
    \node at ({(\rs+0.4)*cos(-45)},{(\rs+0.4)*sin(-45)}) {$K_3$};
   \node at ({(\ru-0.6)*cos(-45)},{(\ru-0.6)*sin(-45)}) {$K_{n-2}$};
    \node at ({(\ru+0.8)*cos(-30)},{(\ru+0.8)*sin(-30)}) {$K_{n-1}$};
 
    \node[align=left] at ({(\rn+\rp)/2},-0.15) {\tiny $+$};
    \node[align=left] at ({(\rp+\rs)/2},-0.15) {\tiny $+$};
    \node[align=left] at (\rs+0.5,-0.15) {\tiny $+$};
    \node[align=left] at (\ru-0.5,-0.15) {\tiny $+$};
    \node[align=left] at (\ru+0.5,-0.15) {\tiny $+$};
    \end{tikzpicture}
    \caption{Dissected Riemann sphere with punctures}
    \label{fig:dissec}
\end{figure}
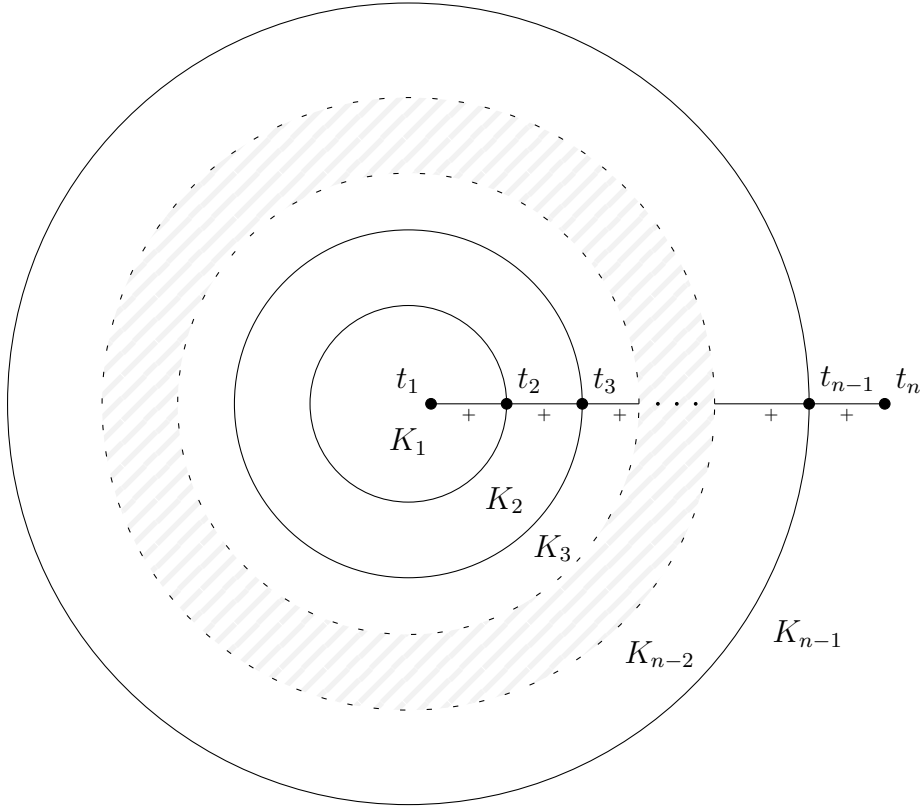
\paragraph{Single-valued formulation on  dissection.}
The other formulation of the RH problem is given in the setting of the single-valued functions on dissected $\mathbb{CP}^1$ instead of the multi-valued functions on the whole $\RSmp$. For definiteness, we order the punctures in the absolute value increasing order\footnote{The strict inequalities can be provided by a M\"obius transformation.} $0 \leq |t_1|  < |t_2| < \dots  < |t_{n-1}| < |t_n| \leq \infty$. Then $\mathbb{CP}^1$ is dissected by circles $|z|=|t_i|,{\scriptstyle i=2,\ldots n-1}$ and segments $[t_i,t_{i+1}], {\scriptstyle i=1,\ldots n-1}$ into regions 
$K_1=\{|z|<|t_2|\}$, $K_i=\{|t_i|<|z|<|t_{i+1}|\}\backslash(t_i,t_{i+1})$ for $i=2,\ldots, n{-}2$, and $K_{n-1}=\{|z|>t_{n-1}\}$ as presented on Fig.~\ref{fig:dissec}.
\begin{definition}\label{def:RH_annulus}
    For a given generic local system $\boldsymbol{M}^{\mathrm{SL}_2}\in\mathcal{L}^{\boldsymbol{\theta}}_n$ with parametrization \eqref{partial_product_param}, the Riemann-Hilbert problem asks for a set of invertible holomorphic single-valued functions $\Phi_i(z):K_i\to \mathrm{Mat}_{2\times2}(\mathbb{C})$ for $i=1,\ldots, n-1$ such that
    \begin{enumerate}
        \item \label{item:circle} On each circle cut $|z|=|t_i|$, $\scriptstyle i=2,\ldots, n-1$ the corresponding $\Phi_i(z)$ and $\Phi_{i-1}(z)$ are glued as
        $\Phi_{i-1}(z)=\Phi_i(z) C_{i,i-1}$.
        
        \item \label{item:segment} On each segment cut $(t_i,t_{i+1})$ the corresponding $\Phi_i(z)$ on the two sides of the segment is glued by 
        $\Phi_i^+(z)=\Phi_i^-(z)\begin{pmatrix}e^{2\pi\ri\sigma_i} & 0\\ 0 & e^{-2\pi\ri\sigma_i}\end{pmatrix}$,
        where the "+" side of the cut is marked on Fig.~\ref{fig:dissec}.
        
        \item \label{item:local} For each $t_i$ there exists its punctured neighborhood $U^\times(t_i)$, where the corresponding functions are given by (c.f. \eqref{Phi_puncture}): 
        \begin{equation}\label{RH problem_diss_local}
            \begin{aligned}
                U^\times(t_1)&:\Phi_1^\pm(z) =H_1(z)   \begin{pmatrix}(z-t_1)^{\theta_1} & 0\\ 0 & (z-t_1)^{-\theta_1}\end{pmatrix}C_1^{\pm}, 
                \\
                U^\times(t_i)&:\Phi_{i-1}^\pm(z)=\Phi_i^\pm(z) C_{i,i-1}=H_i(z)   \begin{pmatrix}(z-t_i)^{\theta_i} & 0\\ 0 & (z-t_i)^{-\theta_i}\end{pmatrix}C_i^{\pm}, \quad  i=2,\ldots, n-1, 
                \\
                U^\times(t_n)&:\Phi_{n-1}^\pm(z)=H_n(z)   \begin{pmatrix}(z-t_n)^{\theta_n} & 0\\ 0 & (z-t_n)^{-\theta_n}\end{pmatrix}C_n^{\pm},  
                \\ 
                \textrm{with}\quad &  C_i^+=C_i^-\begin{pmatrix}e^{2\pi\ri\sigma_{i-1}} & 0\\ 0 & e^{-2\pi\ri\sigma_{i-1}}\end{pmatrix},\quad 
                i=1,\ldots, n, \qquad \sigma_0\coloneqq \theta_1.
           \end{aligned}
        \end{equation}
    \end{enumerate}
\end{definition}

The equivalence of these two formulations of the RH problem is rather straightforward:
\begin{prop}\label{prop:local_disks}
For a given generic local system $\boldsymbol{M}^{\mathrm{SL}_2}\in\mathcal{L}_n^{\boldsymbol{\theta}}$, the normalized RH solution $\Phi(z)$ of Definition~\ref{def:RH_local} is in one-to-one correspondence with the normalized RH solution $(\Phi_i)_{i=0}^n$ of Definition~\ref{def:RH_annulus}.
\end{prop}
\begin{proof}
\ref{def:RH_local}$\Rightarrow$\ref{def:RH_annulus}.
     A solution of the RH problem $\Phi(z)$ on $\mathbb{CP}^1\backslash \bigcup_{i=1}^{n{-}1} [t_i,t_{i+1}]$ becomes a single-valued analytic function. Restricting given $\Phi(z)$ on each $K_i$, we see that the monodromy matrix $M_{[i:1]}$ is conjugated to the jump-matrix of $\Phi(z)|_{K_i}$ through the cut $(t_i,t_{i+1})$. Multiplying $\Phi(z)|_{K_i}$ by a constant matrix from the right to diagonalize this jump-matrix yields $\Phi_i(z)$. 

    \ref{def:RH_annulus}$\Rightarrow$\ref{def:RH_local}. For a given solution $(\Phi_i)_{i=1}^{n-1}$ of the RH problem define $\Phi(z)|_{K_1}\coloneqq\Phi_1(z)$ and then successively analytically continue\footnote{This actually uses the identity theorem.} $\Phi(z)$ as
    $\Phi(z)|_{K_i}=\Phi_i(z)C_{i,i{-}1}\ldots C_{2,1}$ to the whole $\mathbb{CP}^1\backslash \bigcup_{i=1}^{n{-}1} [t_i,t_{i+1}]$. Analytically continuing obtained single-valued function $\Phi(z)$ we obtain a multi-valued function $\Phi(z)$ with prescribed local behavior and thus monodromies.
\end{proof}

\paragraph{Isomonodromic deformations.}
While moving the punctures $t_1,\ldots t_n$, we can ask for the simultaneous transformations of the RH problem solution to preserve the monodromy matrices. To formulate this precisely, let us choose smooth representatives (loops) $(\gamma_i)_{i=1}^n$ of the fundamental group generators \eqref{fundamental_group} that have no mutual or self-intersections (except $z_0$), as on Fig.~\ref{fig:fundamental_group}. Each such representative $\gamma_i$ bounds a region containing $t_i$, which we denote by $U_i$. For any choice of $\tilde{t}_i\in U_i$ cycles $(\gamma_i)_{i=1}^n$ still define a system of generators of $\pi_1(\mathbb{CP}^1\backslash\{\tilde{t}_i\}_{i=1}^n,z_0)$. 
A multivariable holomorphic function  $\Phi\big(\tilde{t}_1,\ldots,\tilde{t}_n\big|z\big)$ for $\tilde{t}_i \in U_i,\, {\scriptstyle i=1,\ldots,n}$ is called the solution of \textit{isomonodromic deformation problem} if it solves the RH problem as function of $z$ for a given common local system.

For the linear system \eqref{linear_system}, isomonodromic deformations correspond to  certain transformations of the connection $A(z)$. More precisely, transformations of the residues $(A_i)_{i=1}^n$ are governed by the Schlezinger equations. A trivial example of the isomonodromic deformations is generated by algebra $\mathfrak{sl}_2$ of the M\"obius transformations that acts on the linear system. Up to that, in the simplest non-trivial case of $n=4$ punctures the Schlezinger equations are equivalent 
the Painlev\'e~VI equation.

\section{Virasoro vertex operators and conformal blocks}\label{sec:CFT}
In this section we recall some basics and fix conventions on the representation theory of the Virasoro algebra in the 2D CFT framework \cite{DMS_book}, focusing on the vertex operators and their fusion with the degenerate ones \cite{ILT14}. 

\subsection{Verma module of the Virasoro algebra and its bosonization}
\label{ssec:Virasoro}

\paragraph{Heisenberg algebra and Fock module.}
\begin{definition}\label{def:Heis}
The Heisenberg algebra $\mathrm{Heis}$ is an infinite dimensional Lie algebra defined by
\begin{equation}
   \mathrm{Heis} = \Big\langle\{a_n\}_{n\in \mathbb{Z}}, \mathbf1 \Big| [a_n, a_m] = \frac{n}2\delta_{n+m, 0}\mathbf1, \quad  [\mathbf1, a_n] = 0 \Big\rangle.
\end{equation}
\end{definition}
\begin{definition}\label{def:Fock}
The Fock module $\mathbb{F}_{\alpha}$ of weight $\alpha$ is the induced module 
\begin{equation}
\mathrm{Ind}^{\mathrm{Heis}}_{\mathrm{Heis}_{\geq0}}\mathbb{C}|\alpha\rangle=\mathbb{C}[\mathrm{Heis}_{<0}]|\alpha\rangle, \qquad \mathrm{Heis}_{\geq0}=\langle\{a_n\}_{n\in \mathbb{Z}_{\geq0}}, \mathbf1\rangle, \qquad  \mathrm{Heis}_{<0}=\langle\{a_n\}_{n\in \mathbb{Z}_{<0}}\rangle,
\end{equation}
where the subalgebra $\mathrm{Heis}_{\geq0}$ acts on the cyclic vector $|\alpha\rangle$ (defined up to a constant factor) by
\begin{equation}
    \ri a_0 |\alpha\rangle = \alpha |\alpha\rangle,\qquad
    a_n |\alpha\rangle = 0, \,\, n\in\mathbb{Z}_{>0}, \qquad \mathbf1|\alpha\rangle = |\alpha\rangle. 
\end{equation}
\end{definition}
In the 2D CFT framework, generators $\{a_n\}_{n\in \mathbb{Z}}$ constitute modes of the \textit{boson field}
\begin{equation}
\ri\varphi (z) \coloneqq -\frac12\partial_{a_0}+ a_0\log z-\sum_{n\in \mathbb{Z}\setminus\{0\}} \frac{a_n}nz^{-n}, \qquad [\partial_{a_0},a_0]=\mathbf1,  
\end{equation}
so that commutation relations between $a_n$'s turn into the boson operator-product expansion (OPE):
\begin{equation}\label{eq:OPE_phi}
\varphi(z)\varphi(w)=-\frac12\log(z-w)+\mathrm{reg}.
\end{equation}
Recall that operator-product expansions present products of two operators (fields) as a function (usually a series) in the difference of their arguments. 
Generally, only the singular part is important and is thus presented; the regular remainder is denoted by $\mathrm{reg}$ as above.

\paragraph{Virasoro algebra and its Verma module.}
\begin{definition}
The Virasoro algebra $\mathrm{Vir}$ is an infinite dimensional Lie algebra generated by $\{L_n\}_{n\in\mathbb{Z}}$ and the central generator $C$ with the commutation relations
\begin{equation}
[L_n,L_m]=(n-m)L_{n+m}+\frac{n(n^2{-}1)}{12}\delta_{n+m,0}\,C, \qquad [L_n,C]=0.
\end{equation}
\end{definition}
Note that below we consider only representations where the central element $C$ acts by a scalar $c\in\mathbb{C}$, called \textit{central charge}. Collecting the other generators into a current $T(z)\coloneqq\sum\limits_{n\in\mathbb{Z}}L_n z^{-n-2}$ turns the Virasoro commutation relations into the following OPE:
\begin{equation}
T(z)T(w)=\frac{c/2}{(z-w)^4}+\frac{2T(w)}{(z-w)^2}+\frac{\partial T(w)}{z-w}+\mathrm{reg}.     
\end{equation}
\noindent In the 2D CFT framework, the current $T(z)$ is referred to as the \textit{stress-energy tensor}.
\begin{definition}
The Virasoro Verma module $\mathbb{M}_{\Delta,c}$ of the highest weight $\Delta$ is the induced module 
\begin{equation}
\mathrm{Ind}^{\mathrm{Vir}}_{\mathrm{Vir}_{\geq0}}\mathbb{C}|\Delta\rangle=U(\mathrm{Vir}_{<0})|\Delta\rangle, \qquad \mathrm{Vir}_{\geq0}=\langle\{L_n\}_{n\in \mathbb{Z}_{\geq0}}, C\rangle, \qquad  \mathrm{Vir}_{<0}=\langle\{L_n\}_{n\in \mathbb{Z}_{<0}}\rangle,
\end{equation}
where the subalgebra $\mathrm{Vir}_{\geq0}$ acts on the cyclic vector $|\Delta\rangle$ (defined up to a constant factor) by
\begin{equation}\label{eq:hw Virasoro}
    L_0|\Delta\rangle=\Delta|\Delta\rangle, \qquad L_n|\Delta\rangle = 0, \,\, n\in\mathbb{Z}_{>0}, \qquad C|\Delta\rangle=c|\Delta\rangle.
\end{equation}
The cyclic vector $|\Delta\rangle$ is referred to as the highest weight vector of $\mathbb{M}_{\Delta,c}$.
\end{definition}
For generic $(\Delta,c)$ the Verma module is irreducible but the precise condition is naturally written in the bosonization framework, which we describe just below.

\paragraph{Bosonization and reducibility of the Verma module \cite{FF90}.} For generic $\alpha$, the Virasoro algebra action defined by the stress-energy tensor
\begin{equation}\label{T_phi}
T(z)=Q\partial^2\varphi(z)-:\big(\partial\varphi(z)\big)^2:
\end{equation}
turns the Fock module $\mathbb{F}_{\alpha}$ into the Verma module $\mathbb{M}_{\Delta_{\alpha},c}$ with the highest weight vector $|\alpha\rangle$ and
\begin{equation}\label{Liouville_param}
c=1+6Q^2, \qquad \Delta_{\alpha}=\alpha(Q-\alpha).
\end{equation}
Note that $\mathbb{M}_{\Delta_{\alpha},c}$ is obtained in this way not only from $\mathbb{F}_{\alpha}$ but also from $\mathbb{F}_{Q-\alpha}$.

We also introduce the parameter \(b\) by \(Q=b+b^{-1}\). The \emph{Kac spectrum} is given by
\begin{equation}\label{Kac_spectrum}
    \mathrm{Kac}_b=\big\{\alpha_{r,s},Q-\alpha_{r,s}\mid {r,s\in\mathbb{Z}_{\geq1}}\big\}, \qquad 
\alpha_{r,s}=\frac{Q}2-\frac{rb^{-1}+sb}2.
\end{equation}
Then the Verma module $\mathbb{M}_{\Delta_{\alpha},c}$ is irreducible iff \(\alpha\) does not belong to the Kac spectrum. 
We will call such \(\alpha\) \emph{generic}. 
We denote by \(\mathbb{L}_{\Delta,c}\) the irreducible module generated by a highest weight vector satisfying~\eqref{eq:hw Virasoro}. 
For example, for generic \(\alpha\) we have \(\mathbb{L}_{\Delta_\alpha,c} \simeq \mathbb{M}_{\Delta_{\alpha},c}\).

If \(\alpha=\alpha_{r,s}\) for some $r,s\in\mathbb{Z}_{\geq1}$ then the Virasoro Verma module $\mathbb{M}_{\Delta_{\alpha_{r,s}},c}$ is reducible. 
For \(b^2\notin\mathbb{Q}\) this module contains a single (up to a constant factor) vector $v_{r,s}$ such that
\begin{equation}
L_0v_{r,s}=(\Delta_{\alpha_{r,s}}+rs)v_{r,s}, \qquad L_m v_{r,s}=0,\quad m\in\mathbb{Z}_{>0}.
\end{equation}
This vector $v_{r,s}$ is called \textit{singular}; it generates the submodule \(U({\mathrm{Vir}_{<0}})v_{r,s}\subset\mathbb{M}_{\Delta_{\alpha_{r,s}},c}\), which is isomorphic to the Verma module \(\mathbb{M}_{\Delta_{\alpha_{r,s}}+rs,c}\). The quotient $\mathbb{M}_{\Delta_{\alpha_{r,s}},c}/\mathbb{M}_{\Delta_{\alpha_{r,s}}+rs,c}$ is irreducible and isomorphic to \(\mathbb{L}_{\Delta_{\alpha_{r,s}},c}\).

We will mainly need the case of \(b^2\in \mathbb{Z}\). 
In this case the Verma module can contain several singular vectors, but for any non-generic \(\alpha\) there exist \(r,s\) (with minimal product \(rs\)) and a singular vector \(v_{r,s}\) as above such that \(\Delta_\alpha=\Delta_{\alpha_{r,s}}\) and $\mathbb{M}_{\Delta_{\alpha_{r,s}},c}/\mathbb{M}_{\Delta_{\alpha_{r,s}}+rs,c}$ is irreducible and isomorphic to \(\mathbb{L}_{\Delta_{\alpha_{r,s}},c}\). 

\begin{example}\label{ex:vacuum}
    The simplest case $r=s=1$ with $\alpha_{1,1}=0$ gives the singular vector $v_{1,1}=L_{-1}|0\rangle$. 
    If \(b^2\notin \mathbb{Q}\) or \(b^2\in \mathbb{Z}\) then the corresponding irreducible module is $\mathbb{L}_{0,c}=U(\mathrm{Vir}_{<-1})|0\rangle$ with $L_m|0\rangle=0,\,m\geq-1$. This module is called the \emph{vacuum module}.
\end{example}

\subsection{Virasoro vertex operators}
\label{ssec:vertex}
\paragraph{General definition.}
\begin{definition}
    For given weights $\Delta_3, \Delta_2, \Delta_1$, a primary Virasoro vertex operator is an operator
    \begin{equation}\label{VO_dcd}
        V^{\Delta_2}_{\Delta_3, \Delta_1}(z)\colon \mathbb{L}_{\Delta_1, c} \rightarrow z^{\Delta_3-\Delta_2-\Delta_1} \mathbb{L}_{\Delta_3, c}[[z]], 
    \end{equation}
    defined by its commutation relation with the $\mathrm{Vir}$ generators
    \begin{equation}\label{Ln_V}
        \big[L_n, V^{\Delta_2}_{\Delta_3, \Delta_1}(z)\big] = \big(z^{n+1}\partial_z + z^n(n{+}1)\Delta_2\big)V^{\Delta_2}_{\Delta_3, \Delta_1}(z), \quad n\in\mathbb{Z}.
    \end{equation}

    For any \(v_2\in \mathbb{L}_{\Delta_2,c}\) we define a Virasoro vertex operator $V^{\Delta_2}_{\Delta_3, \Delta_1}[v_2](z)\colon \mathbb{L}_{\Delta_1, c} \rightarrow z^{\Delta_3-\Delta_2-\Delta_1} \mathbb{L}_{\Delta_3, c}((z))$ by the following properties: 
    \begin{enumerate}[label=(\alph*)]
        \item it is linear with respect to \(v_2\) 
        \item it coincides with the primary vertex operator $V^{\Delta_2}_{\Delta_3, \Delta_1}(z)$ for \(v_2=|\Delta_2\rangle\) 
        \item it satisfies recurrence relations 
            \begin{subequations}\label{desc_recursion}
                \begin{align}
                    &V^{\Delta_2}_{\Delta_3, \Delta_1}[L_{-1}v_2](z) = \partial_z V^{\Delta_2}_{\Delta_3, \Delta_1}[v_2](z),\\
                    &V^{\Delta_2}_{\Delta_3, \Delta_1}[L_{-2}v_2](z) ={}:\!T(z)V^{\Delta_2}_{\Delta_3, \Delta_1}[v_2](z)\!:.
                \end{align}
            \end{subequations}
    \end{enumerate}
    Here the normal-ordered action of the stress-energy tensor is defined by
    \begin{equation}
        :\!T(z)V^{\Delta_2}_{\Delta_3, \Delta_1}[v_2](z)\!:{} =\sum_{n=-\infty}^{-2}z^{-n-2}L_n V^{\Delta_2}_{\Delta_3, \Delta_1}[v_2](z) + \sum_{n=-1}^{+\infty}z^{-n-2}V^{\Delta_2}_{\Delta_3, \Delta_1}[v_2](z)L_n.
    \end{equation}
\end{definition}

The Virasoro commutation condition \eqref{Ln_V} can be encoded in the OPE
\begin{equation}
T(z)V^{\Delta_2}_{\Delta_3,\Delta_1}(w)=\frac{\Delta_2V^{\Delta_2}_{\Delta_3,\Delta_1}(w)}{(z-w)^2}+\frac{\partial V^{\Delta_2}_{\Delta_3,\Delta_1}(w)}{z-w}+\mathrm{reg}.
\end{equation}
The number $\Delta_2$ is referred to as the (conformal) weight of the vertex operator \(V^{\Delta_2}_{\Delta_3,\Delta_1}\). 
These relations determine a non-trivial vertex operator up to a normalization function $N_b(\Delta_3,\Delta_2,\Delta_1)$ defined by 
\begin{equation}\label{VO_normalization}
    V_{\Delta_3, \Delta_1}^{\Delta_2}(z)|\Delta_1\rangle = z^{\Delta_3 - \Delta_2 - \Delta_1}N_b(\Delta_3, \Delta_2, \Delta_1)\Big(|\Delta_3\rangle + \mathrm{O}(z)\Big).
\end{equation}

\paragraph{Degenerate vertex operators.}
Let \(\Delta_i=\Delta_{\alpha_i},\, {\scriptstyle i=1,\ldots, 3}\). If one of the \(\alpha_i\) is in the Kac spectrum $\mathrm{Kac}_b$ (given by \eqref{Kac_spectrum}) then the existence of the non-trivial vertex operator requires conditions on the weights, called \textit{fusion rules}.

In particular, if $\Delta_2=\Delta_{\alpha_{r,s}}$ for some $r,s\in\mathbb{Z}_{\geq1}$, then one has the condition $V^{\Delta_{\alpha_{r,s}}}[v_{r,s}](z)=0$ for the corresponding singular vector $v_{r,s}$. Through recursive relations \eqref{desc_recursion} this condition translates to an order $rs$ differential equation on the primary vertex operator $V^{\Delta_{\alpha_{r,s}}}(z)\eqqcolon \phi_{r,s}(z)$. Such differential equations yield that the \textit{degenerate} vertex operators $\phi_{r,s}(z)$ can act non-trivially from $\mathbb{L}_{\Delta_{\alpha_1}}$ to $\mathbb{L}_{\Delta_{\alpha_3}}$ only for a certain finite set of differences\footnote{Up to symmetry $\alpha\leftrightarrow Q-\alpha$.} $\alpha_3-\alpha_1$.

Analogous fusion rules arise in the cases of $\alpha_3$ or \(\alpha_1\) in the Kac spectrum $\mathrm{Kac}_b$.

\begin{example}\label{ex:fusion 11}
    Consider the case $r=s=1$ with $\alpha_{1,1}=0$. Due to Example~\ref{ex:vacuum} the equation on vertex operator takes the form $\partial_z\phi_{1,1}(z)=0$. 
    Hence \(\phi_{1,1}(z)\) should be independent of \(z\).
    It yields \(\Delta_3=\Delta_1\), and in terms of \(\alpha\) the fusion rule has the form $\alpha_3=\alpha_1$. 
    Thus vertex operator $\phi_{1,1}$ can be identified (up to a constant factor) with the identity operator $\mathbf1$.

    Similarly, if \(\alpha_1=\alpha_{1,1}\) the vertex operator should satisfy \(V^{\Delta_2}_{\Delta_3,\Delta_{1,1}} L_{-1} |0 \rangle=0\). It yields \(\Delta_3=\Delta_2\), and in terms of \(\alpha\) the fusion rule has the form $\alpha_3=\alpha_2$ (or $\alpha_3=Q-\alpha_2$) due to symmetry. 
      
    The case \(\alpha_3=\alpha_{1,1}\) is analogous.
\end{example}

\begin{example}\label{ex:fusion 12}
In the case $r=1,s=2$ with $\alpha_{1,2}=-b/2$ the singular vector $v_{1,2}$  and the resulting differential equation are
\begin{equation}\label{eq:deg12}
v_{1,2}=(L_{-1}^2+b^2L_{-2})|\Delta_{-b/2}\rangle \quad \Rightarrow \quad \partial_z^2\phi_{1,2}(z)+b^2:T(z)\phi_{1,2}(z):{}=0.
\end{equation}
This differential equation yields the fusion rule $\alpha_3=\alpha_1\pm b/2$.
\end{example}

Two cases of the operator $\phi_{1,2}(z)$ of the latter example are of particular importance for this paper. 
Below we use for them the notation $\psi_{\pm}(z)$, i.e. 
\begin{equation}
\psi_\varsigma(z)=\psi_{\alpha;\varsigma}(z)\coloneqq V^{\Delta_{-b/2}}_{\Delta_{\alpha-\varsigma b/2},\Delta_{\alpha}}(z),
\end{equation}
where we specify the value of $\alpha$ only if necessary. 
For brevity we refer to them also in CFT manner as degenerate fields.

\paragraph{Bosonization of $\psi_\pm$.}
For $\alpha_1,\alpha{+}\alpha_1\notin \mathrm{Kac}_b$, the primary vertex operator $V^{\Delta_\alpha}_{\Delta_{\alpha_1+\alpha},\Delta_{\alpha_1}}(z)$ can be represented via the bosonization as operator $\mathbb{F}_{\alpha_1}\to \mathbb{F}_{\alpha+\alpha_1}$ given by the normal ordered exponent of the boson field
\begin{equation}
    :e^{2\alpha \varphi(z)}:{}= e^{\ri\alpha \partial_{a_0}}z^{-2\ri\alpha a_0}\exp\left(-2\ri\alpha\sum_{n\in\mathbb{Z}_{>0}}\frac{a_{-n}}nz^n\right)\exp\left(-2\ri\alpha\sum_{n\in\mathbb{Z}_{>0}}\frac{a_{n}}{-n}z^{-n}\right).
\end{equation}
Furthermore, for $\alpha=\alpha_{1,2}=-b/2$, the corresponding boson exponent $e^{-b\varphi}$ satisfies \eqref{eq:deg12}, and thus it represents $\psi_+$.

Representing $\psi_-$ can be done just by the standard procedure of dressing $\psi_+$ with a screening operator that shifts momenta by \(-2b\). 
Hence we have
\begin{equation}\label{dressing}
    \psi_-(z)=\oint_{c_z} :\!e^{2b\varphi(w)}\!:{} :\!e^{-b\varphi(z)}\!: \rd w=\oint_{c_z} \frac{:\!e^{b\big(2\varphi(w)-\varphi(z)\big)}\!:}{(w-z)^{-b^2}}\rd w.
\end{equation}
Here we should have a closed integration contour, for which we can take a small cycle $c_z$ around $z$ if $b^2\in\mathbb{Z}$. 

Let us set \(b=\ri\sqrt{k}\) for \(k\in\mathbb{Z}_{>0}\). 
Then we have 
\begin{equation}\label{deg_bosonic}
    \psi_+(z)={}:e^{-\ri\sqrt{k}\varphi(z)}:, \qquad \psi_-(z)={}:D_{k-1}(\partial\varphi,\ldots,\partial^{k-1}\varphi) e^{\ri\sqrt{k}\varphi(z)}:\,,  
\end{equation}
where the differential polynomial $D_{k-1}$ is given by
\begin{equation}
(k{-}1)!D_{k-1}(\partial\varphi,\ldots,\partial^{k-1}\varphi)={}:e^{-2\ri\sqrt{k}\varphi(z)}\partial^{k-1}_z\left(e^{2\ri\sqrt{k}\varphi(z)}\right):.  
\end{equation}
The normalization of $\psi_{\pm}$ chosen by \eqref{deg_bosonic} (for $b=\ri\sqrt{k}$),  is used in Sec.~\ref{ssec:normalizations}. 
In terms of the function $N_b$ of \eqref{VO_normalization} it is given by:
\begin{subequations}
    \label{deg_bos_norm}
    \begin{align}
        N^+(\alpha_1)&=N_{b=\ri\sqrt{k}}\left(\alpha_1-\frac{b}2,-\frac{b}2,\alpha_1\right)=1, 
        \\ 
        N^-(\alpha_1)&=N_{b=\ri\sqrt{k}}\left(\alpha_1+\frac{b}2,-\frac{b}2,\alpha_1\right)=\frac{(-1)^{k-1}}{(k{-}1)!}\prod\nolimits_{j=0}^{k-2}(2\ri\sqrt{k}\,\alpha_1+j).
    \end{align}
\end{subequations}
Also, the mutual OPEs of these bosonized $\psi_{\pm}(z)$ are given by
\begin{subequations}
    \begin{align}\label{eq:OPE_deg_diff} 
        &\psi_{\pm}(z)\psi_{\mp}(w)=(\pm1)^{k-1}\binom{2k{-}2}{k{-}1}\,(z-w)^{1-\frac32k} \Big(\mathbf1+O(z-w)\Big),\\
        \label{eq:OPE_deg_same}
        &\psi_{\pm}(z)\psi_{\pm}(w)=(z-w)^{\frac{k}2}\cdot\mathrm{reg}.
    \end{align}
\end{subequations}

\begin{remark}
    In the cases of $b=\ri \sqrt{k}$, $k=1,2$, the degenerate fields $\psi_{\pm}$ can be considered as free fermions (dressed by an additional boson exponent) and symplectic fermions, respectively.
    We discuss these cases in more detail in Sec.~\ref{sec:fermions}.
\end{remark}

\subsection{Conformal blocks and their analytic properties}
\label{ssec:conformal_blocks}
\paragraph{Conformal blocks.}
Recall that \(|0\rangle \) denotes the \emph{vacuum vector}, i.e. the highest weight vector of the vacuum module \(\mathbb{L}_{0,c}\), see Ex.~\ref{ex:vacuum}. It satisfies 
\begin{equation}
    L_m |0\rangle=0, \quad m\in\mathbb{Z}_{\ge -1}
\end{equation}
We define the vacuum covector $\langle0|$ as a functional on  \(\mathbb{L}_{0,c}\) such that 
\begin{equation}\label{vacuum_covector}
    \langle0|L_m=0, \quad m\in\mathbb{Z}_{\leq1}, \qquad  \langle0| 0\rangle=1. 
\end{equation}

\begin{definition}
    Consider a formal composition of $n$ primary vertex operators
    \begin{equation}
        \mathsf{V}_n\coloneqq V^{\Delta_n}_{0,\Delta_n}(t_n)V^{\Delta_{n-1}}_{\Delta_n,\Xi_{n-2}}(t_{n-1})\dots V^{\Delta_3}_{\Xi_3,\Xi_2}(t_3) V^{\Delta_2}_{\Xi_2,\Delta_1}(t_2)V^{\Delta_1}_{\Delta_1,0}(t_1). 
    \end{equation}
    The conformal block is defined as a matrix element
    \begin{equation}\label{eq:conf block def}
        \mathcal{F}_c\big((\Xi_i)_{i=2}^{n-2};(\Delta_i)_{i=1}^n\big|(t_i)_{i=1}^n\big)=\langle0| \mathsf{V}_n|0\rangle    
    \end{equation}
\end{definition}

Several remarks are in order. 
First, the conformal block \(\mathcal{F}\) above is defined via the power series expansion in \(t_{i+1}/t_i\). 
The basic assumption in CFT is that this series converges in some open region of \(t_i\), \(\infty>|t_n|>\dots>|t_2|>|t_1|>0\) and can be analytically continued. 

Second, in the formula~\eqref{eq:conf block def} we assumed that all points \(t_i\) are nonzero. 
One can remove this constraint. Namely, if we set \(t_1=0\) and \(t_n=\infty\), then the conformal block can be written as 
\begin{equation}
    \langle \Delta_n | V^{\Delta_{n-1}}_{\Delta_n,\Xi_{n-2}}(t_{n-1})\dots V^{\Delta_3}_{\Xi_3,\Xi_2}(t_3) V^{\Delta_2}_{\Xi_2,\Delta_1}(t_2) |\Delta_1\rangle
\end{equation}
where \(|\Delta\rangle\) is the highest weight vector of \(\mathbb{M}_{\Delta,c}\) and \(\langle\Delta|\) is a dual covector on the same space.

Third, the conformal block \(\mathcal{F}\) is covariant with respect to the action of \(\SL\), which acts on $\mathbb{CP}^1$ via M\"obius transformations. 
This follows from the infinitesimal property, namely the commutation relations of vertex operators \(V\) with the $\mathfrak{sl}_2=\langle L_{-1},L_0,L_1\rangle$ subalgebra of the Virasoro algebra. 
In particular, for \(n=3\) all dependence on \(t\) follows from this symmetry and conformal block basically reduces to normalizations \(N\):
\begin{equation}
    \langle 0| V^{\Delta_3}_{0,\Delta_3}(t_3) V^{\Delta_2}_{\Delta_3,\Delta_1}(t_2) V^{\Delta_1}_{\Delta_1,0}(t_1) |0\rangle=\frac{N(0,\Delta_3,\Delta_3)N_b(\Delta_3,\Delta_2,\Delta_1)N(\Delta_1,\Delta_1,0)}{(t_2{-}t_1)^{\Delta_2+\Delta_1-\Delta_3}(t_3{-}t_2)^{\Delta_3+\Delta_2-\Delta_1}(t_3{-}t_1)^{\Delta_3+\Delta_1-\Delta_2}}.
\end{equation}
Hence one can view the analytic continuation of \(\mathcal{F}\) as a section of a bundle (bundle of conformal blocks) on \(\mathcal{M}_{0,n}\) (moduli space of curves of genus \(0\) with \(n\) punctures). 

The formula \eqref{eq:conf block def} gives the section of this bundle with prescribed asymptotic behavior at a point on the boundary of \(\mathcal{M}_{0,n}\) inside Deligne--Mumford compactification \(\overline{\mathcal{M}}_{0,n}\). 
This point corresponds to a singular curve obtained by gluing $n-2$ 3-punctured \(\mathbb{CP}^1\). Topologically this is the same data as the pants decomposition used in the previous section and conformal block~\eqref{eq:conf block def} corresponds to the decomposition depicted on Fig.~\ref{fig:pants}. 
The blocks given by \eqref{eq:conf block def} with all possible (allowed in the theory) \( (\Xi_i), {\scriptstyle i=2,\dots,n-2} \) give a basis in the fiber of the bundle of conformal blocks.

Changing the pants decomposition (or, algebraically, changing singular curves obtained by gluing 3-punctured \(\mathbb{CP}^1\)-s) one defines another basis in the fiber of the bundle of conformal blocks. 
The fundamental example corresponds to \(n=4\). 
There are three different pants decompositions in this case, \((t_1t_2|t_3t_4)\), \((t_1t_3|t_2t_4)\), and \((t_1t_4|t_2t_3)\),  see Fig.~\ref{fig:four_punctured_sphere_pants}. 
They correspond to conformal blocks given by the formulas~\eqref{eq:4pt blocks pants}.

\begin{figure}[h]
    \centering

    \begin{tikzpicture}

     \draw (0,0) arc (-90:-270:0.2cm and 0.5cm);
    
    \draw (-3,0) arc (-90:90:0.2cm and 0.5cm);
    \draw (-3,1) arc (90:270:0.2cm and 0.5cm);
    
    \draw(-1,2.2) arc (0:180:0.5cm and 0.2cm);
    \draw(-2,2.2) arc (180:360:0.5cm and 0.2cm);

    \draw (-3,0) .. controls (-2.5,0.2) and (-0.5,0.2) .. (0,0);
    \draw (-3,1) .. controls (-2,1) .. (-2,2.2);
    \draw (0,1) .. controls (-1,1) .. (-1,2.2);

    \filldraw (-3,0.5) circle (2pt) node[left=4pt] {$t_1$};
    \filldraw (-1.5,2.2) circle (2pt) node[above=4pt] {$t_2$};
    
    \begin{scope}[xshift=3cm]

    \draw[dashed] (-3,0) arc (-90:90:0.2cm and 0.5cm); 
    \draw (-3,0) arc (-90:-270:0.2cm and 0.5cm);
    
    \draw(0,0) arc (-90:90:0.2cm and 0.5cm);
    \draw(0,1) arc (90:270:0.2cm and 0.5cm);
    
    \draw(-1,2.2) arc (0:180:0.5cm and 0.2cm);
    \draw(-2,2.2) arc (180:360:0.5cm and 0.2cm);
    
    \draw (-3,0) .. controls (-2.5,0.2) and (-0.5,0.2) .. (0,0);
    \draw (-3,1) .. controls (-2,1) .. (-2,2.2);
    \draw (0,1) .. controls (-1,1) .. (-1,2.2);

    \filldraw (-1.5,2.2) circle (2pt) node[above=4pt] {$t_3$};
    \filldraw (0,0.5) circle (2pt) node[right=4pt] {$t_4$};
    \end{scope} 

     \begin{scope}[xshift=9cm]

     \draw (0,0) arc (-90:-270:0.2cm and 0.5cm);
    
    \draw (-3,0) arc (-90:90:0.2cm and 0.5cm);
    \draw (-3,1) arc (90:270:0.2cm and 0.5cm);
    
    \draw(-1,2.2) arc (0:180:0.5cm and 0.2cm);
    \draw(-2,2.2) arc (180:360:0.5cm and 0.2cm);

    \draw (-3,0) .. controls (-2.5,0.2) and (-0.5,0.2) .. (0,0);
    \draw (-3,1) .. controls (-2,1) .. (-2,2.2);
    \draw (0,1) .. controls (-1,1) .. (-1,2.2);

    \filldraw (-3,0.5) circle (2pt) node[left=4pt] {$t_1$};
    \filldraw (-1.5,2.2) circle (2pt) node[above=4pt] {$t_3$};

    \end{scope}
    
    \begin{scope}[xshift=12cm]

    \draw[dashed] (-3,0) arc (-90:90:0.2cm and 0.5cm); 
    \draw (-3,0) arc (-90:-270:0.2cm and 0.5cm);
    
    \draw(0,0) arc (-90:90:0.2cm and 0.5cm);
    \draw(0,1) arc (90:270:0.2cm and 0.5cm);
    
    \draw(-1,2.2) arc (0:180:0.5cm and 0.2cm);
    \draw(-2,2.2) arc (180:360:0.5cm and 0.2cm);
    
    \draw (-3,0) .. controls (-2.5,0.2) and (-0.5,0.2) .. (0,0);
    \draw (-3,1) .. controls (-2,1) .. (-2,2.2);
    \draw (0,1) .. controls (-1,1) .. (-1,2.2);

    \filldraw (-1.5,2.2) circle (2pt) node[above=4pt] {$t_2$};
    \filldraw (0,0.5) circle (2pt) node[right=4pt] {$t_4$};
    \end{scope}    

   \begin{scope}[xshift=6cm, yshift=-5.5cm]

    \draw (-3,0) arc (-90:90:0.2cm and 0.5cm); 
    \draw (-3,0) arc (-90:-270:0.2cm and 0.5cm);
    
    \draw(0,0) arc (-90:90:0.2cm and 0.5cm);
    \draw(0,1) arc (90:270:0.2cm and 0.5cm);
    
    \draw[dashed](-1,2.2) arc (0:180:0.5cm and 0.2cm);
    \draw(-2,2.2) arc (180:360:0.5cm and 0.2cm);
    
    \draw (-3,0) .. controls (-2.5,0.2) and (-0.5,0.2) .. (0,0);
    \draw (-3,1) .. controls (-2,1) .. (-2,2.2);
    \draw (0,1) .. controls (-1,1) .. (-1,2.2);

    \filldraw (-3,0.5) circle (2pt) node[left=4pt] {$t_1$};
    \filldraw (0,0.5) circle (2pt) node[right=4pt] {$t_4$};
    \end{scope}  

    \begin{scope}[xshift=6cm, yshift=-1.1cm, yscale=-1]

    \draw (-3,0) arc (-90:90:0.2cm and 0.5cm); 
    \draw (-3,0) arc (-90:-270:0.2cm and 0.5cm);
    
    \draw(0,0) arc (-90:90:0.2cm and 0.5cm);
    \draw(0,1) arc (90:270:0.2cm and 0.5cm);
        
    \draw (-3,0) .. controls (-2.5,0.2) and (-0.5,0.2) .. (0,0);
    \draw (-3,1) .. controls (-2,1) .. (-2,2.2);
    \draw (0,1) .. controls (-1,1) .. (-1,2.2);

    \filldraw (-3,0.5) circle (2pt) node[left=4pt] {$t_2$};
    \filldraw (0,0.5) circle (2pt) node[right=4pt] {$t_3$};
    \end{scope}
    \end{tikzpicture}
    
    \caption{Three pants decompositions of the four-punctured sphere.}
    \label{fig:four_punctured_sphere_pants}
\end{figure}
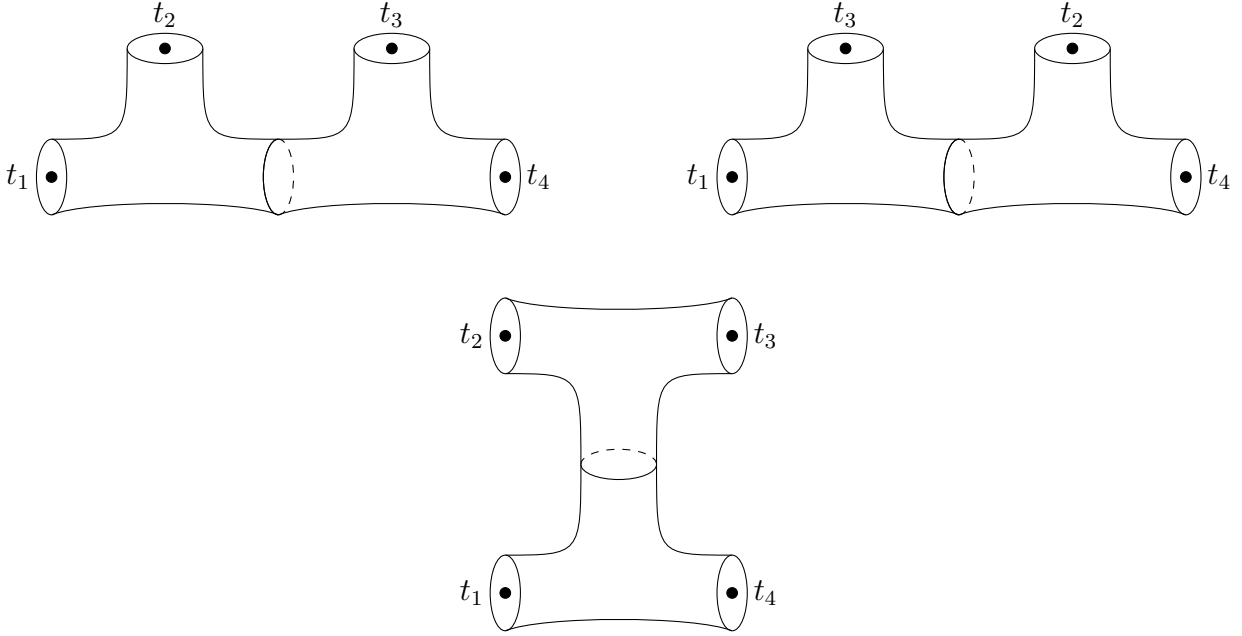
\begin{equation}\label{eq:4pt blocks pants}
    \begin{gathered}
        \langle 0| \Big(V^{\Delta_4}_{0,\Delta_4}(t_4)V^{\Delta_3}_{\Delta_4,\Xi}(t_{3})\Big)\;  \Big( V^{\Delta_2}_{\Xi,\Delta_1}(t_2)V^{\Delta_1}_{\Delta_1,0}(t_1)\Big)|0 \rangle, \quad 
        \langle 0|\Big( V^{\Delta_4}_{0,\Delta_4}(t_4)V^{\Delta_2}_{\Delta_4,\Xi}(t_{2})\Big) \; \Big(V^{\Delta_3}_{\Xi,\Delta_1}(t_3)V^{\Delta_1}_{\Delta_1,0}(t_1)\Big)|0 \rangle, 
        \\[0.3cm] 
        \langle 0| V^{\Delta_4}_{0,\Delta_4}(t_4) \Big(V^{\Xi}_{\Delta_4, \Delta_1} \big[ V^{\Delta_3}_{\Xi,\Delta_2}(t_3-t_2) |\Delta_2\rangle\big] (t_{2}) \Big) V^{\Delta_1}_{\Delta_1,0}(t_1)|0 \rangle.
    \end{gathered}    
\end{equation}

Note that such formulas depend not only on the pants decomposition but also on the cyclic order of holes on each pair of pants. 
For example, for the pants containing 0 and two punctures \(t_1\) and \(t_2\) one can write two formulas 
\begin{equation}
    \Big( V^{\Delta_2}_{\Xi,\Delta_1}(t_2)V^{\Delta_1}_{\Delta_1,0}(t_1)\Big)|0 \rangle, \qquad \text{ vs. } \qquad \Big( V^{\Delta_1}_{\Xi,\Delta_2}(t_1)V^{\Delta_2}_{\Delta_2,0}(t_2)\Big)|0 \rangle. 
\end{equation}
In the analytic setting, the difference is \(|t_2|>|t_1|\) vs. \(|t_1|>|t_2|\). 
The corresponding conformal blocks will have the same asymptotic behavior and differ by a constant factor. 
This is the so-called \emph{braiding transformation} \cite{MS88}. The square of the braiding transformation corresponds to \(t_2\) making a circle around \(t_1\). 
This is a (particular example of) \emph{monodromy transformation}.
The relation between conformal blocks given in \eqref{eq:4pt blocks pants} is given by a composition of braiding transformations and \emph{fusion transformations} \cite{MS88} that do not change the order of punctures but change the pants decomposition.

\paragraph{Braiding and fusion of the degenerate field $\psi_{\pm}$.} 
Following  \cite{ILT14} we will consider below conformal blocks with \(n\) generic fields \(V\) and two degenerate ones \(\psi_\pm\) at points \(z\) and \(z_0\). 
Similarly to the above, such a conformal block depends on a pants decomposition of the sphere with \(n+2\) punctures. 
We will consider only pants decompositions that are refinements of the one given in Fig.~\ref{fig:pants}. Namely we insert \(z,z_0\) into some of the pairs of pants there and then further decompose them. 
The corresponding conformal block (for certain pants decomposition) have the form 
\begin{equation}\label{correlator_two_deg}
    \big\langle 0 \big| 
    V^{\Delta_n}(t_n) 
    \cdots 
    V^{\Delta_{j{+}1}}(t_{j{+}1}) 
    \; \psi_{-\varsigma_0}(z_0)  \;
    V^{\Delta_{j}}(t_{j}) 
    \cdots 
    V^{\Delta_{i{+}1}}(t_{i{+}1})  
    \; \psi_{\varsigma}(z)\; V^{\Delta_i}(t_i) 
    \cdots 
    V^{\Delta_2}(t_1) \big|0 \big\rangle, 
\end{equation}
where \(\varsigma,\varsigma_0=\pm\) and we omitted indices of \(V\) for shortness.

Consider, for example, the insertion of \(\psi_\pm(z)\) into the \((i-1)\)-st pair of pants, see Fig.~\ref{fig:trinion}. Let us identify the punctures with \(0\) ,\(t\), and \(\infty\) such that the corresponding cycles are \(c_{i-1}\), \(c_{t_i}\), and \(c_i\) respectively. 
We represent the possible positions \(z\) of degenerate fields there by filled and empty bullets.  In the bullet point in the cycle \(c\) we take conformal blocks specified by asymptotics in which \(z\) goes to the corresponding puncture. Namely the bullets \(1,1^\circ\) correspond to pants decompositions \((0z|t\infty)\), the bullets \(2,2^\circ\) correspond to pants decompositions \((0\infty|tz)\) and the bullets \(3,3^\circ\) correspond to pants decompositions \((0t|z\infty)\). 
Similarly to~\eqref{eq:4pt blocks pants} the corresponding formulas for operators are
\begin{equation}\label{compositions}
    \begin{gathered}
        c_{i-1}:V_{p_3,p_1-\frac{\varsigma}2 \mathsf{b}^2}^{p_2}\big[|p_2\rangle\big](t)\psi_{\varsigma}(z), 
        \quad\;\;
        c_{t_i}: V_{p_3, p_1}^{p_2-\frac{\varsigma}2\mathsf{b}^2}\big[\psi_{\varsigma}(z - t)|p_2\rangle\big](t),
        \quad \;\;
        c_i: \psi_{-\varsigma}(z)\,V_{p_3-\frac{\varsigma}2\mathsf{b}^2,\, p_1}^{p_2}\big[|p_2\rangle\big](t).
    \end{gathered}
\end{equation}
The pairs of bullets \(j,j^\circ\) on the same boundary cycle differ by a braiding, while the neighboring bullets on different cycles differ by a fusion.

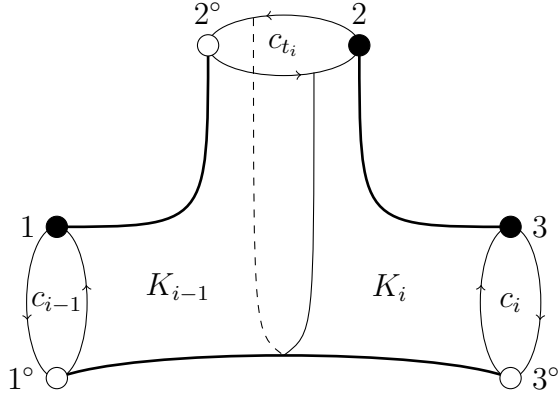
\begin{figure}[h]
    \centering
    \begin{tikzpicture}[scale=2,decoration={
    markings,
    mark=at position 0.3 with {\arrow{>}},
    mark=at position 0.8 with {\arrow{>}}}]
    
    \draw[postaction=decorate] (0,0) arc (270:-90:0.2cm and 0.5cm);
    \draw[postaction=decorate](-3,0) arc (-90:270:0.2cm and 0.5cm);
    \draw[postaction=decorate](-1,2.2) arc (0:360:0.5cm and 0.2cm);
    \draw[line width=1pt] (-3,0) .. controls (-2.5,0.2) and (-0.5,0.2) .. (0,0);
    \draw[line width=1pt] (-3,1) .. controls (-2,1) .. (-2,2.2);
    \draw[line width=1pt] (0,1) .. controls (-1,1) .. (-1,2.2);
    \draw[thin] (-1.3,2.02) .. controls (-1.3,0.3) .. (-1.5,0.15);
    \draw[thin, dashed] (-1.7,2.38) .. controls (-1.7,0.3) .. (-1.5,0.15);

    \node at (-1.5,2.2) {$c_{t_i}$};
    \node at (-3,0.5) {$c_{i{-}1}$};
    \node at (-0,0.5) {$c_i$};

    \node at (-2.2,0.6) {$K_{i-1}$};
    \node at (-0.8,0.6) {$K_i$};

    \filldraw (-3,1) circle (2pt) node[left=4pt] {$1$};
    \filldraw[fill=white] (-3,0) circle (2pt) node[left=4pt] {$1^\circ$};
    \filldraw (-1,2.2) circle (2pt) node[above=4pt] {$2$};
    \filldraw[fill=white] (-2,2.2) circle (2pt) node[above=4pt] {$2^\circ$};
    \filldraw (0,1) circle (2pt) node[right=4pt] {$3$};
    \filldraw[fill=white] (0,0) circle (2pt) node[right=4pt] {$3^\circ$};
    \end{tikzpicture}
    
    \caption{Trinion with braiding and fusion paths}
    \label{fig:trinion}
\end{figure}

In order to write formulas for braiding and fusion, it is convenient to change notation as follows
\begin{equation}\label{p_sb}
    \mathsf{b}\coloneqq-\ri b, \quad p\coloneqq (Q/2-\alpha)b \, :\, \eqref{Liouville_param}\, \Rightarrow \quad c=c(\mathsf{b})=1-6\frac{(1{-}\mathsf{b}^2)^2}{\mathsf{b}^2}, \qquad \Delta_p=\frac{4p^2-(1{-}\mathsf{b}^2)^2}{4\mathsf{b}^2}.
\end{equation}
Then, considering the cycle \(c\) to be centred at \(0\), one can write the braiding of \(\psi_{\varsigma}(z)\) and \(V^p(0)\) as 
\begin{equation}\label{braiding}
    \psi_{p;\varsigma}^{\circ}(z) =B_{\varsigma\varsigma}(p) \psi^\bullet_{p;\varsigma}(e^{-\pi\ri}z), \qquad B_{\varsigma\varsigma'}(p)=\ri^{1 -\mathsf{b}^2}e^{-\pi\ri\varsigma p}\delta_{\varsigma\varsigma'},
\end{equation}
where the transformation direction is indicated on Fig.~\ref{fig:trinion}. Note that in general braiding depends also depends on noramlization \(N_{\mathsf{b}}\), here we assumed unital normalization \(N_{\mathsf{b}}=1\) or any symmetric with respect to permutation \(p_1,p_2\) normalization.

\begin{prop}[Fusion relations \text{\cite[(3.31abc)]{ILT14}}]\label{prop:fusion}
\leavevmode
\begin{enumerate}
\item[i)]
The compositions of the degenerate field $\psi_{\varsigma}$, $\varsigma=\pm$, with the generic vertex operator \(V^{p_2}(t)\) are related by the following fusion formulas
\begin{subequations}
    \label{fusion_rules}
    \begin{align}
        \label{fusion_23}
            \psi^\bullet_{-\varsigma}(z)\, V_{p_3-\frac{\varsigma}2\mathsf{b}^2,\, p_1}^{p_2}[v_2](t) &= \sum_{\varsigma' = \pm}F_{\varsigma \varsigma'}^{[23]}\,V_{p_3, p_1}^{p_2-\frac{\varsigma'}2\mathsf{b}^2}\big[\psi_{\varsigma'}^\bullet(z - t)v_2\big](t),\\
        \label{fusion_21}
            V_{p_3, p_1-\frac{\varsigma}2\mathsf{b}^2}^{p_2}[v_2](t)\, \psi_{\varsigma}^\bullet(z) &= \sum_{\varsigma' = \pm}F_{\varsigma \varsigma'}^{[21]}\,V_{p_3, p_1}^{p_2-\frac{\varsigma'}2\mathsf{b}^2}\big[\psi_{\varsigma'}^{\circ}(z - t)v_2\big](t),\\
        \label{fusion_13}
            \psi_{-\varsigma}^{\circ}(z)\, V_{p_3-\frac{\varsigma}2\mathsf{b}^2, p_1}^{p_2}[v_2](t) &= \sum_{\varsigma' = \pm}F_{\varsigma \varsigma'}^{[13]}\,V_{p_3, p_1-\frac{\varsigma'}2\mathsf{b}^2}^{p_2}\left[v_2\right](t)\,\psi^{\circ}_{\varsigma'}(z).
        \end{align}
\end{subequations}

\item[ii)] With the unital normalization \eqref{VO_normalization} $N_{\mathsf{b}}\equiv 1$ for generic and degenerate vertex operators, the above fusion matrices $F^{[ji]}$ are given by
\begin{equation}\label{Fusion_1}
    F^{[ji]}_{\varsigma\varsigma'}=F_{\varsigma\varsigma'}(p_i,p_j|p_l) \coloneqq \frac{\Gamma\big(1 - 2\varsigma p_i\big)\Gamma\big(2\varsigma'p_j\big)}{\prod_{\pm}\Gamma\big(\frac12 - \varsigma p_i + \varsigma'p_j \pm p_l\big)}, \qquad \{l\}=\{1,2,3\}\setminus\{i,j\}.
\end{equation}
\end{enumerate}
\end{prop}
\begin{remark} \label{rem:fusion_dependence}
If we consider fusion transformations inverse to \eqref{fusion_rules}, then the corresponding fusion matrices $F^{[ij]}=\big(F^{[ji]}\big)^{-1}$ are also given by formula \eqref{Fusion_1}, due to the relation
\begin{equation}
    F(p_3,p_2|p_1)F(p_2,p_3|p_1)=\mathbf1_{2\times2}.
\end{equation}

Also, the fusion transformations \eqref{fusion_rules} are dependent because the monodromy around contractible cycles on the trinion of Fig.~\ref{fig:trinion} must be trivial. 
In this unit normalization, this is provided by the relation
\begin{equation}
    B(p_3)F(p_3,p_2|p_1)B(p_2)F(p_2,p_1|p_3)B(p_1)F(p_1,p_3|p_2)=-\ri^{-3\mathsf{b}^2}\cdot\mathbf1_{2\times2}. 
\end{equation}
\end{remark}

The pairs of pants containing \(0\) and \(\infty\) require additional comments. 
They are similar to each other, so let us concentrate on the one that contains \(0\) and punctures \(t_1\) and \(t_2\), \(|t_2|>|t_1|>0\). 
If we move \(\phi_{1,2}(z)\) to the region \(|t_1|>|z|\), then this degenerate field starts acting on the vacuum \(|0\rangle\). 
Then, due to the vacuum fusion rule (see Ex.~\ref{ex:fusion 11}) the only allowed target has \(p=p_{1,2}=\frac12-\mathsf{b}^2\). Hence we cannot define separate \(\psi_{+},\psi_-\) in this case. This index is restored via the action of \(V^{p_1}\) on the state in \(\mathbb{M}_{\Delta_{1,2},c}\). Hence the braiding connects the following conformal blocks 
\begin{equation}
     V^{p_2}_{p_3,p_1-\frac{\varsigma}{2}\mathsf{b}^2}(t_2)  \psi_{\varsigma}(z)  V^{p_1}_{p_1,0}(t_1)  \big|0  \big\rangle \qquad \text{ vs. } \qquad 
     V^{p_2}_{p_3,p_1-\frac{\varsigma}{2}\mathsf{b}^2}(t_2) V^{p_1}_{p_1-\frac{\varsigma}{2}\mathsf{b}^2,\frac12-\mathsf{b}^2}(t_1) \phi_{1,2}(z)\big|0  \big\rangle 
\end{equation}
Moreover, this braiding can be trivialized in the following sense. Similarly to the RH problem above we make a cut \([t_1,t_2]\). Then one can choose the normalization of the fields such that braiding is trivial unless \(z\) crosses the cut, in which case the conformal block is modified by a local monodromy factor. 

\section{Conformal blocks as solutions of Riemann-Hilbert problem}
\label{sec:sol}

In this section, following ideas of \cite{ILT14} for $k=1$, we modify the Riemann-Hilbert problem \ref{def:RH_annulus} in a way that it is solved by a correlation function of (periodic) vertex operators with insertion of two degenerate fields $\psi_{\pm}(z)$ for $\mathsf{b}^2=k\in\mathbb{Z}_{>0}$.

\subsection{Normalizations}
\label{ssec:normalizations}

\paragraph{Wick rotation of the Ponsot-Teschner normalization \cite{PT99}.}
A convenient normalization for the Virasoro vertex operators in the case of $\mathrm{Re}(b)>0$ was introduced by Ponsot-Teschner in \cite[(10)]{PT99}. It is given in terms of double gamma functions \(\Gamma_b\):
\begin{equation}\label{PT_norm}
    N_b(\alpha_3,\alpha_2,\alpha_1)=\frac{\Gamma_b(2Q-2\alpha_3)\Gamma_b(2\alpha_2)\Gamma_b(2\alpha_1)}{\Gamma_b(2Q-\alpha_1{-}\alpha_2{-}\alpha_3)\Gamma_b(\alpha_1{+}\alpha_2{-}\alpha_3)\Gamma_b(\alpha_3{+}\alpha_1{-}\alpha_2)\Gamma_b(\alpha_2{+}\alpha_3{-}\alpha_1)}.
\end{equation}
This leads to the fusion matrix \cite[(16)]{PT99}, given in terms of trigonometric functions. In particular, it is quasi-periodic with respect to the shift $\alpha_i\mapsto \alpha_i+b^{-1}$.

On the other hand, in the $c=1$ case, which serves as our model example, Iorgov--Lisovyy--Teschner used the normalization \cite[(4.42)]{ILT14} in terms of Barnes $G$-functions:
\begin{equation}\label{Nc=1}
    N_{\mathsf{b}=1}(p_3, p_2, p_1)= \frac{\prod\limits_{\varsigma,\varsigma'=\pm}G\big(1+\varsigma\varsigma'p_3+\varsigma'p_2+\varsigma p_1\big)}{G(1+2p_3)G(1-2p_2)G(1-2p_1)}\,.
\end{equation}
In this normalization the fusion matrices $F^{[ji]}$ also become trigonometric, and quasi-periodic with respect to shifts of $p_3,p_2,p_1$ by $\mathsf{b}^2=1$:
\begin{equation}\label{Fusion_qp}
    F^{[ji]}_{\varsigma\varsigma'}=F^{\textrm{q-per}}_{\varsigma\varsigma'}(p_i, p_j|p_l) \coloneqq \frac{\cos\big(\pi(p_l+\varsigma p_j-\varsigma'p_i)\big)}{\sin(2\pi\varsigma'p_j)}, \qquad \{l\}=\{1,2,3\}\setminus\{i,j\}.
\end{equation}

To obtain a similar normalization for arbitrary \(b\) on the ray $\mathrm{Re}(b)=0$, we use the following Virasoro--Wick rotation of the double gamma function (see, e.g.,~\cite{RT23}):
\begin{equation}
    b\to \mathsf{b}=-\ri b: \qquad \Gamma_b(x) \to \Gamma_{\mathsf{b}}(\ri x+\mathsf{b}).
\end{equation}
Under this transformation, the normalization \eqref{PT_norm} becomes
\begin{equation}\label{our_norm}
    N_{\mathsf{b}}(p_3,p_2,p_1)=\frac{\Gamma_{\mathsf{b}}\left(\mathsf{b}^{-1}(1+2p_3)\right)\Gamma_{\mathsf{b}}\left(\mathsf{b}^{-1}(1-2p_2)\right)\Gamma_{\mathsf{b}}\left(\mathsf{b}^{-1}(1-2p_1)\right)}{\prod\limits_{\varsigma,\varsigma'=\pm1}\Gamma_{\mathsf{b}}\left(\mathsf{b}^{-1}\big(\frac{1+\mathsf{b}^2}2+\varsigma\varsigma' p_3+\varsigma' p_2+\varsigma p_1\big)\right)}. 
\end{equation}
Yet it gives precisely the fusion matrices \eqref{Fusion_qp} for any $\mathsf{b}\in\mathbb{R}_{>0}$. Indeed, using the shift relation
\begin{equation}
    \Gamma_{\mathsf{b}}\left(\frac{x+\mathsf{b}^2}{\mathsf{b}}\right)=\frac{\sqrt{2\pi}\, \mathsf{b}^{x-\frac12}}{\Gamma(x)}\Gamma_{\mathsf{b}}\left(\frac{x}{\mathsf{b}}\right),
\end{equation}
and the symmetry of \eqref{our_norm} under permutations of $(-p_3),p_2,p_1$, we transform the fusion matrices \eqref{Fusion_1}:
\begin{multline}
    F^{[ji]}_{\varsigma\varsigma'}=F_{\varsigma\varsigma'}(p_i,p_j|p_l)\mapsto\frac{N_{\mathsf{b}}(p_i,\frac12{-}\mathsf{b}^2,p_i{-}\varsigma\frac{\mathsf{b}^2}2)}{N_{\mathsf{b}}(p_j{-}\varsigma'\frac{\mathsf{b}^2}2,\frac12{-}\mathsf{b}^2,p_j)}\,\frac{N_{\mathsf{b}}(p_i{-}\varsigma\frac{\mathsf{b}^2}2,p_j,p_l)}{N_{\mathsf{b}}(p_i,p_j{-}\varsigma'\frac{\mathsf{b}^2}2,p_l)}\,F_{\varsigma\varsigma'}(p_i,p_j|p_l)
    \\[0.2cm]
    =F_{\varsigma\varsigma'}^{\textrm{q-per}}(p_i,p_j|p_l), \quad \{l\}=\{1,2,3\}\setminus\{i,j\}.
\end{multline}

\paragraph{Special cases.}
Let \(\mathsf{b}=\sqrt{k}\), \(k \in \mathbb{Z}_{> 0}\). 
Then the fusion matrices \eqref{Fusion_qp} are quasi-periodic with respect to shifts of the momenta $p_i$ by $\mathsf{b}^2$.
Following \cite[App. A]{RoT25}, in this case the double gamma function $\Gamma_{\mathsf{b}}$ is written in terms of Barnes $G$-functions:
\begin{equation}
    \Gamma_{\sqrt{k}}\left(\frac{x}{\sqrt{k}}\right)=(2\pi)^{\frac{x}{2k}{-}\frac{1+k}{4k}}k^{\frac1{4k}\big(x{-}\frac{1+k}2\big)^2}\frac{G_k\big(\frac{1+k}2\big)}{G_k(x)},
\end{equation}
where we introduce the function\footnote{Our $G_k(x)$ differs from $G_{1,k}(x/k)=\prod_{i=0}^{k-1}G\left(\frac{x+i}{k}\right)$ of \cite{RoT25} by elementary prefactors.}
\begin{equation}
    G_k(x)\coloneqq (2\pi)^{\frac{(1-k)(x-k)}{2k}}k^{\frac{(x-1)(x-k)}{2k}} \prod_{i=0}^{k-1}\frac{G\left(\frac{x+i}k\right)}{G\left(\frac{k+i}{k}\right)}.
\end{equation}
The function $G_k$ obeys the following shift relations and has a convenient normalization:
\begin{equation}
    G_k(x{+}k)=\Gamma(x)G_k(x), \qquad G_k(x{+}1)=(2\pi)^{\frac{1-k}{2k}}k^{\frac{x}k-\frac12}\Gamma(x/k)G_k(x),\qquad G_k(k)=1.
\end{equation}
So the general normalization \eqref{our_norm} in this case turns into
\begin{equation}\label{special_norm}
    N_{\mathsf{b}=\sqrt{k}}(p_3, p_2, p_1) =\frac{\prod\limits_{\varsigma,\varsigma'=\pm}G_k\Big(\frac{1+k}2+\varsigma\varsigma'p_3+\varsigma'p_2+\varsigma p_1\Big)}{G_k(1+2p_3)G_k(1-2p_2)G_k(1-2p_1)}.
\end{equation}
Here we have omitted an elementary prefactor, whose presence does not affect the fusion matrices\footnote{This prefactor is $k^{-\frac{3(k-1)^2}{16k}}\big(2\pi k^{\frac{1-k}{2}}\big)^{\frac1k(p_3-p_2-p_1+\frac34(1{-}k))}G_k^{-1}\big(\frac{1+k}2\big)$.}. 
For $k=1$ this formula exactly reproduces normalization \eqref{Nc=1} of \cite{ILT14}. 
Note that this choice of normalization satisfies the condition $N_{\sqrt{k}}(p,\frac{1{-}k}2,p)=\langle p|\mathbf{1}|p\rangle=1$.

\paragraph{Degenerate field renormalization.} 
Above, we use for the degenerate fields $\psi_{p;\pm}(z)$ the same normalization \eqref{special_norm} as for the generic vertex operators.
In the next subsection we use bosonic normalization \(N^\pm\) defined in formula~\eqref{deg_bos_norm}.
These normalizations are related by
\begin{subequations}\label{deg_norm_diff}
    \begin{align}
        &N_{\sqrt{k}}\left(p-\frac{k}2,\frac12{-}k,p\right)=1=N^+(p), \\ 
        &N_{\sqrt{k}}\left(p+\frac{k}2,\frac12{-}k,p\right)=(-1)^k\,\frac{\sin 2\pi p}{\pi}\prod_{j=1}^{k-1}(2p+j)=(-1)^k(k{-}1)!\frac{\sin 2\pi p}{\pi}N^-(p),
    \end{align}
\end{subequations}
where we used \eqref{p_sb}. After this renormalization, we also restore symmetry under permutations of \((-p_3),p_2,p_1\) by using the normalizations \(N^\pm\) when \(p_1=\frac12-k\) or \(p_3=k-\frac12\) instead of \eqref{special_norm}.

\subsection{Periodic vertex operator}\label{ssec:per_VO}
\paragraph{Periodic fusion relations.}
Using the braiding and fusion of the degenerate field $\psi_{\varsigma}(z)$ from Prop.~\ref{prop:fusion}, one can compute the analytic continuation of the conformal block \eqref{correlator_two_deg} around any loop. 
The result is given by sums of conformal blocks of the same form but with shifts of the intermediate momenta. 
Since the momentum shifts in \eqref{fusion_rules} are \(\pm\frac12\mathsf{b}^2\) and every closed loop produces an even number of such shifts, the resulting shifts belong to \(\mathsf{b}^2\mathbb Z\).
The idea of \cite{ILT14} is to construct an infinite sum (Fourier series) of conformal blocks with shifted intermediate momenta that has constant monodromy in \(\SL\).

To implement the above idea, we need to make the fusion relations \eqref{fusion_rules} invariant with respect to $\mathsf{b}^2$ shifts in $p_3$ and $p_1$.
In the normalization chosen in the previous subsection, the matrix \(F\) is already quasi-periodic. 
The fusion relations can be made periodic by including an additional phase factor in the normalization \eqref{special_norm}.
For instance, we can use the normalization
\begin{equation}
    \tilde{N}_{\mathsf{b}}(p_3,p_2,p_1)=e^{\pi\ri\frac{(p_3-p_1)^2-p_2^2}k}N_{\mathsf{b}}(p_3,p_2,p_1).
\end{equation} 

Applying this change together with the change to the bosonic normalization \eqref{deg_norm_diff} turns the fusion relations \eqref{fusion_rules} into
\begin{subequations}
\label{fusion_rules_per}
\begin{align}
\label{fusion_23_per}
        \psi^\bullet_{\varsigma}(z)\, V_{p_3+\frac{\varsigma}2k,\, p_1}^{p_2}[v_2](t) &= \sum_{\varsigma' = \pm}\tilde{F}_{\varsigma \varsigma'}^{[23]}\,V_{p_3, p_1}^{p_2-\frac{\varsigma'}2k}\left[\psi^\bullet_{\varsigma'}(z - t)v_2\right](t),\\
       \label{fusion_21_per}
        V_{p_3, p_1-\frac{\varsigma}2k}^{p_2}[v_2](t)\, \psi^\bullet_{\varsigma}(z) &= \sum_{\varsigma' = \pm}\tilde{F}_{\varsigma \varsigma'}^{[21]}\,V_{p_3, p_1}^{p_2-\frac{\varsigma'}2k}\left[\psi^\bullet_{\varsigma'}(t - z)v_2\right](t),\\
        \label{fusion_13_per}
        \psi^\bullet_{\varsigma}(z)\, V_{p_3+\frac{\varsigma}2k,p_1}^{p_2}[v_2](t) &= \sum_{\varsigma' = \pm}\tilde{F}_{\varsigma \varsigma'}^{[13]}\,V_{p_3, p_1-\frac{\varsigma'}2k}^{p_2}\left[v_2\right](t)\, \psi^\bullet_{\varsigma'}(z),
    \end{align}
\end{subequations}
with the fusion matrices given by
\begin{subequations}\label{fusion_matrices_periodic}
\begin{align}
    \tilde{F}^{[23]}_{\varsigma\varsigma'}(p_3,p_2,p_1)&=\varsigma(\varsigma')^{k-1}\ri^k\, e^{\pi\ri\big(\varsigma(p_3-p_1)-\varsigma' p_2\big)}\tilde{F}_{\varsigma\varsigma'}(-p_3,p_2|p_1),\\ \tilde{F}^{[21]}_{\varsigma\varsigma'}(p_3,p_2,p_1)&=\varsigma^k(\varsigma')^{k-1}\ri\, e^{\pi\ri\big(\varsigma(p_3-p_1)-2\varsigma' p_2\big)}\tilde{F}_{\varsigma\varsigma'}(p_1,p_2|-p_3), \\ \tilde{F}^{[13]}_{\varsigma\varsigma'}(p_3,p_2,p_1)&=\varsigma (\varsigma')^{k-1}\ri^{k}\,e^{\pi\ri\big(\varsigma (2p_3-p_1)-\varsigma' p_3\big)}\tilde{F}_{\varsigma\varsigma'}(p_3,p_1|p_2),
\end{align}
\end{subequations}
where
\begin{equation}
\tilde{F}_{\varsigma\varsigma'}(p_3,p_2|p_1)=\left(\frac{\pi}{(k{-}1)!}\right)^{\frac{\varsigma'-\varsigma}2}\frac{\cos\big(\pi(p_1-\varsigma p_2+\varsigma'p_3)\big)}{\sin^{\frac{1-\varsigma}2}(2\pi p_3)\sin^{\frac{1+\varsigma'}2}(2\pi p_2)}.
\end{equation}
Here, in the fusion relations~\eqref{fusion_rules_per}, we expressed all \(\psi^\circ_\varsigma\) in terms of \(\psi^\bullet_\varsigma\) using the braiding relation~\eqref{braiding}; the corresponding braiding phases are absorbed into \(\tilde F^{[ji]}\).
Notice that these fusion relations turn out to be not just $k$-periodic but even $1$-periodic in $p_1$ and $p_3$.

\begin{prop}\label{prop:fusion_det}
For the fusion matrices \eqref{fusion_matrices_periodic} in the chosen normalization, we have $\det(\tilde{F}^{[ji]})=1$.
\end{prop}
\begin{proof}
Direct calculation.    
\end{proof}

\paragraph{Periodic vertex operator.}
Now we can sum the resulting invariant fusion relations over the $k$-shifts of $p_3$ and $p_1$.
To describe the RH problem solution in the following Sec.~\ref{ssec:solution}, it is convenient from now on to parametrize vertex operators and modules by
\begin{equation}
    \theta=\frac{1{-}k}2-p=b \alpha ,\quad \Rightarrow \quad  \Delta_{\theta}=\frac{\theta(\theta+k{-}1)}{k}=\alpha(Q-\alpha)=\Delta_{\alpha}. 
\end{equation}
\begin{definition}\label{def:block_operator}
    For a given vertex operator $V^{\theta_2}_{\theta_3,\theta_1}(t)$ of the Virasoro algebra with $\mathsf{b}^2=k\in\mathbb{Z}_{\geq1}$, we refer to the block operator
    \begin{subequations}
        \begin{gather}
            \bar{V}^{\theta_2}_{[\theta_3]_k, [\theta_1]_k}(t): \mathbb{L}_{[\theta_1]_k}\to \mathbb{L}_{[\theta_3]_k}, \qquad \mathbb{L}_{[\theta]_k}\coloneqq\bigoplus_{n\in \mathbb{Z}}\mathbb{L}_{\theta{+}nk}, 
            \\ 
            \bar{V}^{\theta_2}_{[\theta_3]_k, [\theta_1]_k}(t)\Big|_{\mathbb{L}_{\theta_1{+}n_1k}\to \mathbb{L}_{\theta_3{+}n_3k}}\coloneqq V^{\theta_2}_{\theta_3{+}n_3k,\theta_1{+}n_1k}(t)
        \end{gather}
    \end{subequations}
as \emph{periodic vertex operator} that depends on $[\theta_1]_k=\theta_1\!\!\!\mod{k}, \,\,[\theta_3]_k=\theta_3\!\!\!\mod{k}$.
\end{definition}

Similarly, we define \(\bar{\psi}^{\pm}\colon \mathbb{L}_{[\theta_1]_k}\to \mathbb{L}_{[\theta_1+k/2]_k} \) to be the block vertex operator. 
For simplicity, below we omit the bar and write \(\psi^{\pm}\) instead of \(\bar{\psi}^{\pm}\).

Recall that for generic \(\theta\) the Verma module \(\mathrm{M}_\theta\) is isomorphic to irreducible one \(\mathrm{L}_\theta\). Hence for generic \(\theta\) we can also write \(\mathbb{M}_{[\theta]_k}\) for \(\mathbb{L}_{[\theta]_k}\).

In this notation, summing the fusion relations \eqref{fusion_rules_per} over the $k$-shifts of $\theta_1$ and $\theta_3$, we obtain fusion relations on the periodic vertex operators $\bar{V}$:
\begin{subequations}
\label{fusion_rules_block}
\begin{align}
\label{fusion_23_block}
        \psi_{\varsigma}(z)\, \bar{V}_{[\theta_3+k/2]_k,\,[\theta_1]_k}^{\theta_2}[v_2](t) &= \sum_{\varsigma' = \pm}\tilde{F}_{\varsigma \varsigma'}^{[23]}\,\bar{V}_{[\theta_3]_k,\, [\theta_1]_k}^{\theta_2+\frac{\varsigma'}2k}\left[\psi_{\varsigma'}(z - t)v_2\right](t),\\
       \label{fusion_21_block}
        \bar{V}_{[\theta_3]_k,\, [\theta_1+k/2]_k}^{\theta_2}[v_2](t)\, \psi_{\varsigma}(z) &= \sum_{\varsigma' = \pm}\tilde{F}_{\varsigma \varsigma'}^{[21]}\,\bar{V}_{[\theta_3]_k, [\theta_1]_k}^{\theta_2+\frac{\varsigma'}2k}\left[\psi_{\varsigma'}(t - z)v_2\right](t),\\
        \label{fusion_13_block}
        \psi_{\varsigma}(z)\, \bar{V}_{[\theta_3+k/2]_k,\, [\theta_1]_k}^{\theta_2}[v_2](t) &=\bar{V}_{[\theta_3]_k,\, [\theta_1+k/2]_k}^{\theta_2}\left[v_2\right](t)\, \sum_{\varsigma' = \pm}\tilde{F}_{\varsigma \varsigma'}^{[13]} \psi_{\varsigma'}(z).
    \end{align}
\end{subequations}

\subsection{Solution of the modified RH problem}
\label{ssec:solution}

\paragraph{Analytic continuation.}
Consider the analytic continuation of a matrix of correlators
\begin{equation}\label{correlator_two_deg_long}
    (\Phi_{n{-}1})_{\varsigma_0,\varsigma}(z)\coloneqq \Big\langle  \bar{V}^{\theta_n}_{0,[\sigma_{n-1}]_k}(t_n)\; \psi_{-\varsigma_0}(z_0)\psi_{\varsigma}(z)\;  \bar{V}^{\theta_{n-1}}_{[\sigma_{n-1}]_k,[\sigma_{n-2}]_k}(t_{n-1}) \ldots \bar{V}^{\theta_3}_{[\sigma_3]_k,[\sigma_2]_k}(t_3)\bar{V}^{\theta_2}_{[\sigma_2]_k,[\sigma_1]_k}(t_2)\bar{V}^{\theta_1}_{[\sigma_1]_k,0}(t_1)\Big\rangle
\end{equation}
as a matrix function of $z$. Here and below, unless stated otherwise, we set \(\sigma_1=\theta_1\) and \(\sigma_{n{-}1}=1-k-\theta_n\). 
We impose the ordering $\infty{>}|t_n|{>}|t_{n-1}|{>}\dots{>}|t_1|{>}0$ of the punctures, and start from the annulus $|t_n|>|z_0|>|z|>|t_{n-1}|$, i.e. consider radially ordered correlators $(\Phi_{n{-}1})_{\varsigma_0,\varsigma}$. 
Using the fusion relation \eqref{fusion_13_block}, we can analytically continue these correlators to the regions $|t_{i+1}|>|z|>|t_i|,\,i=n{-}2,\ldots,1$, successively expressing $(\Phi_{n{-}1})_{\varsigma_0,\varsigma}(z)$ in terms of the radially ordered correlators\footnote{The product $\prod$ of the vertex operators in this and analogous correlators is considered to be in \textit{decreasing} order.}: 
\begin{equation}
    (\Phi_i)_{\varsigma_0,\varsigma}(z)\coloneqq\bigg\langle \bar{V}^{\theta_n}_{0,[\sigma_{n-1}]_k}(t_n)\;\psi_{-\varsigma_0}(z_0)\left(\prod_{j=n-1}^{i+1}\bar{V}_{[\sigma_j+\frac{k}2]_k,[\sigma_{j-1}+\frac{k}2]_k}^{\theta_j}(t_j)\right)\psi_{\varsigma}(z)\left(\prod_{j=i}^2 \bar{V}_{[\sigma_j]_k,[\sigma_{j-1}]_k}^{\theta_j}(t_j)\right)\bar{V}^{\theta_1}_{[\sigma_1]_k,0}(t_1)\bigg\rangle.
\end{equation}
In matrix form, these expressions are related by $\Phi_i=\Phi_{i-1}\big(\tilde{F}^{[13]}(\sigma_i{+}\frac{k}{2},\theta_i,\sigma_{i-1})\big)^T$.
This repeats condition~\ref{item:circle} of the RH problem Def.~\ref{def:RH_annulus}. 

Each radially ordered correlator $\Phi_i$ in the annulus $|t_i|<|z|<|t_{i+1}|$ is multivalued: it has monodromy $(B^2(\sigma_i))_{\varsigma\varsigma'}=e^{2\pi\ri\varsigma \sigma_i}\delta_{\varsigma\varsigma'}$ due to \eqref{braiding}.
If we introduce cut $[t_i,t_{i+1}]$ to make $\Phi_i$ single-valued, then this monodromy will result in a jump along the cut, as in Condition~\ref{item:segment} of Def.~\ref{def:RH_annulus}.

Furthermore, as was explained at the end of Section~\ref{ssec:conformal_blocks}, we can continue the function $\Phi_1$ into the region \(|z|<|t_1|\). 
This continuation can be made single-valued in the region \(|t_2|>|z|\), except for a jump along the cut \([t_2,t_1]\) equal to the monodromy \((B^2(\sigma_1))_{\varsigma\varsigma}\).

Finally, consider the analytic continuation to the region of large \(z\). 
In view of the degenerate field OPEs \eqref{eq:OPE_deg_diff} and \eqref{eq:OPE_deg_same}, we can multiply the radially ordered correlators by $(z_0-z)^{\frac{3}{2}k-1}$ and remove the singularity at $z=z_0$.
We will do this in the final answer in Theorem~\ref{thm:cm2_linear_solve}. This is not important for local behavior properties studied in the next paragraph.
After that we can continue $\Phi_{n-1}$ to the region \(|z|>|t_n|\) with the jump \((B^2(\sigma_{n{-}1}))_{\varsigma\varsigma}\) on the cut \([t_{n{-}1},t_{n}]\). 

\paragraph{Local behavior.}
Let us now consider the local behavior of $\Phi_i(z)$ and $\Phi_{i-1}(z)$ in a neighborhood of $t_i$ (cf. Condition~\ref{item:local} of Def.~\ref{def:RH_annulus}). We claim that
\begin{multline}\label{eq:Phi_i local}
    \Phi_i(z)=\Phi_{i-1}(z)\big(\tilde{F}^{[13]}(\sigma_i{+}k/2,\theta_i,\sigma_{i-1})\big)^T\\=H_i(z)
    \begin{pmatrix} (z-t_i)^{\theta_i} & 0\\
    0 & (z-t_i)^{1-k-\theta_i} \\
    \end{pmatrix}\big(\tilde{F}^{[23]}(\sigma_i{+}k/2,\theta_i,\sigma_{i-1})\big)^T, 
\end{multline}
where \(H_i(z)\) is a matrix-valued function holomorphic in a neighborhood of $t_i$.
The first equality in~\eqref{eq:Phi_i local}, as was explained above, follows from the fusion relation~\eqref{fusion_13_block}. 
The second equality follows from the fusion relation \eqref{fusion_23_block}, together with the asymptotics of the action of $\psi_{\pm}(z)$ on the highest weight vector.
This action has the form
\begin{equation}
    \psi_+(z)|\theta\rangle=N^+(\theta)z^{\theta}|W_+(z)\rangle, \qquad  \psi_-(z)|\theta\rangle=N^-(\theta)z^{1-k-\theta}|W_-(z)\rangle,
\end{equation}
where $|W_{\pm}(z)\rangle\in \mathbb{L}_{\theta\pm k/2}[[z]]$.

Comparing to condition~\ref{item:local} of Def.~\ref{def:RH_annulus}, the obtained local behavior has an extra term $1-k$ in the exponents of $(z-t_i)$ that, however, does not affect the monodromy. Thus radially ordered correlators $(\Phi_i(z))_{i=1}^{n-1}$ provide a solution of the RH problem with the modified local behavior. Now we state this formally and, furthermore, introduce dependence on the Fenchel--Nielsen angles into the correlators.

\paragraph{Conformal block as a solution of $k$-modified RH problem.} We use the following modification of Def.~\ref{def:RH_annulus}.

\begin{definition}\label{def:RH_modif}
    For a given generic local system $\boldsymbol{M}^{\mathrm{SL}_2}\in\mathcal{L}^{\boldsymbol{\theta}}_n$ with parametrization \eqref{partial_product_param}, the \emph{$k$-modified Riemann--Hilbert problem} asks for a set of holomorphic single-valued functions $\Phi_i(z):K_i\to \mathrm{Mat}_{2\times2}(\mathbb{C})$, $i=1,\ldots,n-1$, that satisfy Conditions~\ref{item:circle} and \ref{item:segment} of Def.~\ref{def:RH_annulus}, with the diagonal factor in condition~\ref{item:local} replaced as follows (cf.~\eqref{RH problem_diss_local}):
    \begin{equation}\label{local_modified}
            \begin{aligned}
                U^\times(t_1)&:\Phi_1^\pm(z) =H_1(z)   \begin{pmatrix}(z-t_1)^{\theta_1} & 0\\ 0 & (z-t_1)^{1-k-\theta_1}\end{pmatrix}C_1^{\pm}, 
                \\
                U^\times(t_i)&:\Phi_{i-1}^\pm(z)=\Phi_i^\pm(z) C_{i,i-1}=H_i(z)   \begin{pmatrix}(z-t_i)^{\theta_i} & 0\\ 0 & (z-t_i)^{1-k-\theta_i}\end{pmatrix}C_i^{\pm}, \quad  i=2,\ldots, n-1, 
                \\
                U^\times(t_n)&:\Phi_{n-1}^\pm(z)=H_n(z)   \begin{pmatrix}(z-t_n)^{\theta_n} & 0\\ 0 & (z-t_n)^{1-k-\theta_n}\end{pmatrix}C_n^{\pm},  
                \\ 
                \textrm{with}\quad &  C_i^+=C_i^-\begin{pmatrix}e^{2\pi\ri\sigma_{i-1}} & 0\\ 0 & e^{-2\pi\ri\sigma_{i-1}}\end{pmatrix},\quad 
                i=1,\ldots, n, \qquad \sigma_0\coloneqq \theta_1.
           \end{aligned}
        \end{equation}
\end{definition}

Above we have constructed special solutions of this $k$-modified RH problem for the special local system with $C_{i,i-1}^{-1}=\big(\tilde{F}^{[13]}(\sigma_i{+}k/2,\theta_i,\sigma_{i-1})\big)^T$. This fixed local system can be considered as a reference local system $\check{\boldsymbol{M}}$ from the definition of Fenchel--Nielsen coordinates in Sec.~\ref{ssec:FN}. A solution for a general local system can be obtained by modifying the periodic vertex operators as follows:
\begin{equation}
    \bar{V}_{[\sigma_i]_k,[\sigma_{i-1}]_k}^{\theta_i}(t_i)\to s_i^{-{a_0}/{\sqrt{k}}}\bar{V}_{[\sigma_i]_k,[\sigma_{i-1}]_k}^{\theta_i}(t_i), \quad i=2,\ldots, n-2.
\end{equation}
Indeed, due to the commutation relation
\begin{equation}
    \psi_{\varsigma}(z)s_i^{-{a_0}/{\sqrt{k}}}=s_i^{-{\varsigma}/{2}}\, s_i^{-{a_0}/{\sqrt{k}}}\psi_{\varsigma}(z)
\end{equation}
we obtain that after this modification $C_{i,i-1}^{-1}=\big(\tilde{F}^{[13]}(\sigma_i{+}k/2,\theta_i,\sigma_{i-1})\big)^T\mathrm{diag}(s_i^{-1/2},s_i^{1/2})$, in accordance with \eqref{FN_from_gluing}. 

Putting everything together, for $|t_n|>|z_0|>|t_{n-1}|$ we consider the sequence of functions $\Phi_i:K_i\to \mathrm{Mat}_{2\times2}(\mathbb{C})$, $i=1,\ldots,n-1$, given by the radially ordered correlators of Virasoro vertex operators for $\mathsf{b}^2=k$:
\begin{multline}\label{eq:gen_RH_sol_cft}
    (\Phi_i\big((t_j,\theta_j)_{j=1}^n;(\sigma_j,s_j)_{j=2}^{n-2}\big|z,z_0\big))_{\varsigma_0\varsigma}
    =(-\varsigma)^{k-1}\binom{2k{-}2}{k{-}1}^{-1}\,(z_0-z)^{\frac32k-1}
    \\ 
    \bigg\langle \bar{V}^{\theta_n}_{0,[\sigma_{n-1}]_k}(t_n)\; \psi_{-\varsigma_0}(z_0)\left(\prod_{j=n-1}^{i+1} s_j^{-\frac{a_0}{\sqrt{k}}}\bar{V}_{[\sigma_j+\frac{k}2]_k,[\sigma_{j-1}+\frac{k}2]_k}^{\theta_j}(t_j)\right)\psi_{\varsigma}(z)\left(\prod_{j=i}^2 s_j^{-\frac{a_0}{\sqrt{k}}}\bar{V}_{[\sigma_j]_k,[\sigma_{j-1}]_k}^{\theta_j}(t_j)\right)\bar{V}^{\theta_1}_{[\sigma_1]_k,0}(t_1)\bigg\rangle,
\end{multline}
with $s_{n-1}=1$. 

\begin{theorem}\label{thm:cm2_linear_solve}
    The family of functions \(\big(\Phi_i\big)_{i=1}^{n-1}\) defined in~\eqref{eq:gen_RH_sol_cft} forms a solution of the $k$-modified Riemann-Hilbert problem.
    The corresponding local system belongs to the $\boldsymbol{\theta}$-leaf $\mathcal{L}_n^{\boldsymbol{\theta}}$ and has the Fenchel--Nielsen lengths $(\sigma_j)_{j=2}^{n-2}$ and angles $(s_j)_{j=2}^{n-2}$. 
    
\end{theorem}

    Note that this theorem is conditional, since it is based on the assumed existence of the conformal blocks as analytic functions; see Sec.~\ref{ssec:conformal_blocks}.
    Moreover, even if Virasoro conformal blocks are analytic, this does not immediately imply the analyticity of~\eqref{eq:gen_RH_sol_cft}, since the latter is an infinite series of such blocks.
    We conjecture that this series converges for \(k\geq 1\).

\begin{remark}
    In the construction above, we put $z_0\in K_{n{-}1}$ for convenience of exposition. 
    The radially ordered correlator as a function of $z_0$ can be analytically continued to other $K_i$'s in the same way as it was continued as a function of $z$.
\end{remark}

Notice that by moving punctures \(t_1,\dots,t_n\) while keeping \(\theta_i,\sigma_i,s_i\) fixed, we preserve the monodromy, thus construct an isomonodromic deformation of the solution to the modified RH problem.

The $1$-periodicity of the fusion matrices $\tilde{F}^{[ji]}(\theta_3,\theta_2,\theta_1)$ in $\theta_3$ and $\theta_1$ implies
\begin{corr}\label{cor:generic solution}
    Generic linear combinations of the solutions \eqref{eq:gen_RH_sol_cft}
    \begin{equation}
        \Phi_i(\Lambda)=\sum_{l_2,\ldots l_{n-2}=0}^{k-1}\Lambda_{l_2,\ldots, l_{n-2}}\Phi_i\big((t_j,\theta_j)_{j=1}^n;(\sigma_j+l_j,s_j)_{j=2}^{n-2}\big|z,z_0\big),
    \end{equation}
    depending on a tensor $\Lambda\in(\mathbb{C}^{k})^{\otimes n-3}$,
    also provide solutions of the $k$-modified RH problem of Theorem~\ref{thm:cm2_linear_solve}.
\end{corr}
In general these solutions are expected to be linearly dependent as functions of $z$. Indeed, $\Lambda$ takes values in the space of dimension $k^{n-3}$, while the dimension of the space of solutions of $k$-modified RH problem has linear growth in $n$. The latter is discussed in the following Sec. \ref{ssec:dimension}.
\begin{remark}
    Notice that shifts of $\theta_3,\theta_1$ by \(k\) and by \(1\) correspond to the shifts produced by the screening currents $e^{2b\varphi(w)}$ and $e^{2b^{-1}\varphi(w)}$, respectively.
\end{remark}

\paragraph{Normalization condition and tau function.}
The degenerate field OPEs \eqref{eq:OPE_deg_diff} and \eqref{eq:OPE_deg_same} imply the following asymptotic behavior of the solution \eqref{eq:gen_RH_sol_cft}:
\begin{equation}\label{Phi_to_tau}
    \Phi_{n{-}1}(z,z_0)\sim\langle\bar{\mathsf{V}}_n(\Lambda)\rangle\,\mathbf1_{2\times2}, \qquad z\to z_0,
\end{equation}
where we have introduced the correlator of a product of periodic vertex operators
\begin{equation}
\bar{\mathsf{V}}_n(\Lambda)\coloneqq\sum_{l_2,\ldots l_{n-2}=0}^{k-1}\Lambda_{l_2,\ldots, l_{n-2}}\bar{V}^{\theta_n}_{0,[\sigma_{n-1}]_k}(t_n)\left(\prod_{i=n-1}^2 s_i^{-\frac{a_0}{\sqrt{k}}}\;\bar{V}_{[\sigma_i+l_i]_k,[\sigma_{i-1}+l_{i-1}]_k}^{\theta_i}(t_i)\right)\bar{V}^{\theta_1}_{[\sigma_1]_k,0}(t_1)
\end{equation}
Below we consider this correlator as a function on $(t_i)_{i=1}^n$. 
For $k=1$ ($c=1$), it gives a formula for the $n$-point isomonodromic tau function. 
Thus we define its analogue as follows.
\begin{definition}
For $\mathsf{b}^2=k$, we refer to correlator
\begin{equation}\label{eq:tau_def_gener}
    \uptau\big(\Lambda;(\theta_i)_{i=1}^n;(\sigma_i,s_i)_{i=2}^{n-2}\big|(t_i)_{i=1}^n\big)\coloneqq\langle\bar{\mathsf{V}}_n(\Lambda)\rangle
\end{equation}
as $n$-point $k$-tau function.
\end{definition}

\begin{remark}
    For a generic local system the space generated by $\{\uptau\big(\Lambda;(\theta_i)_{i=1}^n;(\sigma_i,s_i)_{i=2}^{n-2}\big|(t_i)_{i=1}^n\big)\}_{\Lambda\in (\mathbb{C}^{k})^{\otimes(n-3)}}$ has dimension $k^{n-3}$. Indeed, monomials in $t_1,\dots, t_{n}$ appearing in the expansion of the tau-function
    \begin{equation}
        \uptau\big(e_{l_{2}}\otimes\dots \otimes e_{l_{n-2}};(\theta_i)_{i=1}^n;(\sigma_i,s_i)_{i=2}^{n-2}\big|(t_i)_{i=1}^n\big),
    \end{equation}
    corresponding to a standard basis vector are
    \begin{equation}
    t_{1}^{a_1}t_2^{\Delta_{\sigma_{2} {+} l_2{+}kr_2}{-}\Delta_{\theta_2}{-}\Delta_{\sigma_{1}}+a_2}t_{3}^{\Delta_{\sigma_{3}+l_3 + kr_3}-\Delta_{\theta_3}-\Delta_{\sigma_{2} + l_2+kr_2}+a_3}\dots t_{n-1}^{\Delta_{\sigma_{n-1}}-\Delta_{\theta_{n-1}} -\Delta_{\sigma_{n-2}+l_{n-2}+kr_{n-2}}+a_{n-1}}t_{n}^{-\Delta_{\theta_{n}} -\Delta_{\sigma_{n-1}}+a_{n}},
    \end{equation}
    where vectors $(a_{1},\dots,a_{n})$ and $(r_2,\dots,r_{n-2})$ run over $\mathbb{Z}_{\geq 0}^{n}$ and $\mathbb{Z}^{n-3}$, respectively.
    For different values of vectors $(l_2,\dots,l_{n-2})\in (\mathbb{Z}/k\mathbb{Z})^{n{-}3}$ such monomials cannot coincide since 
    \begin{equation}
        \Delta_{\sigma+kr + l}-\Delta_{\sigma+k\tilde{r} + \tilde{l}}= \frac{(2\sigma+l+\tilde{l}-1+k(r+\tilde{r}+1))(k(r-\tilde{r})+l-\tilde{l})}{k},
    \end{equation}
    cannot vanish unless $l = \tilde{l}$.

    Note that the dimension of space of tau functions being $k^{n-3}$ does not contradict the linear dependence of RH problem solutions as in such depenence we allow coefficients to be functions of $t_1,\dots, t_n$. 
\end{remark}

\paragraph{Linear problem and determinant.}
    For a solution $\big(\Phi_i(z)\big)_{i=1}^{n-1}$ of the $k$-modified RH problem, consider determinants $\det\Phi_i(z):K_i\to\mathbb{C}$ and connections $\Phi_i'(z)\Phi_i^{-1}(z):K_i\to\mathrm{Mat}_{2\times2}(\mathbb{C})$.
\begin{prop}\label{prop:conn_det_props}
\leavevmode
\begin{itemize}
   \item[i)] Determinants $\det\Phi_i(z)$ are glued to a holomorphic function on $\RSmp$. At each puncture $t_i$, the resulting function $\det\Phi(z)$ has a pole of order $k{-}1$. It therefore has $n(k{-}1)$ zeros on $\RSmp$, counted with multiplicity.
   
   \item[ii)] Connections $\Phi_i'(z)\Phi_i^{-1}(z)$ are glued into a meromorphic matrix function $A(z)$ on $\RSmp$. Besides simple poles at the punctures $t_i$, poles of $A(z)$ occur at the zeros of $\det\Phi(z)$.
\end{itemize}
\end{prop}
\begin{proof}
    Since all fusion and local monodromy (i.e. \(\tilde{F}\) and \(B^2\)) matrices belong to \(\SL\), functions \(\det\Phi_i(z)\) are glued into a single-valued function $\det\Phi(z)$. Its poles follow from the local behavior~\eqref{local_modified}. The zeroes follow from the degree counting. 

    Connections $\Phi_i'(z)\Phi_i^{-1}(z)$ are glued into a matrix function $A(z)$ by conditions~\ref{item:circle} and~\ref{item:segment} of Def.~\ref{def:RH_annulus}. Its singularities follow from the zeros of \(\det\Phi\) and the singular behavior of \(\Phi\) at the punctures \(t_i\).
\end{proof}

\subsection{Dimension of the $k$-modified RH problem solution space}
\label{ssec:dimension}

\paragraph{Estimate for $k>1$.}

As stated before (Prop.~\ref{prop:RH_unique}), the standard RH problem solution is unique up to a normalization. For $k>1$ we cannot follow this proof as it requires invertibility of the RH problem solution $\Phi(z)$ on $\RSmp$, while $\det\Phi(z)$ has $n(k{-}1)$ zeroes (recall Prop.~\ref{prop:conn_det_props}). 
Yet we can discuss the number of parameters for the $k$-modified RH solution. 

Generically, assume that zeroes $\{w_i\}_{i=1}^{n(k-1)}$ of the determinant function $\det\Phi(z)$ are simple. Then the connection matrix $A(z)$ for the $k$-modified RH solution can be written as
\begin{equation}\label{k_connection_form}
    A(z) = \sum_{i=1}^n\frac{A_i}{z- t_i} + \sum_{i=1}^{n(k{-}1)}\frac{B_i}{z - w_i}, \qquad A_i \sim \mathrm{diag}(\theta_i,1{-}k-\theta_i), \quad B_i \sim \mathrm{diag}(1,0).
\end{equation}
The space of functions of form \eqref{k_connection_form} is a submanifold of the product
\begin{equation}\label{eqn:orb_prod}
\prod_{i=1}^n\mathcal{O}_{(\theta_i,1{-}k-\theta_i)} \times \mathrm{Sym}^{n(k-1)}\left(\mathcal{O}_{(1,0)} \times \RSmp\right),
\end{equation}
where $\mathcal{O}_{(\lambda_1,\lambda_2)}$  is the adjoint orbit of matrix $\mathrm{diag}(\lambda_1,\lambda_2)$. Such orbit is two-dimensional for $\lambda_1\neq\lambda_2$. So, for generic $\theta_i$'s, this product has dimension $n(3k{-}1)$. 

Connections $A(z)$ of the form \eqref{k_connection_form}, coming from the $k$-modified RH problem solution, are subjects to several other conditions.
First, the regularity of $A(z)$ at $\infty$ is provided by the vanishing of the corresponding residue, i.e. $\sum_{i=1}^n A_i + \sum_{i=1}^{n(k{-}1)} B_i=0$; this reduces the dimension by $2$. Prescribing the monodromy of the normalized solution around punctures $(t_i)_{i=1}^n$ reduces the dimension by $2(n{-}1)$. Imposing the absence of the (non-diagonalizable) monodromy around zeroes $\{w_i\}_{i=1}^{n(k-1)}$ reduces the dimension by $n(k{-}1)$. In total, these conditions reduce the dimension to $(3k{-1})n-2-2(n{-}1)-n(k{-}1)=2n(k{-}1)$. As expected, for $k=1$ this already gives zero.

For $k>1$ the asymptotic behavior of solution \eqref{eq:gen_RH_sol_cft} near the base point $z_0$ implies conditions on the corresponding connection. Namely, using \eqref{deg_bosonic} for $z\to z_0$ we obtain (c.f. \eqref{Phi_to_tau}):
\begin{subequations}
\begin{align}
    (\Phi_{n-1}(z,z_0))_{\pm,\pm}&= \langle\bar{\mathsf{V}}_n(\Lambda)\rangle+\frac{3k{-}2}{2c(k)}\langle T(z_0) \bar{\mathsf{V}}_n(\Lambda) \rangle(z_0-z)^2+O\big((z_0-z)^3\big),\\
    (\Phi_{n-1}(z,z_0))_{\pm,\mp}&=(z_0-z)^{2k-1}\mathrm{reg},
\end{align}
\end{subequations}
where $T(z_0)$ to be radially ordered inside the correlator.
This yields additional $4k{-}1$ conditions on the connection matrix $A(z)$, namely   \begin{subequations}\label{eq:asympt_conditions}
    \begin{align}
        &A_{\pm,\mp}^{(m)}(z_0) = 0,\qquad 0\leq m \leq 2k{-}3,\\
    \label{eq:asympt_pt2}
        &A_{\pm,\pm}(z_0)=0, \qquad  A_{+,+}'(z_0) = A_{-,-}'(z_0).
    \end{align}
    \end{subequations}
If all these conditions above are independent, then the dimension of the solution space reduces to $2(n{-}3)(k{-}1)+ 2k{-}5$. For $k=2$, $n=3$ this gives $-1$, whereas in the next subsection we present the solution in this case explicitly. Thus we expect also in general that the conditions are dependent and the actual dimension is greater. We continue with the proof of uniqueness of the solution in the latter case.

    \paragraph{Proof of uniqueness in case of $n=3$ punctures for $k=2$.}
    
    \begin{prop}\label{prop:gRH problem_unique}
        For a given generic local system on $\mathbb{CP}^1\backslash \{t_i\}_{i=1}^3$ there exists a unique solution $\Phi(z)$ of the $2$-modified RH problem with additional conditions
        \begin{equation}\label{eq:Phi_asympt_2}
        \Phi'(z_0) = 0,\qquad \Phi''(z_0)\in \mathbb{C}\cdot \mathbf1_{2\times 2}.
        \end{equation}
    \end{prop}
    \begin{proof}
        Let $\Phi(z)$ be a solution of the $2$-modified RH problem satisfying \eqref{eq:Phi_asympt_2}. Let $\Phi^{(k=1)}(z)$ be a solution of the standard RH problem with the same monodromy and the same exponents $\{\theta_i\}_{i=1}^3$. For $i = 1,2,3$ we have expansions
        \begin{equation}
            \Phi(z) = H_i(z)\begin{pmatrix}(z{-}t_i)^{\theta_i} & \mspace{-20mu}  0\\ 0 &  \mspace{-20mu}(z{-}t_i)^{-\theta_i-1}\end{pmatrix}C_i, \qquad \Phi^{(k=1)}(z) = H_i^{(k=1)}(z)\begin{pmatrix}(z{-}t_i)^{\theta_i} & \mspace{-20mu} 0\\ 0 & \mspace{-20mu} (z{-}t_i)^{-\theta_i}\end{pmatrix}C_i
        \end{equation}
        where the functions $H_i(z),H_i^{(k=1)}(z)$ are analytic and invertible in a neighborhood of $t_i$. This implies that the function $R(z) = \Phi(z)\big(\Phi^{(k=1)}(z)\big)^{-1}$ is a rational function on $\mathbb{CP}^1$ with poles of order at most $1$ at points $t_1,t_2,t_3$. Without loss of generality we can assume $t_1,t_2,t_3\neq \infty$. Then we have $R(z) = \mathbf1_{2\times 2} + \frac{K(z)}{(z-t_1)(z-t_2)(z-t_3)}$ where $K(z)_{\varsigma,\varsigma^{\prime}}$ are polynomials of degree $3$ vanishing at $z_0$.

        We reformulate the conditions on $R(z)$ that are equivalent to conditions on $\Phi(z)$ given by  Definition \ref{def:RH_modif} and asymptotic conditions \eqref{eq:Phi_asympt_2}.

        Near punctures we have $R(z) = H_{i}(z)(z-t_i)^{\mathrm{diag}(0,-1)}\left(H_{i}^{(k{=}1)}(z)\right)^{-1}$. This is equivalent to
        \begin{equation}\label{eq:uniq_pf_eqti}
        (K(t_i)H_{i}(t_i))_{1,1}= (K(t_i)H_{i}(t_i))_{-1,1} = 0,\quad i=1,2,3.
        \end{equation}
        The conditions \eqref{eq:Phi_asympt_2} can be rewritten as
        \begin{equation}\label{eq:uniq_pf_eqas}
        R'(z_0) + A_{c=1}(z_0) = 0,\;\; R''(z_0) + A_{c=1}(z_0)^2 +A_{c=1}'(z_0) = \nu\cdot \mathbf1_{2\times 2},
        \end{equation}
        where $A_{c=1}(z) = \Phi_{c=1}'(z)\Phi_{c=1}(z)^{-1}$ and $\nu \in \mathbb{C}$ is some constant. We have a system of $6$ equations given by \eqref{eq:uniq_pf_eqti} and $8$ equations given by \eqref{eq:uniq_pf_eqas}. We have $12$ variables parametrizing $K(z)$ and additional variable $\nu$. Therefore we obtain an overdetermined system of $14$ affine-linear equations on $13$ variables. 
        
        Note that this system is explicit as $\Phi_{c=1}(z)$ and therefore $A_{c=1}(z), H_{i}^{(k{=}1)}(z)\quad i=1,2,3$ are so. We have checked explicitly that this system has a unique solution.
    \end{proof}

\subsection{Explicit answers in the case $n = 3$}\label{ssec:simple_cases}
\paragraph{Three-point RH problem and degenerate conformal blocks.}

In this section we discuss the solution of the modified RH problem given by Theorem \ref{thm:cm2_linear_solve} for the cases $n = 3, k=1,2$ and give explicit formulas in terms of hypergeometric functions for them. For the case $k=1$ such example can be found in \cite{GM16} as well.

Note that unlike the cases $n>3$ the matrix elements of the solution of modified RH problem are given by single Virasoro algebra conformal blocks (up to a simple normalization prefactors). 
Without loss of generality assume that $t_1 = 0, t_2 = 1, t_3 = \infty$. In this case matrix element of the solution of in the region $1 > |z| > |z_0|$ is given by 
\begin{equation}\label{eq:5pt_2deg_sol}
\Phi_{\varsigma_{0}, \varsigma} = -\varsigma_{0}(z-z_0)^2\frac{N^{\varsigma}(\theta_1{-}\varsigma_{0})N^{-\varsigma_{0}}(\theta_1)N(\theta_3, \theta_2, \theta_1 {+} \varsigma {-} \varsigma_{0})}{N(\theta_3, \theta_2, \theta_1)}\mathcal{F}_{\varsigma,-\varsigma_{0}}(z,z_0),
\end{equation}
where

\begin{equation}\label{eq:CB_5pt_deg_gen}
   \mathcal{F}_{\varsigma,\varsigma_{0}} = \mathcal{F}_{c(k)}\big((k/2,k/2,\theta_{2}), (\theta_1,\theta_{1,\varsigma_{0}},\theta_{1,\varsigma+\varsigma_{0}},\theta_3)| (0, z_0, z_1, 1,\infty)\big),\;\;\;\;\theta_{1,\varsigma} =\theta_{1}+\varsigma k/2.
\end{equation}
The function $\mathcal{F}_{c(k)}$ is the $5$-point conformal block of Virasoro algebra in unital normalizations where two vertex operators are taken to be the degenerate fields $\phi_{1,2}$.

Equation \eqref{eq:deg12} implies differential equations on functions \eqref{eq:CB_5pt_deg_gen}, namely 
        \begin{equation}\label{eq:BPZ5pt}
        \langle\theta_{3}|V^{\theta}(1) :\!\!\left(-\frac{1}{\mathsf{b}^2}\partial_{z}^2 + T(z)\right)\psi_{\varsigma}(z)\!\!: \psi_{\varsigma_{0}}(z_0)|\theta_1\rangle = 0
    \end{equation}
    and similar equation for $\psi_{\zeta_{0}}(z_0)$.

The limit of functions \eqref{eq:CB_5pt_deg_gen} as $z_0 \to 0$ gives $4$-point conformal blocks where one of the vertex operators is degenerate. In this case the equation following from \eqref{eq:deg12} boils down to hypergeometric equation and the degenerate $4$-point conformal blocks in the unital normalizations are given by the matrix

\begin{multline}\label{eq:CB_4pt_deg_gen}
  \tilde{\mathcal{F}}_{\varsigma, \varsigma_{0}} = \mathcal{F}_{c(k)}\big((k/2,\theta_{2}), (\theta_{1,\varsigma_{0}},\theta_{1,\varsigma+\varsigma_{0}},\theta_3)| (0, z_0, z_1, 1,\infty)\big)=
  \\=
  z^{\theta_1^{\varsigma, \varsigma_{0}}-k/2}(1-z)^{\theta_2}\hspace{0pt}_2F_1\left(\theta_1^{\varsigma, \varsigma_{0}} {+}\theta_2{-}\theta_3,\theta_1^{\varsigma, \varsigma_{0}}
   {+}\theta_2{+}\theta_3{+}k{-}1,2 \theta_1^{\varsigma, \varsigma_{0}}|z\right),
\end{multline}
where $\theta_1^{\varsigma, \varsigma_{0}} = \varsigma(\theta_1+\frac{k}{2}(\varsigma_{0}+1)-1/2)+1/2$. 
Degeneration as $z_0\to 0$ of the solution of the corresponding RH problem gives

\begin{equation}\label{eq:4pt_1deg_sol}
    \tilde{\Phi}(z) =  -\varsigma_{0}z^2\frac{N^{\varsigma}(\theta_1{-}\varsigma_{0})N^{-\varsigma_{0}}(\theta_1)N(\theta_3, \theta_2, \theta_1 {+} \varsigma {-} \varsigma_{0})}{N(\theta_3, \theta_2, \theta_1)}\tilde{F}_{\varsigma,-\varsigma_{0}}(z).
\end{equation}

\paragraph{Case $k=1$.}

For the case $k = 1$ the analysis of the RH problem allows to express function $\Phi(z,z_0)$ in terms of functions $\tilde{\Phi}(z),\tilde{\Phi}(z_0)$. Indeed, since the limit as $z_0\to 0$ does not change the monodromy in $z$ the function 
\begin{equation}
\Phi(z,z_0)\tilde{\Phi}(z)^{-1},
\end{equation}
is rational in $z$. From the asymptotic analysis near punctures one obtains that this function is actually constant. Finally, from the initial condition $\tilde{\Phi}(z_0,z_0) = 1_{2\times 2}$ one finds 
\begin{equation}\label{eq:5pt_through_4pt}
    \Phi(z,z_0) = \tilde{\Phi}(z_0)^{-1}\tilde{\Phi}(z).
\end{equation}

Substituting the formulas \eqref{eq:5pt_2deg_sol} and \eqref{eq:4pt_1deg_sol} to \eqref{eq:5pt_through_4pt} one obtains an expression of $5$-point double degenerate conformal blocks via $4$-point degenerate conformal blocks for $c(1)=1$.

\paragraph{Case $k=2$.}

In this case the approach through the RH problem does not work (at least in the same way) due to the prescense of apparent singularities.

We use the following separated vesion of Hirota derivative
\begin{equation}
    f(e^{t}z)g(e^{{-}t}z_0) \eqqcolon \sum_{n = 0}^{\infty}\mathrm{D}_{[\log(z), \log(z_0)]}^{n}(f,g)(z,z_0)t^{n}.
\end{equation}

Denote for shortness
\begin{equation}
H(z) =\, _2F_1\left(\theta _{3}{+}\theta_1{-}\theta_2,\theta_1{-}\theta_3{-}\theta_2{-}1;2{+}2 \theta_1;z\right),\; \bar{H}(z) = H(z)|_{\theta_1\to -1-\theta_1},\; \partial_{z,z_0,\theta} = z\partial_{z} {-} z_{0}\partial_{z_{0}} {-} ( 2\theta{+}1).
\end{equation}

\begin{prop}\label{prop:explicite_cb}
The double degenerate $5$-point $c(2)={-}2$ conformal blocks are given by the following formulas.
$
    \mathcal{F}_{-\varsigma,-\varsigma_{0}} = (\mathcal{F}_{\varsigma,\varsigma_{0}})|_{\theta_1 \to -1 - \theta_1}.
    $
    \begin{multline}\label{eq:exp_5ptcb1}
    \mathcal{F}_{1,1} = \frac{(2\theta_1{+}2)(2\theta_1 {+} 3)}{(\theta_1{-}\theta_2 {+} \theta_3)(\theta_1{-}\theta_2{-}\theta_3 {-}1)((\theta_1{-}\theta_2)(\theta_1 {-}\theta_2 {+}1) {-}\theta_3(\theta_3{+}1))}\frac{z_0^{\theta_1}z^{\theta_1} \left(1{-}z_0\right)^{{-}\theta_2{-}1}(1{-}z)^{{-}\theta_2{-}1} }{\left(z{-}z_0\right)^2}\times 
    \\
     \left(2(z{-}z_0)\mathrm{D}_{[\log(z), \log(z_0)]}^{2}\left(H(z), H(z_0)\right) {-} (z {+} z_0)\mathrm{D}_{[\log(z), \log(z_0)]}^{1}\left(H(z), H(z_0)\right)\right).
\end{multline}

\begin{equation}\label{eq:exp_5ptcb2}
    \mathcal{F}_{1,-1} = \frac{z_0^{\theta_1}z^{{-}1{+}\theta_1}  \left(1{-}z_0\right)^{-\theta_2{-}1}(1{-}z)^{-\theta_2{-}1}}{(2\theta_1+1)(2\theta_1+2)\left(z{-}z_0\right)^2}\left((z-z_0)\partial_{z,z_0,\theta_{1}}^2 {-} (z {+} z_{0})\partial_{z,z_0,\theta_{1}}\right) \left(\bar{H}(z) \, H(z_0)\right).
\end{equation}
\end{prop}
\begin{proof}
     
    Equations \eqref{eq:BPZ5pt} can be transformed into equations on functions $\mathcal{F}_{\varsigma,\varsigma_{0}}$
    \begin{equation}\label{eq:BPZ5rpt1}
         \left({-}\frac{1}{2}\partial_{z}^2 + \frac{2z {-} 1}{z(1{-}z)}\partial_{z} + \frac{1}{(z{-}z_0)^2} + \frac{z_0(1{-}z_0)}{(z{-}z_0)z(1{-}z)}\partial_{z_0} + \frac{1}{z(1{-}z)}\left(\frac{\Delta_{1}}{z} {-} \frac{\Delta_{2}}{z{-}1} + 2 {-} \Delta_{3}\right)\right)\mathcal{F}_{\varsigma,\varsigma_{0}}= 0,
        \end{equation}
    and the equation obtained from this by switching $z\leftrightarrow z_0$.
    
    Fix a generic initial point $(z,z_0)=(x,x_0)$. A solution $F$ of equations \eqref{eq:BPZ5rpt1} is analytic in a neighborhood of $(x,x_0)$ and can be expressed by a convergent series in two variables in this neighborhood. Note that equations \eqref{eq:BPZ5rpt1} allow to express any derivative $\partial_z^i\partial_{z_{0}}^jF$ of a solution $F$ at $(x,x_0)$ as a linear combination of $F(x,x_0), \partial_zF(x,x_0), \partial_{z_0}F(x,x_0),\partial_z\partial_{z_0}F(x,x_0)$. Therefore the space of solutions of the system is at most four-dimensional. By a direct check one obtains that the set $\{F_{\varsigma,\varsigma_{0}}\}_{\varsigma_{1},\varsigma_{0}=\pm 1}$ consists of four linearly-independent solutions of the system given by \eqref{eq:BPZ5rpt1}. The identification with conformal blocks then follows from considering asymptotes as $z,z_0$ tend to $0$.
\end{proof}
\begin{remark}
    One can also use AGT formulas \cite{AGT10,AFLT10} formulas to study conformal blocks with degenerate fields. 
    It is well known (see e.g. \cite{Nekrasov:2017bpsCftV}) that in case of 4-point conformal block with degenerate field the sum over pairs of partitions reduces to a sum over nonnegative integers and reproduces~\eqref{eq:CB_4pt_deg_gen}.

    For \(k=1\) (i.e. \(\epsilon_1+\epsilon_2=0\) in gauge theory parameters) the AGT sum for the 5-point conformal block also simplifies drastically \cite{Bonelli:2011liouville}; in this way one can derive an analogue of the formula~\eqref{eq:5pt_through_4pt} in terms of conformal blocks.
    We observed a similar structure (but more combinatorially involved) for \(k=2\) which can also be used to explain Proposition~\ref{prop:explicite_cb}.
    We expect that this structure (the simplification of the AGT sums and the bilinear expressions in terms of hypergeometric functions) can be generalised to arbitrary \(k\).
\end{remark}

\section{Bilinear relations on $c=-2$ tau functions}
\label{sec:fermions}

In Section \ref{ssec:solution} we have explained the construction of the solution of the (modified) RH problem in terms of correlation functions of periodic vertex operators in a $(1,k)$ CFT generalizing the $c= 1$ construction of \cite{ILT14}.

A significant feature of this construction for the cases $k>1$ is that the determinant of the modified RH problem solution is a non-constant rational function (cf. Prop. \ref{prop:conn_det_props}). 
This rationality is a consequence of the locality between certain bilinear combination of degenerate fields acting on the tensor square of the sum of Verma modules and certain bilinear combinations of the periodic vertex operators. 
Note that this is also based on the assumptions on the analytic properties of conformal blocks mentioned above.

In this section we study the case $k=2$ in more detail. In particular, we obtain the singular parts of OPE's between the mentioned bilinear combinations of vertex operators and the mentioned bilinear combinations of degenerate fields. 
We further use these OPE's to give an explicit formula for the determinant. 
We note that CFT provides overdetermined set of relations on this determinant. 
Studying these relations we find bilinear identities on $c=-2$ tau functions. 
 
 We also obtain the expression for the tau function in terms of the connection corresponding to the linear problem and the generalized Wick rules that allow to express correlation functions with multiple insertions of the degenerate fields via the correlation functions with insertions of at most $2$ degenerate fields.

We start from reviewing $k=1$ case in more detail as many natural objects can be defined for general $k\in \mathbb{Z}_{>0}$ from analogy with this case.

\subsection{Case $k=1$}
\label{ssec:free_fermions}

\subsubsection{Free fermions from the degenerate fields.}
In the case of $\mathsf{b}^2=k=1$ the degenerate operators are related to the free fermions, which are discussed in \cite[Sec.5]{ILT14}. They also play the central role in \cite{GM16}. Introduce operators $\Psi_{\pm}$ and $\Psi_{\pm}^*$ by

\begin{equation}
\Psi_{\pm}(z)=e^{\ri\varphi_0(z)} \psi_{\pm}(z), \qquad  \Psi_{\pm}^*(z)=e^{-\ri\varphi_0(z)} \psi_{\mp}(z).
\end{equation}
Here $\varphi_0(z)$ is a free boson commuting with the degenerate fields $\psi_{\pm}(z)$ and satisfying OPE \eqref{eq:OPE_phi}. The OPEs \eqref{eq:OPE_deg_diff} and  \eqref{eq:OPE_deg_same} turn into the free fermion OPEs
\begin{equation}
\Psi_{\varsigma}(z)\Psi_{\varsigma'}^*(w)=\frac{\delta_{\varsigma\varsigma'}}{z-w}+\mathrm{reg}, \qquad  \Psi_{\varsigma}(z)\Psi_{\varsigma'}(w)=\mathrm{reg}, \quad   \Psi_{\varsigma}^*(z)\Psi_{\varsigma'}^*(w)=\mathrm{reg}. 
\end{equation}
With respect to the extended stress-energy tensor $T(z)-{}:(\partial\varphi_0)^2:$ (cf. \eqref{T_phi}) operators $\Psi_{\pm}, \Psi^*_{\pm}$ have weight $1/2$ as it should be for the free fermions.

\subsubsection{Field $\mathcal{I}(z)$.}
The paper \cite{GM16} introduces the operator
\begin{equation}
    \mathcal{I}(z) = \sum_{\varsigma = \pm}\Psi^{*}_{\varsigma}(z)\otimes \Psi_{\varsigma}(z).
\end{equation}

The operator $\mathcal{I}(z)$ is invariant with respect to the following $\mathrm{GL}_2$-action on the free fermions
\begin{equation}
g:  \Psi_{\varsigma}(z) \mapsto \sum_{\varsigma^{\prime}=\pm} g_{\varsigma,\varsigma^{\prime}}\Psi_{\varsigma^{\prime}}(z),\;\;\; \Psi^{*}_{\varsigma}(z) \mapsto \sum_{\varsigma^{\prime}=\pm} (g^{-1})_{\varsigma,\varsigma^{\prime}}\Psi^{*}_{\varsigma^{\prime}}(z).
\end{equation}

An important feature of this operator is the locality/commutativity with respect to the tensor square of the periodic vertex operators $\bar{V}^{\theta}(t)^{\otimes 2}$. 
\begin{prop}[\text{\cite[Theorem 3]{GM16}}]
    The operator $\mathcal{I}(z)$ is local with respect to the tensor square $\bar{V}^{\theta}(t)^{\otimes 2}$. Moreover, the operator $\mathcal{I}(z)$ commutes with $\bar{V}^{\theta}(t)^{\otimes 2}$.
    \begin{equation}\label{eq:c1I_comm_loc}
    \mathcal{I}(z) \bar{V}^{\theta}(t)^{\otimes 2} = \mathrm{reg},\;\; \bar{V}^{\theta}(t)^{\otimes 2}\mathcal{I}(z) = \mathrm{reg},\;\; [\mathcal{I}(z), \bar{V}^{\theta}(t)^{\otimes 2}] = 0.
    \end{equation}
\end{prop}

This commutativity implies a hierarchy of bilinear relations satisfied by the tau function \eqref{eq:tau_def_gener}
\begin{equation}\label{eq:c1tau}
    \tau_{F} = \langle \bar{V}^{\theta_{n}}(t_{n})\dots \bar{V}^{\theta_{1}}(t_{1})\rangle,
\end{equation}
see \cite[Section 6]{GM16}. 

\subsubsection{Wick rules and three-point RH problem}\label{ssec:FF_Wick}

The commutativity \eqref{eq:c1I_comm_loc} as well implies the generalized Wick theorem, see \cite[Theorem 2.7]{AZ13} for details. Namely, there is an explicit formula for the correlation functions of the products of $\bar{V}^{\theta_{j}}(t_j)$ with multiple insertions of $\Psi$'s and $\Psi^{*}$'s in terms of such correlation functions with insertion of just one $\Psi$ and one $\Psi^{*}$. 

This allows to give a definition for the operator $\bar{V}^{\theta}(t)$ only in terms of the RH problem without using the normalization formulas \eqref{special_norm}. Let us assume for simplicity that $|z_0| < |t|$.
\begin{definition}[\text{\cite[Section 4.1]{GM16}}]\label{def:GM_MOP}
    For a given local system on $\mathbb{CP}^{1}\backslash\{0, t, \infty\}$ define the corresponding vertex operator as an operator $\bar{V}^{\theta_1}(t):\mathbb{M}_{[\theta_1]_1}\rightarrow \mathbb{M}_{[\theta_3]_1}$ such that 
    \begin{enumerate}
        \item\label{item:FFVO1} Functions 
        \begin{equation}
            (\Phi_{1})_{\varsigma_{0},\varsigma}(z) = (z-z_0)\frac{\langle \theta_{3}|\bar{V}^{\theta_{2}}(t)\Psi^{*}_{\varsigma}(z)\Psi_{\varsigma_{0}}(z_0)| \theta_{1}\rangle}{\langle \theta_{3}|\bar{V}^{\theta_{1}}(t)| \theta_{1}\rangle},\;\;\; (\Phi_{2})_{\varsigma_{0},\varsigma}(z) = (z-z_0)\frac{\langle \theta_{3}|\Psi^{*}_{\varsigma}(z)\bar{V}^{\theta_{2}}(t)\Psi_{\varsigma_{0}}(z_0)| \theta_{1}\rangle}{\langle \theta_{3}|\bar{V}^{\theta_{1}}(t)| \theta_{1}\rangle},
        \end{equation}
        satisfy conditions of Definition \ref{def:RH_annulus} (i.e. define the solution of three point RH problem). 
        \item\label{item:FFVO2} $\bar{V}^{\theta_{2}}(t)$ satisfy the Basic Bilinear Condition that is
        \begin{equation}\label{eq:BBC}
            \bigg[\sum_{\varsigma, l}\Psi_{\varsigma, l}\otimes \Psi^{*}_{\varsigma, -l}, \bar{V}^{\theta_{2}}(t)\otimes \bar{V}^{\theta_{2}}(t)\bigg] = 0.
        \end{equation}
    \end{enumerate}
\end{definition}

Note that the condition \eqref{eq:BBC} is nothing but a certain mode of the commutativity relation \eqref{eq:c1I_comm_loc}.

\begin{remark}
    Recall the factorization of the $5$-point double degenerate conformal blocks that solve the $3$-point RH which we discuss as the $k=1$ case in Section \ref{ssec:simple_cases}. This also finds an explanation via the commutativity relations \eqref{eq:c1I_comm_loc}, see \text{\cite[Section 4.2]{GM16}} for details.
\end{remark}

\subsection{Case $k=2$}
\label{ssec:I}
\subsubsection{Symplectic fermions from the degenerate fields.}

In the case of $\mathsf{b}^2=k=2$ the degenerate fields $\psi_{\pm}$ can be considered as symplectic fermions $J^{\pm}$. Indeed, the weight of the degenerate fields is $\Delta_{12}=1$, while the OPEs \eqref{eq:OPE_deg_diff} and \eqref{eq:OPE_deg_same} become those for the symplectic fermions
\begin{equation}
J^{\varsigma}(z)J^{\varsigma'}(w) = \frac{\epsilon_{\varsigma\varsigma'}}{(z-w)^2} -\epsilon_{\varsigma\varsigma'}T(w)+O(z-w), \quad \varsigma,\varsigma'=\pm,
\end{equation}
where $\epsilon_{\varsigma\varsigma'}$ is the antisymmetric symbol ($\epsilon_{+-}=1$). Note that this OPE is preserved by the transformation
\begin{equation}\label{eqn:SF_SL2symmetry}
        J^{\varsigma} \mapsto \sum_{\varsigma'=\pm}g_{\varsigma\varsigma'}J^{\varsigma'}, \qquad g\in \SL.
\end{equation}
More details on symplectic fermions one can find in \cite{CR13}.

\subsubsection{Field $I(z)$.}
Here we construct an analogue of the bilinear operator (field) $\mathcal{I}(z)$ from Sec.~\ref{ssec:free_fermions} in the $c = -2$ case.

\begin{definition}
    Define the bilinear operator $I(z): \mathbb{M}_{[\theta]_2}\otimes \mathbb{M}_{[\theta^{\prime}]_2} \rightarrow \mathbb{M}_{[\theta+1]_2}\otimes \mathbb{M}_{[\theta^{\prime}+1]_2}[[z, z^{-1}]]$ by formula
    \begin{equation}
        I(z) \coloneqq J^+(z)\wedge J^-(z) = J^+(z)\otimes J^-(z) - J^-(z)\otimes J^+(z).
    \end{equation}
\end{definition}
We will always assume that the difference $\theta-\theta^{\prime} \in \{0, \pm 1\}$. In other words, the operator $I(z)$ can be thought as acting on the tensor square of a direct sum $\mathbb{M}_{[\theta]_2} \oplus \mathbb{M}_{[\theta+1]_2}$.

Operator $I(z)$ is preserved by the $\mathrm{SL}_2\left(\mathbb{C}\right)$-symmetry \eqref{eqn:SF_SL2symmetry} of the symplectic fermion algebra :
\begin{equation}
\sum_{\varsigma_1,\varsigma_2=\pm}g_{+,\varsigma_1}g_{-,\varsigma_2}\,\,J^{\varsigma_1}(z)\wedge J^{\varsigma_2}(z) = \det(g)I(z) = I(z).
\end{equation}

The commutation relations of the modes of $I(z)$ with themselves and with the Virasoro algebra acting diagonally are described as follows:
\begin{prop}\label{prop:IIOPE}
\begin{itemize}\hfill 
    \item The field $I(z)$ is a primary field of weight $2$ with respect to $T^{\Delta}(z)= T(z)\otimes \mathbf1+\mathbf1\otimes T(z)$.
    \item Field $I(z)$ satisfies the following OPE
        \begin{equation}\label{eq:IIOPE}
            I(z)I(z_{0}) = \frac{2}{(z-z_{0})^4} - \frac{2T^{\Delta}(z_0)}{(z-z_0)^2} - \frac{\partial T^{\Delta}(z_0)}{z-z_0} + \mathrm{reg}.
        \end{equation}
\end{itemize}
        
    \end{prop}

Another reason to study the operator $I(z)$ comes from the consideration of the modified RH problem discussed in Section \ref{sec:sol}. Prop. \ref{prop:conn_det_props} shows that the determinant of the solution of the modified RH problem is a rational function. This rational function has an interpretation as a correlation function as follows.
\begin{prop}\label{prop:CFT_det}
The determinant of the modified RH problem solution \eqref{eq:gen_RH_sol_cft} is given by
\begin{equation}
\frac{\det(\Phi_{\varsigma,\varsigma_{0}})}{(z-z_0)^4} = \det\left(\langle\bar{\mathsf{V}}_nJ^{\varsigma}(z)J^{\varsigma_{0}}(z_0)\rangle_{\varsigma, \varsigma_{0} = \pm}\right) = \frac{1}{2}\langle\bar{\mathsf{V}}_n^{\otimes 2}I(z)I(z_0)\rangle.
\end{equation}
\end{prop}

\subsubsection{The determinant of the modified RH problem solution}\label{ssec:det}

Before proceeding we will need the following notations:
\begin{enumerate}
    \item  
Denote 
    \begin{equation}
\bar{V}_{[\sigma]_{2}, l,l^{\prime},[\tilde{\sigma}]_{2}, \tilde{l},\tilde{l}^{\prime}}^{\theta, \wedge 2}(t) 
    =\bar{V}_{[\sigma{+}l{+}1]_2,[\tilde{\sigma}{+}\tilde{l}{+}1]_2}^{\theta{+}1}(t){\otimes} \bar{V}_{[\sigma{+}l^{\prime}{+}1]_2,[\tilde{\sigma}{+}\tilde{l}^{\prime}{+}1]_2}^{\theta{-}1}(t) 
        {-}
        \bar{V}_{[\sigma{+}l{+}1]_2,[\tilde{\sigma}{+}\tilde{l}{+}1]_2}^{\theta{-}1}(t){\otimes} \bar{V}_{[\sigma{+}l^{\prime}{+}1]_2,[\tilde{\sigma}{+}\tilde{l}^{\prime}{+}1]_2}^{\theta{+}1}(t).
\end{equation}
\item Let $1\leq i_1 {<} i_2 {<}\dots {<}i_r {\leq} n$. Denote by $(\bar{\mathsf{V}}_n^{\otimes2})^{(i_1\dots i_r)}$ the operator $\bar{\mathsf{V}}_n^{\otimes 2}$, where we replace
\begin{equation}
    \bar{V}_{[\sigma_{i_s}+l_{i_s}]_2,[\sigma_{i_{s}-1}+l_{i_s-1}]_2}^{\theta_{i_s}}(t_{i_s})\otimes \bar{V}_{[\sigma_{i_s}+l_{i_s}^{\prime}]_2,[\sigma_{i_{s}-1}+l_{i_s-1}^{\prime}]_2}^{\theta_{i_s}}(t_{i_s}) \mapsto \bar{V}_{[\sigma_{i_{s}}]_{2}, l_{i_{s}}{+}1,l^{\prime}_{i_{s}}{+}1,[\tilde{\sigma}_{i_{s}}]_{2}, \tilde{l}_{i_{s}}{+}1,\tilde{l}^{\prime}_{i_{s}}{+}1}^{\theta, \wedge 2}(t_{i_s}),
\end{equation}
for all $s = 1,\dots, r$.
\item Denote $\uptau^{\otimes2,(i_1,\dots, i_r)} = \langle (\bar{\mathsf{V}}_n^{\otimes 2})^{(i_1,\dots, i_r)}\rangle$.
\end{enumerate}

In particular, $(\bar{\mathsf{V}}_n^{{\otimes} 2})^{(i_1,\dots,i_r)}$ is a sum of $2^{r}$ tensor products of copies of $\bar{\mathsf{V}}_n$ with shifted $\theta$'s and summation indices and $\uptau^{\otimes2,(i_1,\dots, i_r)}$ is nothing but a certain bilinear combination of functions $\uptau$ with various values of parameters $\theta_1,\dots, \theta_{n},\Lambda_{l_{2},\dots, l_{n-2}}$.

\begin{prop}\label{prop:det_answ}
The determinant of the modified RH problem solution \eqref{eq:gen_RH_sol_cft} is given by
        \begin{equation}\label{eq:det_answ}
         \det\left(\Phi(\{t_j\}, z, z_0)\right) = \uptau^2 - \sum_{1\leq i < j \leq n}\frac{\theta_{i}\theta_{j}(z-z_{0})^2(t_i - t_j)^4}{2(z-t_i)(z-t_j)(z_0 - t_i)(z_0 - t_j)}\uptau^{\otimes2,(i,j)}.
    \end{equation}
    \end{prop} 

    To prove this statement we need the two following preparatory lemmas.
\begin{lemma}\label{lemma:IV_OPE}
Field $I(z)$ is local with respect to $\bar{V}_{[\sigma+l]_2,[\tilde{\sigma}+\tilde{l}]_2}^{\theta}(t)\otimes \bar{V}_{[\sigma+l^{\prime}]_2,[\tilde{\sigma}+\tilde{l}^{\prime}]_2}^{\theta}(t)$. Moreover, their OPE is given by 

\begin{subequations}\label{eq:IV_OPE}
\begin{align}
        I(z)\left(\bar{V}_{[\sigma+l]_2,[\tilde{\sigma}+\tilde{l}]_2}^{\theta}(t)\otimes \bar{V}_{[\sigma+l^{\prime}]_2,[\tilde{\sigma}+\tilde{l}^{\prime}]_2}^{\theta}(t)\right) 
        &=
        -\theta \frac{\bar{V}_{[\sigma]_{2}, l,l^{\prime},[\tilde{\sigma}]_{2}, \tilde{l},\tilde{l}^{\prime}}^{\theta, \wedge 2}(t)}{z - t} + \mathrm{reg},\\
         \left(\bar{V}_{[\sigma+l]_2,[\tilde{\sigma}+\tilde{l}]_2}^{\theta}(t)\otimes \bar{V}_{[\sigma+l^{\prime}]_2,[\tilde{\sigma}+\tilde{l}^{\prime}]_2}^{\theta}(t)\right) I(z)
        &=
        -\theta \frac{\bar{V}_{[\sigma]_{2}, l,l^{\prime},[\tilde{\sigma}]_{2}, \tilde{l},\tilde{l}^{\prime}}^{\theta, \wedge 2}(t)}{z - t} + \mathrm{reg}.
\end{align}
\end{subequations}

\end{lemma}

\begin{proof}
    We omit the lower indices of the operators $\bar{V}_{[\sigma{+}l]_2,[\tilde{\sigma}{+}\tilde{l}]}^{\theta}(t)$ and write just $\bar{V}^{\theta}(t)$ for brevity. 
    We prove the first equation of  
    \eqref{eq:IV_OPE}, the proof of the second is analogous.
    We use the fusion rules \eqref{fusion_rules_per} and Prop. \ref{prop:fusion_det}. 
      \begin{multline} 
         I(z)\bar{V}^{\theta}(t)\otimes \bar{V}^{\theta}(t) = 
J^{+}(z) \wedge J^{-}(z)\bar{V}^{\theta}(t)\otimes \bar{V}^{\theta}(t) = \sum_{\varsigma_1, 
\varsigma_2}\tilde{F}^{[13]}_{1,\varsigma_1}\tilde{F}^{[13]}_{-1,\varsigma_2}\bar{V}^{\theta}(t)\otimes \bar{V}^{\theta}(t)J^{\varsigma_1}(z) \wedge J^{\varsigma_2}(z) 
         =\\= \sum_{\varsigma_1, 
\varsigma_2}\tilde{F}^{[13]}_{1,\varsigma_1}\tilde{F}^{[13]}_{-1,\varsigma_2}\epsilon_{\varsigma_1,\varsigma_2}\bar{V}^{\theta}(t)\otimes \bar{V}^{\theta}(t)J^{+}(z) \wedge J^{-}(z) =
\det\left(\tilde{F}^{[13]}\right)\bar{V}^{\theta}(t)\otimes \bar{V}^{\theta}(t)
 I(z) = \bar{V}^{\theta}(t)\otimes \bar{V}^{\theta}(t) I(z). 
     \end{multline} 
     
     \begin{multline} 
     I(z)\left(\bar{V}^{\theta}(t)\otimes \bar{V}^{\theta}(t)\right) = \sum_{\varsigma_1,\varsigma_2}\epsilon_{\varsigma_1,\varsigma_2}J^{\varsigma_1}(z)\bar{V}^{\theta}(t) \otimes J^{\varsigma_2}(z)\bar{V}^{\theta}(t) =
\\= \sum_{\varsigma_1, 
\varsigma_2,\varsigma_1', \varsigma_2'}\epsilon_{\varsigma_1, \varsigma_2}\tilde{F}^{[23]}_{\varsigma_1, 
\varsigma_1'}\tilde{F}^{[23]}_{\varsigma_2, \varsigma_2'}\bar{V}^{\theta+\varsigma_1'}[J^{\varsigma_1'}(z-t)|\theta\rangle](t) \otimes 
\bar{V}^{\theta+\varsigma_2'}[J^{\varsigma_2'}(z-t)|\theta\rangle](t) =\\= 
\underbrace{\det(\tilde{F}^{[23]})}_{=1}\sum_{\varsigma_1', \varsigma_2'}\epsilon_{\varsigma_1', 
\varsigma_2'}\bar{V}^{\theta+\varsigma_1'}[J^{\varsigma_1'}(z-t)|\theta\rangle](t) \otimes 
\bar{V}^{\theta+\varsigma_2'}[J^{\varsigma_2'}(z-t)|\theta\rangle](t) =\\
= -\theta 
\frac{\bar{V}^{\theta+1}(t)\otimes \bar{V}^{\theta-1}(t)-\bar{V}^{\theta-1}(t)\otimes \bar{V}^{\theta+1}(t)}{z - t} + \mathrm{reg}. 
     \end{multline}
     Here in the last equality we use the expansions \eqref{eq:I_V2sh_pf_J_vacL}.
\end{proof}

\begin{lemma}\label{lemma:I_V_OPE_shifted}
OPE between $I(z)$ and $\bar{V}_{[\sigma]_{2}, l,l^{\prime},[\tilde{\sigma}]_{2}, \tilde{l},\tilde{l}^{\prime}}^{\theta, \wedge 2}$ is given by
\begin{subequations}\label{eqs:I_V_OPE_shifted}
    \begin{align}
        I(z)\bar{V}_{[\sigma]_{2}, l,l^{\prime},[\tilde{\sigma}]_{2}, \tilde{l},\tilde{l}^{\prime}}^{\theta, \wedge 2}(t)= 
        \frac{1}{\theta}\left(\frac{4\Delta_{\theta}}{(z{-}t)^3} {+} \frac{1}{(z{-}t)^2}\partial_{t} {+} \frac{1}{z{-}t}\left(\mathbf{1}\otimes \partial_{t} {-} \partial_{t}\otimes \mathbf{1}\right)^2 \right)\left(\bar{V}_{[\sigma{+}l]_2,[\tilde{\sigma}{+}\tilde{l}]_2}^{\theta}(t)\otimes \bar{V}_{[\sigma{+}l^{\prime}]_2,[\tilde{\sigma}{+}\tilde{l}^{\prime}]_2}^{\theta}(t)\right) {+} \mathrm{reg}.\\
        \bar{V}_{[\sigma]_{2}, l,l^{\prime},[\tilde{\sigma}]_{2}, \tilde{l},\tilde{l}^{\prime}}^{\theta, \wedge 2}(t)I(z) =
        \frac{1}{\theta}\left(\frac{4\Delta_{\theta}}{(z{-}t)^3} {+} \frac{1}{(z{-}t)^2}\partial_{t} {+} \frac{1}{z{-}t}\left(\mathbf{1}\otimes \partial_{t} {-} \partial_{t}\otimes \mathbf{1}\right)^2 \right)\left(\bar{V}_{[\sigma{+}l]_2,[\tilde{\sigma}{+}\tilde{l}]_2}^{\theta}(t)\otimes \bar{V}_{[\sigma{+}l^{\prime}]_2,[\tilde{\sigma}{+}\tilde{l}^{\prime}]_2}^{\theta}(t)\right) {+} \mathrm{reg}.
    \end{align}
\end{subequations}
\end{lemma}

We present the proof in Appendix \ref{app:proofs}. It is similar to the proof of Lemma \ref{lemma:IV_OPE}.

    \begin{proof}[Proof of Prop. \ref{prop:det_answ}]
        From Prop. \ref{prop:CFT_det} and Prop. \ref{prop:conn_det_props} we have 
        \begin{equation}
        f(z, z_0) \coloneqq 2\det\left(\Phi(\{t_j\}, z, z_0)\right) = (z - z_0)^4\langle\bar{\mathsf{V}}_n^{\otimes 2}I(z)I(z_0)\rangle \in \mathbb{C}(z,z_0).
        \end{equation}

        First, we fix $z_0\neq 0,t_1,\dots , t_n$ and study the function $f(\cdot, z_0)$. Possible singularities of this function are at $t_j$'s, $z_0$, $\infty$.

        Since $I(z)$ is a primary field of dimension $2$ the function $f(\cdot,z_0)$ is regular at $\infty$. From OPE \eqref{eq:IIOPE}
        we obtain that $f(\cdot,z_0)$ is regular at $z_0$ and 
        \begin{equation}\label{eq:det_pf_normzz0}
            f(z_0,z_0) = 2\uptau^2.
        \end{equation}
        From OPE \eqref{eq:IV_OPE} we obtain
        \begin{equation}\label{eq:det_pf_exp_tj}
            \langle\bar{\mathsf{V}}_n^{\otimes 2}I(z)I(z_0)\rangle = -\frac{\theta_{j}}{z - t_{j}}\langle\bar{\mathsf{V}}_n^{(j),\otimes 2}I(z_0)\rangle + \mathrm{reg}_{z\rightarrow t_j},\quad j=1,\dots,n.
        \end{equation}

        The expansions \eqref{eq:det_pf_exp_tj} and normalization \eqref{eq:det_pf_normzz0} allow to recover the dependence of $f(z,z_0)$ on $z$,
        \begin{equation}\label{eq:det_pf_f_from_g}
            f(z, z_0) = -\sum_{j = 1}^{n}\theta_{j}\left(\frac{1}{z-t_j} - \frac{1}{z_0 - t_{j}}\right)h_{j}(z_0) + 2\uptau^2, \textit{ where }h_{j}(z_0) \coloneqq (z_0 - t_j)^4\langle\bar{\mathsf{V}}_n^{(j),\otimes 2}I(z_0)\rangle.
        \end{equation}
        
        Similarly to the consideration above the only possible singularities of $h_{j}(z_0)$ are at $t_1,\dots, t_{n}$. By Lemma \ref{lemma:I_V_OPE_shifted}, $\langle\bar{\mathsf{V}}_n^{(j),\otimes 2}I(z_0)\rangle$ has a pole of order $\leq 3$ at $t_j$. This implies
        \begin{equation}\label{eq:det_pf_gj_norm}
            h_{j}(t_j) = 0.
        \end{equation}

        Using OPE \eqref{eq:IV_OPE} we obtain
        \begin{equation}\label{eq:det_pf_exp_tj_g}
            h_{j}(z_0) = -(t_{i}-t_{j})^4\frac{\theta_{i}}{z_0 - t_i}\langle\bar{\mathsf{V}}_n^{(i,j),\otimes 2}\rangle + \mathrm{reg}_{z_0 \rightarrow t_i} = -(t_{i}-t_{j})^4\frac{\theta_{i}}{z_0 - t_i}\uptau^{\otimes2,(i,j)} + \mathrm{reg}_{z_0 \rightarrow t_i},\quad i=1,\dots, n,\; i\neq j.
        \end{equation}
        The expansions \eqref{eq:det_pf_exp_tj_g} and normalization \eqref{eq:det_pf_gj_norm} allow to  recover the rational functions $h_j$ for $j = 1,\dots, n$.
        \begin{equation}\label{eq:gj_formula}
            h_{j}(z_0) =  -\sum_{i:i\neq j}\theta_{i}(t_{i}-t_{j})^4\left(\frac{1}{z_0 - t_i}-\frac{1}{t_j - t_i}\right)\uptau^{\otimes2,(i,j)}.
        \end{equation}
        Substituting equations \eqref{eq:gj_formula} into \eqref{eq:det_pf_f_from_g} we get
        \begin{multline}
            f(z, z_0) = 2\uptau^2 + \sum_{j = 1}^{n}\sum_{i : i \neq j}\theta_{j}\theta_{i}(t_i - t_j)^4\left(\frac{1}{z-t_j} - \frac{1}{z_0 - t_{j}}\right)\left(\frac{1}{z_0 - t_{i}} - \frac{1}{t_{j} - t_{i}}\right)\uptau^{\otimes2,(i,j)} = \\ =  2\uptau^2 - (z-z_0)^2\sum_{i < j}\frac{\theta_{j}\theta_{i}(t_i - t_j)^4}{(z - t_i)(z_0 - t_i)(z - t_j)(z_0 - t_j)}\uptau^{\otimes2,(i,j)}.
        \end{multline}
    \end{proof}

\subsubsection{Bilinear relations on $c=-2$ tau functions}\label{ssec:answers}

\begin{theorem}\label{thm:tau_identities}
    The $c=-2$ tau functions satisfy the following differential-difference relations.
   \begin{equation}\label{eq:tau_identities}
   \sum_{i = 1,\;i\neq j}^{n}\theta_j\theta_i\,(t_i{-}t_j)^d\,\uptau^{\otimes2,(i,j)}=
   \begin{cases}
   4\Delta_{\theta_j}\uptau^2, & d=2\\
   2\uptau\,\partial_{t_j}\uptau, & d=1\\
   2D^2_{[t_j]}(\uptau,\uptau), & d=0
   \end{cases}, \quad \textrm{for} \quad j=1,\ldots, n,
   \end{equation}
   where $D^2_{[t_j]}(\uptau,\uptau)=\uptau\partial_{t_j}^2\uptau-(\partial_{t_j}\uptau)^2$ is the second Hirota derivative.
\end{theorem} 
\begin{proof}
    The identities follow from Lemma \ref{lemma:I_V_OPE_shifted} that gives the expansion of the correlating function $\langle \bar{\mathsf{V}}_n^{\otimes 2}I(z)I(z_0)\rangle$ in the limit $z, z_0 \to t_j$ compared to the expansion of the  explicit expression for this function given by Prop. \ref{prop:det_answ}. Note that even though the proof of Prop. \ref{prop:det_answ} uses Lemma \ref{lemma:I_V_OPE_shifted} the only part of Lemma \ref{lemma:I_V_OPE_shifted} used in the proof of Prop. \ref{prop:det_answ} is that $I(z_0)\bar{\mathsf{V}}_n^{(j),\otimes 2}$ has a pole of order $\leq 3$ as $z_0 \to t_j$.
    
    Recall the functions $h_{j}(z_0) = (z_0 - t_j)^4\langle\bar{\mathsf{V}}_n^{(j),\otimes 2}I(z_0)\rangle$ from the proof of Prop.~\ref{prop:det_answ}. Then from Prop.~\ref{lemma:I_V_OPE_shifted} we get expansions
    \begin{equation}
    \theta_{j}\bar{\mathsf{V}}_n^{(j),\otimes 2}I(z_0) = \frac{4\Delta_{\theta_j}\bar{\mathsf{V}}_n^{\otimes 2}}{(z_0 {-} t_j)^3} {+} \frac{1}{(z_0 {-} t_j)^2}\partial_{t_j}\left(\bar{\mathsf{V}}_n^{\otimes 2}\right) {+} \frac{1}{(z_0 {-} t_j)}\left(\mathbf{1}\otimes \partial_{{j}} {-} \partial_{{j}}\otimes \mathbf{1}\right)^2\bar{\mathsf{V}}_n^{\otimes 2} {+} \mathrm{reg}_{z_0 \rightarrow t_j},\quad j = 1,\dots,n.
    \end{equation}
    This gives
    \begin{subequations}
    \begin{align}
            \label{eq:pf_lhs_tau_alg}\theta_{j}h_{j}'(t_{j}) &= 4\Delta_{\theta_j}\langle \bar{\mathsf{V}}_n \rangle^2,\\
            \label{eq:pf_lhs_tau_evol}\theta_{j}h_{j}''(t_{j}) &= 2\partial_{t_j}\left(\langle \bar{\mathsf{V}}_n \rangle^2\right),\\
            \label{eq:pf_lhs_tau_Hirota}\theta_{j}h_{j}'''(t_{j}) &= 12 \left(\langle\bar{\mathsf{V}}_n\rangle\partial_{t_j}^{2}\langle\bar{\mathsf{V}}_n\rangle - \left(\partial_{t_j}\langle\bar{\mathsf{V}}_n\rangle\right)^2\right).
        \end{align}
    \end{subequations}

    As well from the formula for $h_{j}(z_0)$ given by \eqref{eq:gj_formula} we obtain
    \begin{equation}\label{eq:bilin_pf_1}
           \theta_j h_j^{(3-d)}(t_j) = (3-d)!\sum_{i\neq j}\theta_j\theta_i(t_i{-}t_j)^{d}\langle\bar{\mathsf{V}}_n^{(i,j),\otimes 2}\rangle.
        \end{equation}

    Comparing Equation \eqref{eq:bilin_pf_1} for $d = 2, 1,0$ with Equations \eqref{eq:pf_lhs_tau_alg}, \eqref{eq:pf_lhs_tau_evol} and \eqref{eq:pf_lhs_tau_Hirota} respectively we obtain the cases of \eqref{eq:tau_identities}.
\end{proof}

In the simplest case \(n=4\), we have numerically checked the bilinear relations in Theorem~\ref{thm:tau_identities} in the first orders of \(t\) expansion. 



We believe that there are similar identities for tau functions when $k>2$. 
However, the analogue of the determinant formula~\eqref{eq:det_answ} for \(k>2\) is expected to be more complicated, because the singular parts of the corresponding OPEs are more complicated than those in Lemmas~\ref{lemma:IV_OPE} and~\ref{lemma:I_V_OPE_shifted}.

\subsubsection{Definition of $c=-2$ tau function in terms of the linear problem}
\label{ssec:taudef}

It is shown in work \cite{ILT14} that the tau function defined as a correlator of monodromy vertex operators in $k=1$ case is precisely isomonodromic tau function defined by

\begin{equation}\label{eq:tau_c1}
\partial_{t_j}\log(\tau) = \frac{1}{2}\oint_{C_{t_j}}\frac{\mathrm{d}z}{2\pi\ri}\mathrm{Tr}(A^{2}(z)),
\end{equation}
where $A(z)$ is a flat connection corresponding to the local system defining and $C_{t_j}$ is a small contour around the point $t_j$.

In this section we show that similar relation exists in $k=2$ case for the function $\uptau$. More precisely in this section we relate the connection $A(z,z_0) = \partial_z\Phi(z,z_0)\Phi^{-1}(z,z_0)$ corresponding to the solution of the modified RH problem \eqref{eq:gen_RH_sol_cft} to the $c=-2$ tau function $\uptau$ defined by \eqref{eq:tau_def_gener}. 

Note that in the $c=1$ case the connection corresponding to the solution of RH problem given by \ref{thm:cm2_linear_solve} does not depend on $z_0$ so it worths nothing to obtain
\begin{equation}\label{eq:tau_c1_add}
\partial_{t_j}\log(\tau) = \frac{1}{2}\oint_{C_{t_j}}\frac{\mathrm{d}z_0}{2\pi\ri}\oint_{C_{t_j}}\frac{\mathrm{d}z}{2\pi\ri}\frac{\mathrm{Tr}(A^{2}(z))}{z-z_0}.
\end{equation}

In the $k=2$ case the connection $A(z,z_0)$ essentially depends on $z_0$ therefore the relation \eqref{eq:tau_c1} cannot be generalized to the case of $\uptau$. However the relation \eqref{eq:tau_c1_add} can be generalized as follows. 

\begin{prop}\label{prop:tau_isom_def} We have
\begin{equation}
\frac{1+4\Delta_{\theta_j} - 8\Delta_{\theta_j}^2}{4\Delta_{\theta_j}^2}\partial_{t_j}\log\uptau = \frac{1}{2}\oint_{C_{t_j}}\frac{\mathrm{d} z_{0}}{2\pi \ri}\oint_{C_{t_j}}\frac{\mathrm{d} z}{2\pi \ri}\frac{\mathrm{Tr}(A^{2}(z, z_0))}{z - z_0}.
\end{equation}
\end{prop}
Recall that $\Delta_{\theta_j} = \frac{\theta_j(\theta_j+1)}{2}$.
\begin{proof}
    The proof is based on OPE and is highly technical. We sketch it in Appendix \ref{app:proofs}.
\end{proof}

\subsubsection{Generalized Wick theorem for symplectic fermions}\label{ssec:SF_Wick}
Equations \eqref{eqs:I_V_OPE_shifted} can be thought of as an analogue of the commutativity \eqref{eq:c1I_comm_loc} which plays the key role in free-fermionic setting (as we briefly recall in Section \ref{ssec:FF_Wick}). 
Here we present the generalized Wick theorem for symplectic fermions.

We use notation 
\begin{equation}
f_{\varsigma_1,\dots, \varsigma_m}(z_1,\dots,z_m) \coloneqq \langle \bar{\mathsf{V}}_nJ^{\varsigma_1}(z_1)\dots J^{\varsigma_m}(z_m)\rangle,
\end{equation}
in particular, $f_{\varnothing}= \langle \bar{\mathsf{V}}_n\rangle$. We also denote 
\begin{equation}
    f^{\otimes 2,(i_1,\dots,i_r)}_{\varsigma_1,\dots, \varsigma_{m_1}| \varrho_1,\dots, \varrho_{m_2}}(z_1,\dots,z_{m_1}|w_1,\dots, w_{m_2}) \coloneqq \langle \bar{\mathsf{V}}_n^{(i_1,\dots, i_r),\otimes 2}(J^{\varsigma_1}(z_1)\dots J^{\varsigma_{m_1}}(z_{m_1}))\otimes (J^{\varrho_1}(w_1)\dots J^{\varrho_{m_2}}(w_{m_2}))\rangle,
\end{equation}

\begin{prop}{(Wick theorem for symplectic fermions)}\label{prop:Wick SF}
The correlation functions with multiple insertions of symplectic fermions satisfy the following equations for any $\varsigma_1,\dots, \varsigma_m,\varrho \in \{\pm 1\}$.
\label{prop:Wick_SF}
    \begin{multline}\label{eq:Wick_SF}
        f_{\varsigma_1,\dots, \varsigma_m,\varrho}(z_1,\dots, z_m,w) = \frac{(-1)^{m+1}}{f_{\varnothing}}\biggl(\sum_{i=1}^{m}(-1)^{i}f_{\varsigma_1,\dots,\widehat{\varsigma}_{i},\dots, \varsigma_m}(z_1,\dots, \hat{z}_i,\dots, z_m)\partial _{z_i}((z_i - w)f_{\varsigma_i,\varrho}(z_i,w)) +\\
        + \sum_{j = 1}^{n}\theta_{j}(t_j-w)f_{\varsigma_1,\dots, \varsigma_m|\varrho}^{(j),\otimes 2}(z_1,\dots, z_m|w)\biggl)
    \end{multline}
\end{prop}

\begin{proof}
    Denote $J_1^{\varsigma}(z) \coloneqq J^{\varsigma}(z)\otimes 1,\; J_2^{\varsigma}(z) \coloneqq 1\otimes J^{\varsigma}(z)$. We have OPE
    \begin{equation}\label{eq:IJ_OPE}
    I(z)J_{1}^{\varsigma}(w) = \frac{J^\varsigma_2(w)}{(z-w)^2}+\frac{\partial J^{\varsigma}_2(w)}{(z-w)}+ \mathrm{reg},\qquad I(z)J_{2}^{\varsigma}(w) = -\frac{J^\varsigma_1(w)}{(z-w)^2}-\frac{\partial J^{\varsigma}_1(w)}{(z-w)}+ \mathrm{reg}.
    \end{equation}
    Consider the correlator
    \begin{equation}
    h(z) = \langle \bar{\mathsf{V}}_n^{\otimes2} I(z)J_{1}^{\varsigma_1}(z_1)\dots J_{1}^{\varsigma_m}(z_m)J_2^{\varrho}(w)\rangle.
    \end{equation}
    The function $h$ is rational in $z$ and is decaying at $\infty$ as $z^{-4}$ (cf. the proof of Prop. \ref{prop:det_answ}). Using OPE's \eqref{eq:IV_OPE}, \eqref{eq:IJ_OPE} we find
    \begin{multline}
    h(z) = -\sum_{j=1}^n \frac{\theta_j}{z-t_j}f_{\varsigma_1,\dots, \varsigma_m|\varrho}^{(j),\otimes 2}(z_1,\dots, z_m|w) +\\ + \sum_{i=1}^{m}(-1)^{i-1}f_{\varsigma_1,\dots,\widehat{\varsigma}_{i},\dots, \varsigma_m}(z_1,\dots, \hat{z}_i,\dots, z_m)\partial_{z_i}\left(\frac{f_{\varsigma_i,\varrho}(z_i,w)}{(z-z_i)}
    \right) + (-1)^{m+1}f_{\varnothing}\partial_{w}\left(\frac{f_{\varsigma_1,\dots, \varsigma_m,\varrho}(z_1,\dots, z_m,w)}{(z-w)}
    \right).
    \end{multline}

Computing the residue with respect to $z$ at $\infty$ we obtain
\begin{multline}\label{eq:wick_rule_hexp1}
    0 = -\sum_{j=1}^n \theta_{{j}}f_{\varsigma_1,\dots, \varsigma_m|\varrho}^{(j),\otimes 2}(z_1,\dots, z_m|w) + \sum_{i=1}^{m}(-1)^{i-1}f_{\varsigma_1,\dots,\widehat{\varsigma}_{i},\dots, \varsigma_m}(z_1,\dots, \hat{z}_i,\dots, z_m)\partial _{z_i}f_{\varsigma_i,\varrho}(z_i,w) +\\ + (-1)^{m+1}f_{\varnothing}\partial_wf_{\varsigma_1,\dots, \varsigma_m,\varrho}(z_1,\dots, z_m,w).
\end{multline}
Computing the next term of the expansion of $h(z)$ at $\infty$ we find
\begin{multline}\label{eq:wick_rule_hexp2}
    0 = -\sum_{j=1}^n t_j\theta_{{j}}f_{\varsigma_1,\dots, \varsigma_m|\varrho}^{(j),\otimes 2}(z_1,\dots, z_m|w) + \sum_{i=1}^{m}(-1)^{i-1}f_{\varsigma_1,\dots,\widehat{\varsigma}_{i},\dots, \varsigma_m}(z_1,\dots, \hat{z}_i,\dots, z_m)\partial _{z_i}(z_if_{\varsigma_i,\varrho}(z_i,w)) +\\ + (-1)^{m+1}f_{\varnothing}\partial_w(w f_{\varsigma_1,\dots, \varsigma_m,\varrho}(z_1,\dots, z_m,w)).
\end{multline}

By linearly combining equations \eqref{eq:wick_rule_hexp1} and \eqref{eq:wick_rule_hexp2} we obtain
\begin{multline}
    \sum_{j=1}^n (w-t_j)\theta_{{j}}f_{\varsigma_1,\dots, \varsigma_m|\varrho}^{(j),\otimes 2}(z_1,\dots, z_m|w) + \sum_{i=1}^{m}(-1)^{i-1}f_{\varsigma_1,\dots,\widehat{\varsigma}_{i},\dots, \varsigma_m}(z_1,\dots, \hat{z}_i,\dots, z_m)\partial _{z_i}((z_i-w)f_{\varsigma_i,\varrho}(z_i,w)) +\\ + (-1)^{m+1}f_{\varnothing} f_{\varsigma_1,\dots, \varsigma_m,\varrho}(z_1,\dots, z_m,w) =0
\end{multline}
\end{proof}

\begin{remark}
    Wick rule (Prop. \ref{prop:Wick_SF}) together with fusion (Equations \eqref{fusion_rules_per}) allows to reconstruct the correlation functions of $\bar{\mathsf{V}}_n$ with multiple insertions of symplectic fermions by the correlation functions with at most two insertions. This implies that the correlators with three vertex operators and at most two insertions of symplectic fermions determine the vertex operators uniquely. These correlators can be obtained from three-point modified RH problem. In Prop. \ref{prop:gRH problem_unique} we have shown that solution of such a problem is unique and gave it explicitly.
\end{remark}

\appendix
\section{Expansions for $c = -2$}
\label{app:OPEs}
Here we collect expansions of the actions of various combinations of currents $J^{\pm}(z)$ on a highest weight vector in the sum of Verma modules $\bigoplus_{n\in\mathbb{Z}} \mathbb{M}_{\theta+n} = \mathbb{M}_{[\theta]_2}\oplus \mathbb{M}_{[\theta+1]_{2}}$.

By explicit calculation we get    \begin{equation}\label{eq:I_V2sh_pf_J_vacL}
        \begin{aligned}
            J^{+}(z)|\theta\rangle &= z^{\theta}\left(1 + \frac{1}{\theta + 1}L_{-1}z + \left(\frac{L_{-1}^2}{\theta + 1} + L_{-2}\right)\frac{z^2}{2(\theta + 2)} + O(z^3)\right)|\theta {+} 1\rangle,\\
            J^{-}(z)|\theta\rangle &= z^{-\theta - 1}\left(-\theta + L_{-1}z + \left(\frac{L_{-1}^2}{-\theta + 1} + L_{-2}\right)\frac{z^2}{2} + O(z^3)\right)|\theta {-} 1\rangle.
        \end{aligned}
    \end{equation}

\begin{equation}\label{eq:expand_JJ}
    J^{+}(z)|\theta\rangle {\wedge} J^{-}(z)|\theta\rangle = -\frac{\theta}{z}|\theta {+} 1\rangle {\wedge} |\theta - 1\rangle {-} \left(L_{-1}|\theta {+} 1\rangle  {\wedge} \frac{\theta}{\theta+1}|\theta {-} 1\rangle-  |\theta {+}1\rangle{\wedge}L_{-1}|\theta {-} 1\rangle\right) {+} O(z).
\end{equation}

\begin{multline}\label{eq:expand_JJder}
    (J^{+})^{\prime}(z)|\theta\rangle {\wedge} (J^{-})^{\prime}(z)|\theta\rangle = \frac{\theta^2(\theta {+} 1)}{z^3}|\theta {+} 1\rangle {\wedge} |\theta {-} 1\rangle {-} \\
    {-}\frac{\theta}{z^2}\left({}(\theta{+}1)|\theta {-} 1\rangle {\wedge} L_{{-}1}|\theta {+} 1\rangle {-} \theta L_{{-}1}|\theta {-} 1\rangle {\wedge}|\theta {+} 1\rangle\right) {+} O(z^{-1}).
\end{multline}

We also have
\begin{equation}\label{eq:I_V2sh_pf_J_vac_am}
\begin{aligned}
    (\theta+1)^{-1}J^{+}(z)\wedge J^{-}(z)(|\theta{-}1\rangle\wedge L_{-1}|\theta{+}1\rangle) &=
    -(\theta+2)z^{-4}\left(2 +z\theta^{-1}\big(L_{-1}\otimes 1 + 1\otimes L_{-1}\big) + O(z^2)\right)|\theta\rangle^{\otimes 2},\\
    (\theta-1)^{-1}J^{+}(z)\wedge J^{-}(z)(L_{-1}|\theta{-}1\rangle\wedge |\theta{+}1\rangle) &=z^{-4}\left(2(-\theta-1) + z(L_{-1}\otimes 1 + 1\otimes L_{-1}) + O(z^2)\right)|\theta\rangle^{\otimes 2}.
\end{aligned}
\end{equation}

\section{Technical proofs}\label{app:proofs}

\subsection{Proof of Lemma \ref{lemma:I_V_OPE_shifted}}

    Using Prop. \ref{prop:fusion} we obtain
    \begin{multline}\label{eq:I_V2sh_pf_1}        I(z)\left(V^{\theta + 1}(t)\otimes V^{\theta-1}(t)\right) =\\=\sum_{\varsigma_1, \varsigma_2, \varsigma_{3}, \varsigma_{4}}\epsilon_{\varsigma_1, \varsigma_2}\tilde{F}^{[23]}_{\varsigma_1, \varsigma_{3}}\tilde{F}^{[23]}_{\varsigma_2, \varsigma_{4}}V^{\theta+1+\varsigma_3}[J^{\varsigma_{3}}(z-t)|\theta + 1\rangle](t) \otimes V^{\theta-1+\varsigma_4}[J^{\varsigma_{4}}(z-t)|\theta - 1\rangle](t) =\\ = -V^{\theta}[J^{-}(z-t)|\theta + 1\rangle](t) \otimes V^{\theta}[J^{-}(z-t)|\theta + 1\rangle](t) +\\+ V^{\theta+2}[J^{+}(z-t)|\theta + 1\rangle](t) \otimes V^{\theta-2}[J^{-}(z-t)|\theta - 1\rangle](t).    \end{multline}

    Applying expansions \eqref{eq:I_V2sh_pf_J_vacL} 
    we get
    \begin{subequations}
        \begin{equation}
            V^{\theta+2}[J^{+}(z-t)|\theta + 1\rangle](t) \otimes V^{\theta-2}[J^{-}(z-t)|\theta - 1\rangle](t) = O(z-t) = \mathrm{reg}.
        \end{equation}
        \begin{multline}\label{eq:I_V2sh_pf_2}
            V^{\theta}[J^{-}(z-t)|\theta + 1\rangle](t) \otimes V^{\theta}[J^{+}(z-t)|\theta - 1\rangle](t) =\\= \frac{1}{(z-t)^3}V^{\theta}[\left(-\theta-1 + L_{-1}(z{-}t) + \left(\frac{L_{-1}^2}{-\theta} + L_{-2}\right)\frac{(z{-}t)^2}{2} + O((z{-}t)^3)\right)|\theta \rangle](t) \otimes \\ \otimes V^{\theta}[\left(1 + \frac{1}{\theta}L_{-1}(z{-}t) + \left(\frac{L_{-1}^2}{\theta} + L_{-2}\right)\frac{(z{-}t)^2}{2(\theta + 1)} + O((z{-}t)^3)\right)|\theta\rangle](t)
        \end{multline}
    \end{subequations}

    Then we obtain
\begin{multline}\label{eq:Prop_OPE_IskewV_proof1}
        I(z)(V^{\theta + 1}(t)\otimes V^{\theta-1}(t)-V^{\theta - 1}(t)\otimes V^{\theta+1}(t)) =\\= -V^{\theta}[J^{-}(z-t)|\theta + 1\rangle](t) \otimes V^{\theta}[J^{+}(z-t)|\theta - 1\rangle](t)  - V^{\theta}[J^{+}(z-t)|\theta - 1\rangle](t)\otimes V^{\theta}[J^{-}(z-t)|\theta + 1\rangle](t) + \mathrm{reg} = 
        \\
        = -\frac{1}{(z-t)^3}V^{\theta}[\left(-\theta-1 + L_{-1}(z{-}t) + \left(\frac{L_{-1}^2}{-\theta} + L_{-2}\right)\frac{(z{-}t)^2}{2} + O((z{-}t)^3)\right)|\theta\rangle](t) \otimes 
        \\ \otimes V^{\theta}[\left(1 + \frac{1}{\theta}L_{-1}(z{-}t) + \left(\frac{L_{-1}^2}{\theta} + L_{-2}\right)\frac{(z{-}t)^2}{2(\theta + 1)} + O((z{-}t)^3)\right)|\theta\rangle](t) - 
        \\
        - \frac{1}{(z-t)^3}V^{\theta}[\left(1 + \frac{1}{\theta}L_{-1}(z{-}t) + \left(\frac{L_{-1}^2}{\theta} + L_{-2}\right)\frac{(z{-}t)^2}{2(\theta + 1)} + O((z{-}t)^3)\right)|\theta\rangle](t) \otimes 
        \\
        \otimes V^{\theta}[\left(-\theta-1 + L_{-1}(z{-}t) + \left(\frac{L_{-1}^2}{-\theta} + L_{-2}\right)\frac{(z{-}t)^2}{2} + O((z{-}t)^3)\right)|\theta\rangle](t) +\mathrm{reg}.
    \end{multline}
    Computing the singular part of the expansion of \eqref{eq:Prop_OPE_IskewV_proof1} at $z = t$ we obtain
    \begin{multline}
        I(z)(V^{\theta + 1}(t)\otimes V^{\theta-1}(t)-V^{\theta - 1}(t)\otimes V^{\theta+1}(t)) = \\
        = \frac{1}{\theta}\left(\frac{4\Delta_\theta}{(z-t)^3} + \frac{1}{(z-t)^2}(\partial_t\otimes 1 + 1\otimes \partial_t) + \frac{1}{(z-t)}(\partial_t \otimes 1-1\otimes\partial_t)^2\right)\left(V^{\theta}(t)\otimes V^{\theta}(t)\right) + \mathrm{reg}.
    \end{multline}

\subsection{Proof of Prop. \ref{prop:tau_isom_def}}
    In this proof we use the notions and results introduced in Section \ref{sec:fermions}.  
    
    We have
    \begin{equation}\label{eq:TrA2thrDet}
    \frac{1}{2}\mathrm{Tr}\left(A^{2}(z)\right) = \frac{1}{2}\left(\mathrm{Tr}\left(A(z)\right)\right)^2 - \det\left(A(z)\right) = \frac{1}{2}\left(\frac{\det^{\prime}\left(\Phi(z)\right)}{\det\left(\Phi(z)\right)}\right)^2 - \frac{\det\left(\Phi^{\prime}(z)\right)}{\det\left(\Phi(z)\right)}.
\end{equation}

Let $M_{\varsigma_1,\varsigma_2} = \langle \bar{\mathsf{V}}_n J^{\varsigma_2}(z)J^{\varsigma_1}(z_0)\rangle$. From equations \eqref{eq:TrA2thrDet} and  \eqref{eq:gen_RH_sol_cft} we obtain
\begin{equation}
 \frac{1}{2}\mathrm{Tr}\left(A^{2}(z)\right) = \frac{1}{2}\log'(\det(M))^2 + \frac{6}{z-z_0}\log'(\det(M)) + \frac{12}{(z-z_0)^2} -\frac{\det(M')}{\det(M)}.
\end{equation}

Using the expansions \eqref{eq:expand_JJ} and OPE we obtain
\begin{equation}\label{eq:detM_exp1}
    2\det(M) = {-}\frac{\theta_{j}}{(z{-}t_j)}\langle\bar{\mathsf{V}}_{n}^{(j),\otimes 2}I(z_0)\rangle {+} \theta_j \langle\tilde{\bar{\mathsf{V}}}_{n}^{(j),\otimes 2}I(z_0)\rangle {+} O((z{-}t_j)).
\end{equation}

Here by $\tilde{\bar{\mathsf{V}}}_{n}^{(j),\otimes 2}$ we denote the standard sum of compositions of periodic vertex operators, where instead of $\bar{V}_{[\sigma_{j}+l_{j}]_2,[\sigma_{j-1}+l_{j-1}]_2}^{\theta_{j}}(t_{j})\otimes \bar{V}_{[\sigma_{j}+l_{j}^{\prime}]_2,[\sigma_{j-1}+l_{j-1}^{\prime}]_2}^{\theta_{j}}(t_{j})$ we insert
\begin{multline}
    \mathcal{D}\bar{V}_{[\sigma]_{2}, l,l^{\prime},[\tilde{\sigma}]_{2}, \tilde{l},\tilde{l}^{\prime}}^{\theta, \wedge 2}(t_{j}) \coloneqq\mathcal{D}_{1}^{(j)}(\bar{V}_{[\sigma_{j}+l_{j}{+}1]_2,[\sigma_{j-1}{+}l_{j-1}{+}1]_2}^{\theta_{j}{+}1}(t_{j})\otimes \bar{V}_{[\sigma_{j}+l_{j}^{\prime}{+}1]_2,[\sigma_{j-1}+l_{j-1}^{\prime}{+}1]_2}^{\theta_{j}{-}1}(t_{j})) 
    -\\-
    \mathcal{D}_{2}^{(j)}(\bar{V}_{[\sigma_{j}+l_{j}{+}1]_2,[\sigma_{j-1}{+}l_{j-1}{+}1]_2}^{\theta_{j}{-}1}(t_{j})\otimes \bar{V}_{[\sigma_{j}+l_{j}^{\prime}{+}1]_2,[\sigma_{j-1}+l_{j-1}^{\prime}{+}1]_2}^{\theta_{j}{+}1}(t_{j})),
\end{multline}
where
\begin{equation}
    \mathcal{D}_1^{(j)} = \frac{1}{\theta_{j}}1\otimes \partial_{t_j} - \frac{1}{\theta_j{+}1}\partial_{t_j} \otimes 1,\;\; \mathcal{D}_2^{(j)} = \frac{1}{\theta_{j}} \partial_{t_j} \otimes 1- \frac{1}{\theta_j{+}1}  1\otimes \partial_{t_j}.
\end{equation}

Further using the expansions  \eqref{eq:expand_JJder} and OPE we find
\begin{multline}\label{eq:detMp_exp1}
    2\det(M') = \frac{2\Delta_{\theta_{j}}\theta_{j}}{(z{-}t_j)^3}\langle\bar{\mathsf{V}}_n^{(j),\otimes 2}I(z_0)\rangle {+}\\
    {+} \frac{2\Delta_{\theta_{j}}\theta_{j}}{(z{-}t_j)^2}\left(\frac{1}{2\Delta_{\theta_{j}}}(\partial_{t_{j}}\otimes 1 + 1\otimes \partial_{t_{j}})\big(\langle\bar{\mathsf{V}}_n^{(j),\otimes 2}I(z_0)\rangle\big) - \langle\tilde{\bar{\mathsf{V}}_n}^{(j),\otimes 2}I(z_0)\rangle\right) + O\big((z-t_j)^{-1}\big).
\end{multline}

From equations \eqref{eq:detM_exp1} and \eqref{eq:detMp_exp1} we obtain
\begin{equation}
   \frac{\det(M')}{\det(M)} =  -\frac{2\Delta_{\theta_j}}{(z-t_j)^2} - \frac{1}{z-t_j}\partial_{t_j}(\log(\langle\bar{\mathsf{V}}_n^{(j),\otimes 2}I(z_0)\rangle)) + O(1).
\end{equation}
\begin{equation}
    \log'(\det(M))^2 = \frac{1}{(z-t_j)^2} +\frac{2\langle\tilde{\bar{\mathsf{V}}}_{n}^{(j),\otimes 2}I(z_0)\rangle}{\langle\bar{\mathsf{V}}_n^{(j),\otimes 2}I(z_0)\rangle(z-t_j)} +  O(1).
\end{equation}

\begin{equation}
    \frac{1}{2}\oint_{C_{t_j}}\frac{\mathrm{d} z}{2\pi \ri}\frac{\mathrm{Tr}(A^{2}(z, z_0))}{z - z_0}= -\frac{\frac{13}{2}+2\Delta_{\theta_j}}{(z_0-t_j)^2} - \frac{\partial_{t_j}\langle\bar{\mathsf{V}}_n^{(j),\otimes 2}I(z_0)\rangle + \langle\tilde{\bar{\mathsf{V}}}_{n}^{(j),\otimes 2}I(z_0)\rangle}{2(z_0-t_j)\langle\bar{\mathsf{V}}_n^{(j),\otimes 2}I(z_0)\rangle}.
 \end{equation}

Lemma \ref{lemma:I_V_OPE_shifted} implies
\begin{equation}
    \langle\bar{\mathsf{V}}_n^{(j),\otimes 2}I(z_0)\rangle = \frac{1}{\theta_{j}(z_0-t_j)^3}\left(4\Delta_{\theta_j}\uptau^2 + 2(z_0-t_j)\uptau\partial_{t_j}(\uptau) + O((z_0-t_j)^2)\right).
\end{equation}
We use OPE and the expansion \eqref{eq:I_V2sh_pf_J_vac_am} to obtain
\begin{multline}
    I(z_0)\mathcal{D}\bar{V}_{[\sigma]_{2}, l,l^{\prime},[\tilde{\sigma}]_{2}, \tilde{l},\tilde{l}^{\prime}}^{\theta, \wedge 2}(t_j) =\\
    =\left(\frac{2(1-4\Delta_{\theta_j})}{\theta_j(z-t_j)^4} -\frac{3}{\theta_{j}(z-t_j)^3}(\partial_{t_{j}}\otimes 1 + 1\otimes \partial_{t_{j}}) \right)(V_{l_j}^{\theta_{j}}(t_j)\otimes V_{\tilde{l}_j}^{\theta_{j}}(t_j))+ O((z-t_j)^{-2}),
\end{multline}
which implies
\begin{equation}
    \langle\tilde{\bar{\mathsf{V}}}_{n}^{(j),\otimes 2}I(z_0)\rangle = \frac{2}{\theta_j(z_0-t_j)^4}\left((1-4\Delta_{\theta_j})\uptau^2 - 3\uptau\partial_{t_j}(\uptau)(z_0-t_j) + O((z_0-t_j)^2)\right).
\end{equation}
    
We obtain

\begin{equation}
    \frac{1}{2}\oint_{C_{t_j}}\frac{\mathrm{d} z}{2\pi \ri}\frac{\mathrm{Tr}(A^{2}(z, z_0))}{z - z_0} = -\frac{\frac{15}{2}+\frac{1}{2\Delta_{\theta_j}}+2\Delta_{\theta_j}}{(z_0-t_j)^2}+ \frac{\left(1+4\Delta_{\theta_j} - 8\Delta_{\theta_j}^2\right)\partial_{t_j}(\log(\uptau))}{4\Delta_{\theta_j}^2(z_0-t_j)}+O(1).
\end{equation}

\printbibliography[
heading=bibintoc,
title={References}
]

\printaddresses

\end{document}